\documentclass{article}
\usepackage{color, xcolor, colortbl}
\usepackage{graphicx,epstopdf}
\usepackage{geometry}
\usepackage{amsmath,amssymb,amsthm}
\usepackage{algorithm}
\usepackage{algorithmic}
\usepackage{bm}
\usepackage[caption=false]{subfig}
\usepackage{appendix}
\usepackage{multirow}
\usepackage{braket}
\usepackage[english]{babel}
\usepackage[qm]{qcircuit}
\usepackage{adjustbox}
\usepackage{hyperref}
\usepackage{comment}
\usepackage[capitalize]{cleveref}
\usepackage[square,numbers,sort]{natbib}

\usepackage{cellspace} 
\usepackage{makecell} 
\setcellgapes{3pt}
\usepackage[roman]{complexity}

\newcommand{\bvec}[1]{\mathbf{#1}}

\newcommand{\vr}{\bvec{r}}

\newcommand{\vx}{\bvec{x}}
\newcommand{\vy}{\bvec{y}}

\newcommand{\vG}{\bvec{G}}

\renewcommand{\Re}{\operatorname{Re}}
\renewcommand{\Im}{\operatorname{Im}}
\newcommand{\Tr}{\operatorname{Tr}}

\newcommand{\I}{\mathrm{i}}

\newcommand{\mc}[1]{\mathcal{#1}}

\newcommand{\wt}[1]{\widetilde{#1}}

\newcommand{\abs}[1]{\left\lvert#1\right\rvert}
\newcommand{\norm}[1]{\left\lVert#1\right\rVert}

\newcommand{\ud}{\,\mathrm{d}}
\newcommand{\Or}{\mathcal{O}}
\newcommand{\NN}{\mathbb{N}}
\newcommand{\RR}{\mathbb{R}}
\newcommand{\CC}{\mathbb{C}}
\newcommand{\ZZ}{\mathbb{Z}}
\newcommand{\GG}{\mathbb{G}}

\newtheorem{thm}{\protect\theoremname}
\theoremstyle{plain}
\newtheorem{lem}[thm]{\protect\lemmaname}
\theoremstyle{plain}
\newtheorem{rem}[thm]{\protect\remarkname}
\theoremstyle{plain}
\newtheorem*{lem*}{\protect\lemmaname}
\theoremstyle{plain}
\newtheorem{prop}[thm]{\protect\propositionname}
\theoremstyle{plain}
\newtheorem{cor}[thm]{\protect\corollaryname}

\newtheorem{defn}[thm]{Definition}

\providecommand{\corollaryname}{Corollary}
\providecommand{\lemmaname}{Lemma}
\providecommand{\propositionname}{Proposition}
\providecommand{\remarkname}{Remark}
\providecommand{\theoremname}{Theorem}

\begin{document}

\title{Fast inversion, preconditioned quantum linear system solvers, {fast Green's function computation}, and fast evaluation of matrix functions}
\author{Yu Tong\thanks{Department of Mathematics, University of California, Berkeley,
CA 94720. }
\and
Dong An\thanks{Department of Mathematics, University of California, Berkeley,
CA 94720. }
\and
Nathan Wiebe\thanks{Department of Computer Science, University of Toronto, Toronto On, Department of Physics M5S 1A1, University of Washington, Seattle WA 98195 and High Performance Computing Division, Pacific Northwest National Laboratory, Richland WA 99354}
\and
Lin Lin\thanks{Department of Mathematics and Challenge Institute for Quantum Computation, University of California, Berkeley, and Computational Research Division, Lawrence Berkeley National Laboratory, Berkeley, CA 94720. }
}
\date{\today}
\maketitle
\begin{abstract}
Preconditioning is the most widely used and effective way for treating ill-conditioned linear systems in the context of classical iterative linear system solvers. We introduce a quantum primitive called fast inversion, which can be used as a preconditioner for solving quantum linear systems. The key idea of fast  inversion is to directly block-encode a matrix inverse through a quantum circuit implementing the inversion of eigenvalues via classical arithmetics. We demonstrate the application of preconditioned linear system solvers for computing single-particle Green's functions of quantum many-body systems, which are widely used in quantum physics, chemistry, and materials science. We analyze the complexities in three scenarios: the Hubbard model, the quantum many-body Hamiltonian in the planewave-dual basis, and the Schwinger model. {We also provide a method for performing Green's function calculation in second quantization within a fixed particle manifold and note that this approach may be valuable for simulation more broadly.} Besides solving linear systems, fast inversion also allows us to develop fast algorithms for computing matrix functions, such as the efficient preparation of Gibbs states. We introduce two efficient approaches for such a task, based on the contour integral formulation and the inverse transform respectively.  
\end{abstract}

\section{Introduction}

Linear systems  appear ubiquitously in scientific and engineering computations. Accelerated solution of linear systems on quantum computers, or the quantum linear system problem (QLSP) has received a significant amount of attention in the past decade \cite{HarrowHassidimLloyd2009,ChildsKothariSomma2017,GilyenSuLowEtAl2019,SubasiSommaOrsucci2019,AnLin2019,ChakrabortyGilyenJeffery2018,WossnigZhaoPrakash2018,CaoPapageorgiouPetrasEtAl2013,XuSunEndoEtAl2019,Bravo-PrietoLaRoseCerezoEtAl2019,Ambainis2012,LinTong2020}. Solving QLSP means finding a solution vector $\ket{x}$ stored as a quantum state  (up to a normalization constant), so that $\ket{x}=A^{-1}\ket{b}/\norm{A^{-1}\ket{b}}$. The main advantage provided by a quantum computer is that the number of the qubits needed to store the matrix and the solution vector only scales logarithmically with respect to the matrix dimension, thus overcoming the curse of dimensionality on a classical computer. On the other hand, the cost of a quantum algorithm for solving a generic QLSP scales at least as $\Omega(\kappa(A))$ \cite{HarrowHassidimLloyd2009}, where $\kappa(A):=\norm{A}\norm{A^{-1}}$ is the  condition number of $A$.\footnote{Throughout the paper $\norm{A}\equiv \norm{A}_2$ is the 2-norm (or the operator norm) of an operator $A$, and $\norm{u}\equiv \norm{u}_2$ is the 2-norm of a vector $u$.} This can be expensive if the linear system is ill-conditioned. It is therefore of great interest if we can exploit certain special structures of QLSPs to reduce the cost.

For a classical iterative algorithm such as the conjugate gradient (CG) method, the most effective way to accelerate the solution of ill-conditioned linear systems is to find a preconditioner $M$ so that 1) $\kappa(MA)\ll \kappa(A)$, 2) the matrix-vector multiplication $M\ket{\psi}$ is easily accessible, and in particular its cost is independent of $\kappa(M)$ \cite{Saad2003}. On a classical computer, the condition 2) can be satisfied for instance, if $M$ is a diagonal matrix or can be easily diagonalized, or if $M$ is obtained by a sparse direct method such as the incomplete Cholesky factorization \cite{meijerink1977iterative,manteuffel1980incomplete,lin1999incomplete}. Then the cost for solving the transformed equation $MA\ket{x}=M\ket{b}$ is determined by $\kappa(MA)$ instead of $\kappa(A)$. The same strategy can be used to reduce the complexity of a quantum linear solver \cite{clader2013preconditioned}. 

In this paper, we focus on a QLSP of the form 
\begin{equation}
\ket{x}=(A+B)^{-1} \ket{b}/\norm{(A+B)^{-1} \ket{b}},
\label{eqn:AplusB}
\end{equation}
where $A,B\in\CC^{N\times N}$, and $N=2^n$.  Throughout the paper, unless stated otherwise, we assume $\norm{A}$ can be very large, while $\norm{B}$, $\|A^{-1}\|$, $\|(A+B)^{-1}\|=\Or(1)$. Therefore the condition numbers $\kappa(A),\kappa(A+B) = \Theta(\norm{A})$. Such a scenario occurs frequently in scientific computing e.g. if $A$ is obtained by discretizing an unbounded operator, such as the Laplace operator in a confined domain (see \cref{sec:fastinv_diagaonlize}), or if $A$ represents a term in a quantum many-body Hamiltonian that is significantly larger than other terms (see \cref{sec:example_manybody}). We would like to obtain a quantum linear system solver, of which the cost is independent of $\norm{A}$. In particular, we will illustrate  in detail the computation of single-particle Green's functions of a quantum many-body system. {It is worth pointing out that although much efforts have been made to efficiently solve the QLSP, suitable applications of QLSP solvers remain scarce, as it is often difficult to efficiently load classical data into the quantum circuit, and to output useful information using a limited number of measurements. We demonstrate that the problem of computing Green's functions do not not suffer from such problems, and can be a suitable end-to-end application of QLSP solvers.}

A closely related, and more general problem is to evaluate  matrix functions of the form $f(A+B)$, where $f(\cdot)$ is a smooth function defined on the spectrum of $H=A+B$ \cite{Higham2008}. Here for simplicity we assume $A,B$ are Hermitian matrices, so the spectrum of $H$ is on the real line. Under similar assumptions above, we would like to obtain quantum algorithms of which the cost is independent of $\norm{A}$. This obviously depends on the choice of the function $f$. We will focus, for concreteness, on the function $f(x)=e^{- x}$.  This choice is motivated in part by the importance of preparing Gibbs states in quantum simulation~\cite{poulin2009preparing}, machine learning~\cite{kieferova2017tomography} and quantum algorithms for semi-definite programming~\cite{brandao2017quantum}.

\subsection{Overview of the results and related works}

\vspace{1em}
\noindent\textit{Quantum linear system solver:}

{Starting from the ground-breaking work of \cite{HarrowHassidimLloyd2009}, in the past decade, several quantum algorithms have been developed to improve the performance of generic QLSP solvers~\cite{Ambainis2012,ChildsKothariSomma2017,SubasiSommaOrsucci2019,GilyenSuLowEtAl2019,AnLin2019,LinTong2020}. We review these methods in Appendix~\ref{sec:QLS_solvers}. }

Recently quantum-inspired classical algorithms based on certain $\ell^2$-norm sampling assumptions \cite{Tang2018quantum, Tang2019quantum} have been developed that are only up to polynomially slower than quantum linear system solvers. However, it is unclear whether the classical $\ell^2$-norm sampling can be achieved efficiently without access to a quantum computer in the setting of this work. The quantum-inspired classical algorithms also suffer from many practical issues making their application limited to highly specialized problems \cite{ArrazolaEtAl2019}. Most importantly, the assumption of low-rankness is crucial in these algorithms. The methods presented in this work assumes a block-encoding model, which could be used to efficiently represent low-rank as well as full-rank matrices on a quantum computer.

\vspace{1em}
\noindent\textit{Preconditioned quantum linear system solver:}

Solving the quantum linear system problem using preconditioning to deal with large condition number has been discussed in works such as \cite{clader2013preconditioned,ShaoXiang2018}. The idea of \cite{clader2013preconditioned} is to use the sparse approximate inverse (SPAI) preconditioner \cite{BenziMeyerTuma1996,BenziTuma1998}, which uses a $d$-sparse matrix as a preconditioner. The SPAI, denoted by $M$, can be constructed by solving a least-squares procedure for each row of the matrix $A$. However, for many problems, an efficient preconditioner in the form of SPAI may not exist, and the work of \cite{clader2013preconditioned} did not provide an efficient quantum implementation to construct SPAI nor its performance analysis. In \cite{ShaoXiang2018}, the preconditioner $M$ is taken to be a circulant matrix that can be efficient diagonalized using a quantum Fourier transform (QFT). However, the complexity of the algorithm in \cite{ShaoXiang2018} can depend on $\kappa(M)$, which should not be expected to be smaller than $\kappa(A)$. Furthermore, neither work provides an upper bound for $\kappa(MA)$, which is a key quantity determining e.g. the circuit depth. 

In this paper, we propose a different mechanism for constructing efficient preconditioners, called \textit{fast inversion}. The inspiration of fast inversion is that any invertible, $1$-sparse matrix $A$ can be efficiently implemented on a quantum computer via classical arithmetics. In particular, the cost of {constructing a block-encoding of $A^{-1}$} is independent of $\kappa(A)$. Note that fast inversion does not violate the complexity lower bound for solving QLSP, which is a statement of the efficiency of QLSP solvers \cite{HarrowHassidimLloyd2009} applied to general matrices. This is in parallel to the fast forwarding process of $1$-sparse matrices in Hamiltonian simulation \cite{childs2003exponential,ahokas2004improved}, which does not violate the theorem of ``no-fast-forwarding'' \cite{berry2007efficient}. Furthermore, if $A$ can be unitarily diagonalized, so that both the diagonalization procedure and the encoding of eigenvalues can be efficiently implemented on a quantum computer, then $A$ can be fast inverted.  
Fast inversion can be viewed as a quantum primitive for a wide range of tasks. For example, we describe an efficient implementation of inverting certain normal matrices, such as circulant matrices. 

{We introduce a parameter $\xi=\|A^{-1}\ket{b}\|$, and without loss of generality rescale $A$ so that $\|A^{-1}\|=1$. As will be analyzed in Section~\ref{sec:fastinv}, if we consider $\xi$ and $\kappa(A)$ as two independent parameters, the immediate benefit of the fast inversion is that unlike any other methods in the literature, the cost of solving the QLSP depends only on $\xi$. The value of $\xi$ is bounded from below by $1/\kappa(A)$. Therefore in the worst case when $\xi=1/\kappa(A)$, the cost for solving the QLSP still depends linearly on $\kappa(A)$.  However, we will demonstrate in \cref{rem:complexity_lower_bound_xi} that such dependence through $\xi$ already reaches the complexity lower bound, and cannot be improved by any QLSP solver. Furthermore, such a bound for $\xi$ is usually not tight for many examples of practical interest. This is demonstrated via a concrete example of using fast inversion to solve a translational invariant elliptic partial differential equation (PDE) in \cref{sec:fastinv_diagaonlize}. 
}

Now for the linear system \eqref{eqn:AplusB}, assuming $A$ can be fast inverted so that $M=A^{-1}$, the cost for solving the the preconditioned linear system of \cref{eqn:AplusB} depends on $\kappa(M(A+B))=\kappa(I+A^{-1}B)$, which can be bounded in terms of $\norm{B},\norm{A^{-1}},\norm{(A+B)^{-1}}$ (or more accurately, their block-encoding factors, see \cref{lem:cond_precond_linear_sys}). This is in contrast to other QLSP preconditioning techniques where no such bound is known.  We introduce a new parameter $\xi=\|(A+B)^{-1}\ket{b}\|$ where $\ket{b}$ is the normalized quantum state corresponding to the right-hand side. 
Using QSVT, we obtain a gate-based implementation of a preconditioned linear system solver, and its cost is independent of $\norm{A}$, but only depends on $\xi$, and several block-encoding subnormalization factors which are upper bounds of the norms $\norm{B}$, $\norm{A^{-1}}$, and $\norm{(A+B)^{-1}}$ (\cref{thm:precond_linear_sys} and \cref{cor:precond_linear_sys_solver}). {In the worst case  the query complexity of the preconditioned linear system solver can depend superlinearly on $\norm{B}$ (or the corresponding block-encoding factor $\alpha_B$). However in some cases the worst case estimate can be significantly improved and the scaling with respect to $\norm{B}$ can be linear. We discuss such implications in  \cref{rem:upperbound_kappaW} and \cref{rem:scale_alphaB}.}
Throughout the paper we adopt the worst case estimate of $\norm{(I+A^{-1}B)^{-1}}$, which is responsible for the superlinear scaling with respect to $\alpha_B$ in \cref{tab:compare_algs_hubbard} and \cref{tab:compare_algs_thermal} below.

\vspace{1em}
\noindent\textit{Computing single-particle Green's functions of quantum many-body systems:}

As an application of the preconditioned linear system solver, we consider the problem of computing the one-particle Green's function of a quantum many-body system, which is a standalone linear system problem in high dimensions. Most of the literature on quantum simulation thus far focus on estimating the ground state energy and preparing the ground state. However, once the ground state is found, one may further evaluate the single-particle Green's function, which carries important spectroscopic information and is widely used in quantum physics, chemistry and materials science \cite{NegeleOrland1988,MartinReiningCeperley2016}. Calculation of Green's functions are computationally challenging. Previous works \cite{BauerWeckerMillisEtAl2015,endo2020calculation} focused on evaluating Green's functions in the time domain via Hamiltonian simulation. Ref.~\cite{cai2020quantum} computes the response function, which is closely related to the Green's function, in the frequency domain.  Here we will provide a preconditioned linear system method for direct computation of the Green's function in the frequency domain. 

The setup of the problem is as follows. Suppose we are given a Hamiltonian $\hat{H}=\hat{A}+\hat{B}$, so that $z \hat{I}+\hat{A}$ can be fast inverted for some properly chosen $z\in \CC$. We assume we have an $(\alpha_H,m_H,0)$-block-encoding of $\hat{H}$, as well as the oracles introduced in Section~\ref{sec:preconditioned_greens_function_solver}. We also assume there is an oracle available to construct the ground state $\ket{\Psi_0}$ of the Hamiltonian to precision $\varsigma$ in terms of trace-norm distance with probability at least $p$, and we denote this oracle by $U_\Psi$. The goal is to compute the Green's function $G^{(+)}_{ij}(z)=\braket{\Psi_0|\hat{a}_i(z-\hat{H}+E_0)^{-1}\hat{a}_j^\dagger|\Psi_0}$ as defined in Section~\ref{sec:greens_function} (and a corresponding $G^{(-)}$) for $z$ satisfying $|\Im z|\geq \eta>0$. 

This task can be accomplished using either HHL (based on phase estimation), LCU, or QSVT, to construct a block-encoding of $(z-\hat{H}+E_0)^{-1}$. We then apply the non-unitary Hadamard test described in Appendix~\ref{sec:non_unitary_hadamard} to estimate the expectation value. The analysis in Appendix~\ref{sec:query_complexities_HHL_LCU_Greens} shows that the number of queries to $U_\Psi$ scales linearly with the block-encoding subnormalization factor of $(z-\hat{H}+E_0)^{-1}$, which is upper bounded by $1/\eta$, and the number of queries to the block-encoding of $\hat{H}$ scales linearly with the product of the above subnormalization factor and the number of queries used in the block-encoding of the matrix inverse, with the latter scaling linearly with the condition number. Here the condition number scales linearly with $|z|+\alpha_H\sim |z|+\|\hat{A}\|$. Using our preconditioning technique we can remove this dependence on $z+\|\hat{A}\|$. The detailed analysis can be found in \cref{sec:preconditioned_greens_function_solver}.
We also provide a few concrete examples such as the Hubbard model, the quantum many-body Hamiltonian in a planewave dual basis set, and the Schwinger model in \cref{sec:example_manybody}.

{In this application, the outputs are the matrix elements of the Green's function $G_{ij}(z)$, rather than a quantum state representing the solution to a QLSP. Therefore the complexity does not involve the parameter $\xi$ as in the setting of the preconditioned QLSP. As a consequence the speedup we discussed in the previous paragraph depends only on the structure of the Hamiltonian.}




\begin{table}[H]
\makegapedcells
  \begin{center}
      \begin{tabular}{p{2cm}|p{2cm}|p{3cm}|p{3cm}|p{2.7cm}}
        \hline
        \hline
        &  Algorithm &  Queries to $U_\Psi$ & Queries to block-encodings & Error \\
        \hline
        w.o. preconditioner & HHL & $\Or( \frac{1}{\eta\sqrt{p}\epsilon}\log(\frac{1}{\varsigma}) )$ & $\wt{\Or}( \frac{|z|+\alpha_H}{\eta^3\epsilon^2}  )$ & $\epsilon+\Or(\frac{\varsigma}{\eta})$ \\
        \cline{2-5}
        & LCU/QSVT & $\Or( \frac{1}{\eta\sqrt{p}\epsilon}\log(\frac{1}{\varsigma}) )$ & $\wt{\Or}( \frac{|z|+\alpha_H}{\eta^2\epsilon})$ & $\epsilon+\Or(\frac{\varsigma}{\eta})$  \\
        \hline
        w. preconditioner & This work  & 
        {$\Or( \frac{1}{\wt{\sigma}_{\min}\sqrt{p}\epsilon}\log(\frac{1}{\varsigma}) )$}
        & {$\wt{\Or}( \frac{\alpha_B}{\wt{\sigma}^2_{\min} \epsilon} )$} 
        & {$\epsilon+\Or(\frac{\varsigma}{\wt{\sigma}_{\min}})$} \\
        \hline
        \hline
      \end{tabular}
  \end{center}
  \caption{Comparison of the number of queries to the ground state and relevant block-encodings needed using different algorithms for computing
  an entry of the single-particle Green's function $G(z)$ of a quantum many-body Hamiltonian of the form $\hat{H}=\hat{A}+\hat{B}$, where $|z|\geq\eta$. The operators $\hat{A}$, $\hat{B}$, and $\hat{H}$ are given in their block-encodings with subnormalization factors $\alpha_A$, $\alpha_B$, and $\alpha_H$ respectively, satisfying $\alpha_H\sim \alpha_A\sim \norm{\hat{A}}\gg\alpha_B$. 
  {$\wt{\sigma}_{\min}$ is defined in Theorem~\ref{thm:greens_function} and $\wt{\sigma}_{\min}=\Omega(\eta/\alpha_B)$.}
  The error comes from three parts: preparing the ground state, block-encoding of the matrix inverse, and the Hadamard test with amplitude estimation. The latter two are controlled by $\epsilon$ while the first is controlled by $\varsigma$. For simplicity we assume the ground energy is known exactly here. The error as a result of inexact ground energy is included in Theorem~\ref{thm:greens_function}. }
  \label{tab:compare_algs_hubbard}
\end{table}

\vspace{1em}
\noindent\textit{Fast algorithm for evaluating matrix functions:}

The method of solving QLSP \eqref{eqn:AplusB} is a special case of computing matrix functions $f(A+B)\ket{b}$ with $f(x)=x^{-1}$. Here for simplicity we restrict $A,B$ to be Hermitian matrices. In parallel to solving QLSP, the evaluation of a general matrix function can also be performed using the phase estimation algorithm, similar to its use in the HHL algorithm \cite{HarrowHassidimLloyd2009}. Similarly LCU \cite{ChildsKothariSomma2017,BerryChildsKothari2015}, and QSP/QSVT \cite{GilyenSuLowEtAl2019,LowChuang2017} can also be used to evaluate matrix functions, and achieve better dependence on various parameters, especially the desired precision $\epsilon$.

\begin{table}[t!]
  \begin{center}
    \makegapedcells
      \begin{tabular}{p{2cm}|p{5.5cm}|p{4cm}}
        \hline\hline
                &  Algorithm &  Query complexities \\
        \hline
        w.o. preconditioner & Phase estimation \cite{poulin2009sampling}  & $\wt{\Or}(\frac{\alpha_H}{\xi\epsilon})$   \\
        \cline{2-3}
        & LCU \cite{vanApeldoorn2020quantum} & $\wt{\Or}(\frac{\alpha_{H}}{\xi}\log(\frac{1}{\epsilon}))$  \\
        \hline
        w. preconditioner & This work (contour integral) & 
        {$\wt{\Or}(\frac{\alpha_B}{\xi\wt{\sigma}'^2_{\min}}\log(\frac{1}{\epsilon}))$}\\
        \cline{2-3}
         & This work (inverse transformation) & {$\wt{\Or}\left(\frac{\alpha_B}{\xi\wt{\sigma}^2_{\min}}\left[\log\left(\frac{1}{\epsilon}\right)\right]^5\right)$}\\
        \hline
        \hline
      \end{tabular}
  \end{center}
  \caption{Comparison of the performance of different algorithms for preparing the state $e^{-H}\ket{b}/\xi$, where $\xi=\|e^{-H}\ket{b}\|$ and $H=A+B\succ 0$. We assume  $A$, $B$, and $H$ are given in their block-encodings with subnormalization factors $\alpha_A$, $\alpha_B$, and $\alpha_H$ respectively, and $\alpha_H\sim \alpha_A\sim \norm{A}\gg\alpha_B$. 
  {$\wt{\sigma}'_{\min}$ and $\wt{\sigma}_{\min}$ are defined in Theorems~\ref{thm:matrix_exp_contour} and \ref{thm:matrix_exp_inverse_transform} respectively. 
  {$\wt{\sigma}'_{\min}=\Omega(1/\alpha_B)$}
  and $\wt{\sigma}_{\min}=\Omega(1/(1+\|(A+B)^{-1}\|\|B\|))$.}
  Ref.~\cite{poulin2009sampling}, {which} uses phase estimation to prepare Gibbs state, estimates the number of queries to time-evolution, instead of block-encodings of the Hamiltonians. In this table we assume the time-evolution is done using Hamiltonian simulation methods such as in Ref.~\cite{LowChuang2019}, which simulates time-evolution for time $t$ with $\wt{\Or}(\alpha_H t)$ queries to $U_H$. It also uses 
 $\Or(\log(\frac{1}{\xi\epsilon}))$ qubits in the ``energy register''.}
  \label{tab:compare_algs_thermal}
\end{table}

The cost of each method depends on the actual form of $f(\cdot)$. Here for concreteness we consider $f(x)=e^{- x}$, which is directly related to the problem of preparing Gibbs state s in quantum physics. Without loss of generality we assume $H=A+B\succ 0$, so that $|f|\le 1$ evaluated on the spectrum of $H$. The costs of preparing the state $e^{-H}\ket{b}/\xi$ where $\xi=\|e^{-H}\ket{b}\|$, using phase estimation and LCU are given in \cref{tab:compare_algs_thermal}, and they all depend directly on the subnormalization factor in the block-encoding of $H$ denoted by $\alpha_H$. Naturally we have $\alpha_H\geq \|H\|\sim \|A\|$. We note that the Gibbs state preparation is a special case of the task discussed above. We can simply set $\ket{b}$ to be the maximally-entangled state to obtain a purified Gibbs state. In this particular case $\xi=\sqrt{N/Z}$ where $N$ is the Hilbert space dimension and $Z=\Tr(e^{-H})$ is the partition function. This will be discussed in Section~\ref{sec:purified_gibbs}. For simplicity and in order to be consistent with other works such as Ref.~\cite{vanApeldoorn2020quantum} we omit the dependence on temperature in the \cref{tab:compare_algs_thermal}, and this dependence will be discussed in Section~\ref{sec:purified_gibbs} as well. 

We propose two methods to evaluate $e^{-H}$ given the ability of fast inversion of $A$. The first method is based on the Cauchy contour integral formulation, 
\begin{equation}
f(x)=\frac{1}{2\pi\I} \oint_{\mc{C}} \frac{f(z)}{z-x} \ud z,
\end{equation}
where $\mc{C}$ is a simple closed curve and $f(\cdot)$ is analytic on a region containing $\mc{C}$ and its interior, and $x$ is inside $\mc{C}$. After proper discretization, the evaluation of the matrix function $f(A+B)$ becomes solving a series of preconditioned quantum linear system problems, which can be combined together using LCU. 
Our second method is based on a simple inverse transformation, namely 
\begin{equation}
f(x)=f(y^{-1}-a),
\end{equation}
where $y=(x+a)^{-1}$ for some $a\in\CC$. When $A+B$ is invertible, we may simply take $a=0$. Again we use preconditioned quantum linear system solver, combined with a standard LCU/QSVT procedure to evaluate $f(A+B)$. 
The cost of the two methods to prepare $e^{-(A+B)}$ are given in \cref{tab:compare_algs_thermal}, which is independent of $\norm{A}$.

There is a class of algorithms based on quantum walks and Metropolis sampling to prepare the Gibbs state \cite{yung2012quantum,ozols2013quantum,temme2011quantum}, which can be seen as a special case of implementing the matrix function $e^{-H}$. The complexity typically depends mainly on the gap of the transition matrix of the Markov chain, and thus the complexity estimate involves a different set of parameters. Therefore we do not compare the complexities of these algorithms with our methods. For example, Ref.~\cite{chowdhury2016quantum} uses LCU to prepare the Gibbs state, but it uses a different input model of the Hamiltonian from the block-encoding we use in this work.



\subsection{Notations}
\label{sec:notations}

A matrix $A\in \CC^{2^n\times 2^n}$ is referred to as an $n$-qubit matrix or $n$-qubit operator. Unless otherwise explained, we use the notation $N=2^n$, and $[N]=\set{0,\ldots,N-1}$. We will extensively use the technique of block-encoding, which is a way of embedding an arbitrary matrix as a sub-matrix of a larger unitary matrix.  Here using a unitary matrix $U_A$ to encode $A$ as a sub-matrix means that there exist  a normalizing constant $\alpha>0$ such that
\begin{equation}
    U_A = \begin{bmatrix} A/\alpha & *\\ *& * \end{bmatrix},
\end{equation}
where $*$ denotes arbitrary matrix blocks of proper sizes. In general, the matrix that we block-encode may only approximate $A/\alpha$.  We use the following notation to describe such encodings. 

\begin{defn}[Block-encoding, \cite{GilyenSuLowEtAl2018}]\label{defn:block_encoding}
An $(m+n)$-qubit unitary operator $U_A$ is called an $(\alpha, m, \epsilon)$-block-encoding of an $n$-qubit operator $A$, if 
$
\norm{A-\alpha(\bra{0^m}\otimes I_n) U_A (\ket{0^m}\otimes I_n)}\leq \epsilon.
$
\end{defn}
Here $m$ is the number of ancilla qubits for block-encoding, and $\alpha$ is called the block-encoding factor, or the subnormalization factor. The block-encoding has long been explicitly used in algorithms such as the quantum linear systems algorithm~\cite{HarrowHassidimLloyd2009}. The block-encoding is a powerful and versatile model, which can be used to efficiently encode density operators, Gram matrices, positive-operator valued measure (POVM), sparse-access matrices, as well as addition and multiplication of block-encoded matrices (we refer to \cite{GilyenSuLowEtAl2018} for a detailed illustration of such constructions). 

\begin{rem}\label{rem:noerror_blockencoding}
For simplicity of discussion, we may often assume the given block-encodings are error-free, e.g. we may assume $U_A$ is an $(\alpha, m,0)$-block-encoding of $A$. The error due to the given block-encodings can often be taken into account without much technical difficulties, but may complicate the presentation of results. We can then focus on the error introduced by other parts of the algorithm, such as that due to the polynomial approximation of smooth functions. \end{rem}

We also use the following notations throughout the paper:  The block-encoding of a matrix $A$ is generally denoted by $U_A$. Since the 2-norm of a unitary matrix $U_A$ is 1, it is guaranteed that the 2-norm of $A/\alpha$, which is a submatrix of $U_A$, is upper bounded by 1. This implies $\|A\|\leq\alpha$. Therefore in this work we usually bound the norm of a matrix in terms of its block-encoding subnormalization factor, which in many cases is known \textit{a priori}, for example in the case of $d$-sparse matrices \cite{ChildsKothariSomma2017,GilyenSuLowEtAl2018}.
Since this paper uses the inverse of matrices extensively, we may use $U'_A:=U_{A^{-1}}$ to denote the block-encoding of $A^{-1}$. For convenience, we use $\alpha_A,m_A$ to denote the subnormalization factor and the number of ancilla qubits for $A$, respectively, and use $\alpha_A',m_A'$ to denote those for $A^{-1}$.
{Throughout the paper we also frequently use $\wt{\sigma}_{\min}$ to denote a lower bound of the smallest singular value of a matrix of the form $I+A^{-1}B$.}

To simplify the notation, we may omit the normalization factor in the QLSP problem $\ket{x}=A^{-1}\ket{b}/\norm{A^{-1}\ket{b}}$, and write $\ket{x}\propto A^{-1}\ket{b}$ or $A\ket{x}\propto \ket{b}$. However, the normalization factor is not arbitrarily chosen, and the resulting state $\ket{x}$ is well defined. Though the phase factor in $\ket{x}$ is often not important, this allows us to define the distance between an approximate solution to QLSP $\ket{\wt{x}}$ and the true solution directly via the vector 2-norm $\norm{\ket{\wt{x}}-\ket{x}}$. In this paper we mostly use the vector 2-norm to quantify error, with the exception in Theorem~\ref{thm:greens_function} where we use trace distance to quantify the error of ground state preparation.  When we say a target quantum state $\ket{\phi}$ is prepared to precision $\epsilon$, it means that we prepare a quantum state $\ket{\psi}$ such that $\|\ket{\phi}-\ket{\psi}\|\leq \epsilon$. The relationship between the 2-norm distance and the more commonly used fidelity and trace distance is as follows: if for two pure states $\ket{\phi}$ and $\ket{\psi}$, $\|\ket{\phi}-\ket{\psi}\|=\epsilon$, then the fidelity $F$ satisfies
\[
F:=|\braket{\phi|\psi}|^2\geq (1-\epsilon^2/2)^2,
\]
and the trace distance 
between the two density matrices $\rho_{\phi}=\ket{\phi}\bra{\phi}$ and $\rho_{\psi}=\ket{\psi}\bra{\psi}$ satisfies
\[
T\left(\rho_{\phi},\rho_{\psi}\right):=\frac{1}{2} \Tr\left[\sqrt{(\rho_{\phi}-\rho_{\psi})^{\dagger}(\rho_{\phi}-\rho_{\psi})}\right]=\sqrt{1-|\braket{\phi|\psi}|^2} \leq \sqrt{1-(1-\epsilon^2/2)^2}=\epsilon\sqrt{1-\epsilon^2/4}.
\]

Additionally, we use the following asymptotic notations besides the usual $\Or$ notation throughout the paper: we write $f=\Omega(g)$ if $g=\Or(f)$; $f=\Theta(g)$ if $f=\Or(g)$ and $g=\Or(f)$; $f=\wt{\Or}(g)$ if $f=\Or(g\operatorname{poly}\log(g))$. We denote by $C^{m}(\mc{I})$ set of functions on an interval $\mc{I}$, which is differentiable $m$ times and the $m$-th derivative is continuous. Correspondingly $C^{\infty}(\mc{I})$ is the set of infinitely differentiable functions on $\mc{I}$ (also called smooth functions).

\begin{rem}(Dilation of a non-Hermitian matrix) When solving QLSP, it is often assumed that $A$ is a Hermitian matrix \cite{HarrowHassidimLloyd2009,ChildsKothariSomma2017}. This is because a non-Hermitian matrix can be dilated into a Hermitian matrix using one ancilla qubit
\begin{equation}
\label{eqn:dilation_trick}
    \wt{A}= \begin{bmatrix} 0 & A\\ A^\dagger & 0 \end{bmatrix}.
\end{equation}
When $A$ is given by its block-encoding $U_A$, the dilated Hermitian matrix $\wt{A}$ can be obtained through $U_{\wt{A}}=\ket{0}\bra{1}\otimes U_A + \ket{1}\bra{0}\otimes U_A^\dagger$ with subnormalization factor $1$. Note that this requires the controlled version of $U_A,U_A^{\dag}$. The quantum singular value transformation technique in \cref{sec:qsp_qlsp} can directly solve QLSP for non-Hermitian matrices without requiring a dilation step. In this paper, we assume $A\in\CC^{N\times N}$ is a general square matrix, and will explicitly specify when $A$ is taken to be a Hermitian matrix.
\end{rem}


\subsection{Organization of this paper}

The rest of the paper is organized as follows. In Section~\ref{sec:fastinv} we discuss certain matrices we can fast-invert on a quantum computer. This enables us to precondition linear systems, which is discussed in Section~\ref{sec:preconditioned_qlsp}. We then discuss two applications of the preconditioning technique we developed:
computing the many-body Green's function in Section~\ref{sec:many_body_greens_function}, and  evaluating matrix function $e^{-\beta H}$ in Section~\ref{sec:fast_algs_matrix_functions}.  Conclusion and discussion are given in Section~\ref{sec:discuss}. A brief review of quantum singular value transformation for solving QLSP, together with certain details of proofs and constructions are given in the appendices.

\section{Fast inversion}\label{sec:fastinv}
Our preconditioning method relies on fast-inverting certain class of matrices efficiently. 
{
\begin{defn}[Fast-invertible matrices]
A matrix $A$ is fast-invertible if, after rescaling $A$ so that $\|A^{-1}\|=1$, a $(\Theta(1),m,\epsilon)$-block-encoding of $A^{-1}$ can be obtained, and the number of queries to the oracles that determine $A$ is independent of the condition number $\kappa(A)$.
\end{defn}
}

{In this definition, the oracles are not restricted to yield a direct block-encoding of $A$. This can be seen in the examples of the unitarily diagonalizable matrices (\cref{sec:fastinv_diagaonlize}) and the $1$-sparse matrices (\cref{sec:fastinv_1sparse} and \cref{sec:fastinv_1sparse_general}).}

{Before further discussion, we first clarify the relation between the notion of fast-invertible matrices we propose here and the notion of fast-forwardable matrices. The two concepts are clearly closely related. In particular, if $A$ is a non-singular, Hermitian matrix $A$ that can be unitarily diagonalized efficiently, then $A$ is both fast-invertible and fast-forwardable, in the sense that the circuit depth for constructing the block-encoding of $A^{-1}$ and $e^{\I At}$ can be independent of $\kappa(A)$ and $t$, respectively. However, there is also an important difference: if $A$ is fast-forwardable, then the query complexity for preparing the state $e^{\I At}\ket{b}$ can be independent of $t$ for any $\ket{b}$. On the other hand, if a matrix $A$ is fast-invertible, then to prepare a normalized state that is parallel to $A^{-1}\ket{b}$, the query complexity still depends on $\|A^{-1}\ket{b}\|$, which in the worst case is lower bounded by $1/\kappa(A)$, if we rescale $A$ so that $\|A^{-1}\|=1$. However this lower bound is often not tight, as can be seen in the $d$-dimensional elliptic PDE example we discuss in \cref{prop:cost_elliptic}, and leads to vast over-estimation of the cost. Therefore we take $\|A^{-1}\ket{b}\|$ as an independent parameter rather than using the worst-case bound $1/\kappa(A)$. There are also instances in which the goal is not to prepare a quantum state but to read out a scalar value, as in Green's function evaluation discussed in \cref{sec:many_body_greens_function}. For these instances $\|A^{-1}\ket{b}\|$ can be irrelevant and the number of queries to $A$ can be completely independent of $\kappa(A)$.}
 

\subsection{Fast inversion of diagonal, and general $1$-sparse matrices}\label{sec:fastinv_1sparse}

We first consider a diagonal matrix $D\in \CC^{N\times N}$. Note that $\norm{D}=\max_i |D_{ii}|$, and $\norm{D^{-1}}=(\min_i |D_{ii}|)^{-1}$. The condition number is then $\kappa(D)=\max_i |D_{ii}|/\min_i |D_{ii}|$. Our goal is to solve the QLSP, i.e. to obtain a normalized state 
\begin{equation}
\ket{x}\propto D^{-1}\ket{b},
\label{eqn:qlsp_diag}
\end{equation}
where
\[
\ket{b}=\sum_{i\in[N]}b_i\ket{i}.
\]

Now assume that the diagonal entry of $D$ is accessible via an oracle 
\begin{equation}
\label{eq:diag_oracle}
    O_D\ket{i}\ket{0^l}=\ket{i}\ket{D_{ii}}, \quad i\in[N],
\end{equation}
where $\ket{D_{ii}}$ is a binary representation of the diagonal entry $D_{ii}$ (for simplicity assume the $l$-bit approximation of $D_{ii}$ is exact). 

\begin{rem}
\label{rem:polylog_overhead}
We assume each diagonal entry of $D$ can be efficiently computed using a classical boolean circuit of size $\Or(\polylog(N))$ with $m=\Or(\polylog(N))$ ancilla bits. In this case we can construct a quantum circuit with $\Or(\polylog(N)+l)$ gates and $\Or(\polylog(N)+l)$ ancilla qubits to implement this oracle $O_D$ \cite[Lemma 10.10]{arora2009computational}. For simplicity we omit these ancilla qubits in Eq.~\eqref{eq:diag_oracle}, since after reversing these quantum gates, their values will return to $\ket{0^m}$ at the end of the computation.
\end{rem}

The circuit for the block-encoding of $D^{-1}$, denoted by $U_{D^{-1}}$ is as in \cref{fig:fastinv_circuit}.
\begin{figure}[H]
  \centering
  \includegraphics[width=0.45\textwidth]{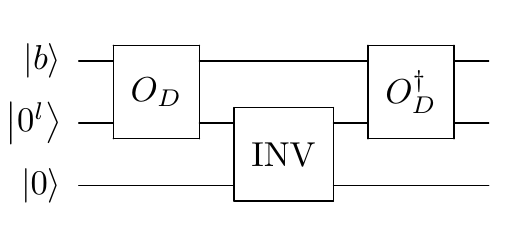}
  \caption{Quantum circuit for fast inversion.}
  \label{fig:fastinv_circuit}
\end{figure}
Here the inversion circuit $\mathrm{INV}$ satisfies 
\begin{equation}
\mathrm{INV}\ket{\zeta}\ket{0} = \ket{\zeta}\left( \frac{1}{\alpha'_D\zeta}\ket{0}+\sqrt{1-\Big|\frac{1}{\alpha'_D\zeta}\Big|^2}\ket{1} \right),
\label{eqn:}
\end{equation}
where $\ket{\zeta}$ is an $l$-qubit state representing an $l$-bit number $\zeta$. This circuit is also used in the HHL algorithm \cite{HarrowHassidimLloyd2009}. 
Here $\alpha'_D \ge (\min_i |D_{ii}|)^{-1}=\norm{D^{-1}}$. 
The output of the circuit is 
$$
U'_D\ket{b}\ket{0^l}\ket{0}=\sum_i  (\alpha'_D D_{ii})^{-1} b_i\ket{i}\ket{0^l}\ket{0} + \sum_i \sqrt{1-|(\alpha'_D D_{ii})^{-1}|^2}b_i\ket{i}\ket{0^l}\ket{1}. 
$$
Running the circuit and measuring the ancilla registers (i.e. the last two registers), and upon getting all zero output, which we take as a success, we will have a quantum state proportional to $D^{-1}\ket{b}$ in the first register. Hence $U'_{D}$ is an $(\alpha'_D,m'_D,0)$-block-encoding of $D^{-1}$ (recall $U'_D \equiv U_{D^{-1}}$)  and $m'_D=\Or(l+\polylog(N))$ when we take into account the ancilla qubits that have been omitted as mentioned in \cref{rem:polylog_overhead}. 

\begin{prop}[Fast inversion for diagonal matrices]
\label{prop:fastinv_diag}
For a diagonal matrix $D$ whose diagonal entries can be accessed through the oracle $O_D$ given in Eq.~\eqref{eq:diag_oracle}, and $\alpha'_D\ge 1/\min_i\{|D_{ii}|\}$, we can construct an $(\alpha'_D,m'_D,0)$-block-encoding of $D^{-1}$, given in Figure~\ref{fig:fastinv_circuit}, using $O_D$ and $O_D^\dagger$ both exactly once. Here $m'_D=\Or(l+\polylog(N))$ if $O_D$ uses $\Or(l+\polylog(N))$ ancilla qubits.
\end{prop}

For simplicity of discussion here, we assume $\alpha'_D=\norm{D^{-1}}$. The success probability of the above procedure is $(\|D^{-1}\ket{b}\|/\alpha'_D)^2$.
Denote by $\xi=\|D^{-1}\ket{b}\|$ (consistent with that in \cref{sec:qsp_qlsp}).  With amplitude amplification \cite{BrassardHoyerMoscaEtAl2002} we can boost the success probability to be greater than $1/2$ by repeating the above process  $\Or(\alpha'_D/\xi)$ times. Hence the success probability depends on both the operator $D$ as well as the state $\ket{b}$.  In the worst case $\braket{i|b}$ vanishes everywhere else other than $i=\arg\max_{i} |D_{ii}|$.  Then $\xi =\|D\|^{-1}$, and the number of repetition becomes $\Or(\norm{D^{-1}} \norm{D})= \Or(\kappa(D))$.  On the other hand, 
{if $|\braket{j|b}|$ is lower bounded by a constant for $j=\arg\min_{i} |D_{ii}|$, then $\xi= \Theta(\norm{D^{-1}})$, and the success probability of the fast inversion is $(\xi/\alpha_D')^2=\Omega(1)$. }
In this case, the cost of the fast inversion is $\Or(1)$ and is independent of $\kappa(D)$. 

{The cost of a QLSP solver may be reduced, if the effective condition number is much smaller than the condition number. Here the effective condition number refers to the ratio between the largest and the smallest singular values, whose corresponding singular vectors have a nonzero overlap with the right-hand side $\ket{b}$. Note that this is \emph{not} the reason of the speedup we discuss above. In the above discussion we do not require $\braket{j|b}$ to vanish anywhere, and therefore our method does not rely on the effective condition number being smaller than $\kappa(D)$.}


Let us now contrast the results above with the standard QSVT method for solving linear systems (briefly reviewed in \cref{sec:qsp_qlsp}). Start from an $(\alpha_D,m_D,0)$-block-encoding of $D$ denoted by $U_D$, we may take $\alpha_D=\norm{D}$. Applying \cref{thm:qsvt_qlsp}, 
using QSVT we can solve the QLSP \eqref{eqn:qlsp_diag} and obtain the solution to precision $\epsilon$ with probability at least $1/2$,  using $\Or((\kappa(D)^2/\|D\|\xi)\log(\kappa(D)/(\|D\|\xi\epsilon)))$ queries to $U_D$ and $U_D^\dagger$. The circuit depth for block-encoding the matrix inversion is $\Or(\kappa(D)\log(\kappa(D)/(\|D\|\xi\epsilon)))$, and this circuit and its inverse are repeated $\Or(\kappa(D)/\|D\|\xi)$ times in the amplitude amplification \cite{BrassardHoyerMoscaEtAl2002} procedure. So the fast inversion is always more efficient in terms of the circuit depth for block-encoding the matrix inversion. Considering the entire procedure for solving the QLSP, we need to take the value of $\xi$ into account. Since $\xi\in[1/\|D\|,\|D^{-1}\|]$, in the worst case when $\xi=1/\|D\|$, the fast-inversion method results in a quadratic speedup with respect to $\kappa(D)$. In the best case when $\xi=\|D^{-1}\|=\kappa(D)/\|D\|$, the cost of QSVT is still $\wt{\Or}(\kappa)$ while the cost of the fast inversion is $\Or(1)$.  

To illustrate how fast inversion works, let us consider a concrete example: 
\begin{equation}
D=\sum_{j=1}^n Z_j+(n+1)I_n,
\label{eqn:diag_ham}
\end{equation}
where $Z_j$ is the Pauli-Z matrix on the $j$-th qubit, and $I_n$ is the $n$-qubit identity operator. Then $\norm{D}=2n+1$, $\norm{D^{-1}}=1$, and $\kappa(D)=2n+1$. We may construct a $(1,m,\epsilon)$-block-encoding of $D^{-1}$ as follows. Given a state $\ket{i}\equiv \ket{s_1\cdots s_n}$ with $i\in [N]$ represented by a binary string and $s_j\in\{0,1\}$, we may first use take $O_D$ to be a quantum adder circuit, i.e.
\[
O_{D}\ket{i}\ket{0^l}=\ket{i}\ket{D_{ii}}, \quad D_{ii}=\left(2\sum_{j=1}^n s_j\right)+1,
\]
which can be implemented using a quantum adder circuit that uses $l=\lceil\log n\rceil+2$ ancilla qubits. We then find a $(1,l+1,0)$-block-encoding of $U'_{D}$ as in \cref{fig:fastinv_circuit}. If the right hand side vector $\ket{b}=\ket{0^n}$, then the number of repetitions needed to achieve $\Omega(1)$ success probability is $\Or(\kappa(D))=\Or(n)$. On the other hand, if $\ket{b}=\ket{1^n}$, only $\Or(1)$ repetitions are sufficient for the fast inversion method to succeed.

For $1$-sparse matrices that are not necessarily diagonal, we consider two different access models.
In the first case, given a general invertible, $1$-sparse matrix $A\in\CC^{N\times N}$, it can be written as $A=\Pi D$, where $\Pi$ is a permutation matrix, and $D$ is a diagonal matrix. We assume that we have direct access to the permutation  $\Pi$. Then $A$ is invertible if and only if $D$ is invertible. 
Given the availability of a $(1,m'_{\Pi},0)$-block-encoding of the unitary matrix
$\Pi^{-1}$ denoted by $U'_{\Pi}$, as well as an $(\alpha'_D,m'_D,\epsilon)$-block-encoding of $D^{-1}$ denoted by $U'_{D}$, we obtain an $(\alpha'_D,m'_D+m'_{\Pi},\epsilon)$-block-encoding of $A^{-1}=D^{-1}\Pi^{-1}$ via multiplication of block-encoded matrices \cite{GilyenSuLowEtAl2019}. The whole circuit takes three oracles queries in total.

In the second case, we only assume we have query access to the column of
the single nonzero element in each row (since the matrix is 1-sparse),
as well as to the value of the each element. The details for
constructing the fast inversion in this case is given in
Appendix~\ref{sec:fastinv_1sparse_general}. 

\subsection{Fast inversion of normal matrices}\label{sec:fastinv_diagaonlize}

If $A\in \CC^{N\times N}$ is a normal matrix, i.e. $A A^{\dag}=A^{\dag} A$, then $A$ can be unitarily diagonalized as $A=V D V^{\dag}$, where $V\in \CC^{N\times N}$ is a unitary matrix and $D\in \CC^{N\times N}$ is a diagonal matrix. Therefore $A$ is invertible if and only if $D$ is invertible. Using the fast inversion routine, we obtain an efficient block-encoding of $A^{-1}$ as
\begin{equation}
U'_{A}= (V\otimes I_{l+1})U'_{D}(V^{\dag}\otimes I_{l+1}).
\label{eqn:fastinv_diagonalize}
\end{equation}
Therefore we have the following proposition:
\begin{prop}[Fast inversion for normal matrices]
Suppose the eigenvalues of a normal matrix $A=VDV^\dagger$, where $V$ is unitary and $D$ is diagonal, can be accessed through the oracle $O_D$ given in Eq.~\eqref{eq:diag_oracle}, and $V$ can be efficiently implemented in a quantum circuit. Also let $\alpha'_D\geq 1/\min_i\{|D_{ii}|\}$, then we can construct an $(\alpha'_D,m_D',0)$-block-encoding of $A^{-1}$, using $O_D$, $O_D^\dagger$, $V$, and $V^\dagger$ each exactly once. Here $m'_D=\Or(l+\polylog(N))$ if $O_D$ uses $\Or(l+\polylog(N))$ ancilla qubits. 
\label{prop:fastinverse}
\end{prop}

Let us now consider another example of solving a linear system via fast inversion. Consider the following $d$-dimensional elliptic equation with periodic boundary conditions
\begin{equation}
-\Delta u(\vr)+u(\vr)=b(\vr), \quad \vr\in \Omega=[0,1]^d.
\label{eqn:elliptic_equation}
\end{equation}
Using a planewave basis set, we may expand  $u,b$ as 
\[
u(\vr)=\sum_{\vG\in\GG}\hat{u}(\vG) \exp(\I \vG \cdot \vr), \quad b(\vr)=\sum_{\vG\in\GG}\hat{b}(\vG) \exp(\I \vG \cdot \vr).
\]
where the planewaves index set is \[
\GG = \Set{
\vG=2\pi(g_1,\ldots,g_d)|  g_i\in\ZZ, i=1,\ldots,d}.
\]
The solution can then be readily written down in the Fourier space as 
\[
\hat{u}(\vG)=\frac{1}{\abs{\vG}^2+1} \hat{b}(\vG), \quad \vG\in\GG.
\]

We now use a finite number of $N$ planewaves to approximate the solution to \cref{eqn:elliptic_equation}. For simplicity we assume $N=2^n=2^{d\mathfrak{n}}$, so that there are $h^{-1}:=N^{1/d}=2^{\mathfrak{n}}$ planewaves per dimension. We further assume Here $h$ can be viewed as an effective mesh size. Then the planewave indices are restricted to 
\begin{equation}
\GG_h = \Set{
\vG=2\pi(g_1,\ldots,g_d)| -\frac{1}{2h}\le g_i< \frac{1}{2h}, g_i\in\ZZ, i=1,\ldots,d},
\label{eqn:planewave_restrict}
\end{equation}
with the cardinality $|\GG_h|=N$, and the resulting discretized QLSP can be written as $A\ket{u}\propto \ket{b}$. We may  write $A=VDV^{\dag}$, where $V$ is the $d$-dimensional quantum Fourier transform (QFT) \cite{NielsenChuang2000}, and the diagonal entries of $D$ are known and labeled by $\vG$ as
\[
D(\vG)=|\vG|^2+1, \quad \vG\in\GG_h.
\]


Hence the largest singular value of $A$ is
\[
\|A\|=\sigma_{\max}=d \frac{(2\pi)^2}{(2h)^2}+1=\frac{d\pi^2}{h^2}+1,
\]
which grows as $h^{-2}$ as the number of planewaves increases. The smallest singular value of $A$ is $1/\|A^{-1}\|=\sigma_{\min}=1$. Therefore the condition number $\kappa(A)=\Or(dh^{-2})$. A block-encoding of $D^{-1}$ can be explicitly constructed using classical arithmetics, and therefore we may fast invert the matrix $A$.

Let us consider $d=1$ first, where $n=\mathfrak{n}$ and $V$ is the standard QFT for $\mathfrak{n}$-qubits, denoted by $\mathrm{F}_{\mathfrak{n}}$. The implementation of $\mathrm{F}_{\mathfrak{n}}$ costs $\Or(\mathfrak{n}^2)$ gates. In the $d$-dimensional setting, $V$ can be constructed by $d$ copies of $\mathrm{F}_{\mathfrak{n}}$ as
\[
V=\mathrm{F}_{\mathfrak{n}}\otimes\cdots\otimes \mathrm{F}_{\mathfrak{n}}.
\]
So the total cost is $\widetilde{\Or}(\mathfrak{n}d)=\widetilde{\Or}(n/d)$. Therefore the circuit depth for block-encoding $A^{-1}$
is independent of the condition number $\kappa(A)$.

To consider the query complexity, we first note that the norm of the solution
\begin{equation}
\int |u(\vr)|^2\ud \vr=\sum_{\vG\in\GG} |\hat{u}(\vG)|^2=
\sum_{\vG\in\GG} \frac{1}{(\abs{\vG}^2+1)^2} |\hat{b}(\vG)|^2
\label{eqn:xi_elliptic}
\end{equation}
is well defined as long as the Fourier coefficients of the right hand side $\hat{b}(\vG)$ decays rapidly enough as $\abs{\vG}\to\infty$ (e.g. when $b(\vr)$ is a smooth function). Therefore with a finite truncation
\[
\sum_{\vG\in\GG_h} \frac{1}{(\abs{\vG}^2+1)^2} |\hat{b}(\vG)|^2=\Theta(1),
\]
and the quantity
\[
\xi= \norm{A^{-1}\ket{b}}=\Theta(1).
\]
This is asymptotically the best scenario for solving QLSP as discussed in Section~\ref{sec:fastinv_diagaonlize} and Appendix~\ref{sec:qsp_qlsp}. Combining the results of our bound on the complexity of the linear systems problem (given in Appendix~\ref{sec:qsp_qlsp} as \cref{thm:qsvt_qlsp}) and \cref{prop:fastinverse}, we have the following proposition for the cost of solving the elliptic equation \eqref{eqn:elliptic_equation}. 
\begin{prop}\label{prop:cost_elliptic}
In order to solve the $d$-dimensional elliptic equation \eqref{eqn:elliptic_equation} with a smooth right hand side $b(\vr)$ on the $d$-dimensional torus with precision $\epsilon$ and success probability larger than $1/2$, using a planewave discretization \cref{eqn:planewave_restrict} with grid size $h$ along each direction, the circuit depth of the quantum singular value transformation is  $\Or(dh^{-2}\log(1/\epsilon))$. The total number of queries to $U_A,U_A^{\dag}$ is $\Or(dh^{-2}\log(1/\epsilon))$, and the number of queries to $U_b$ is $\Or(1)$. The circuit depth, number of queries to $O_D$, $O_D^\dagger$, $V$, $V^\dagger$, $U_b$ are all $\Or(1)$ using fast inversion.
\end{prop}

In the example above, $A$ is a Hermitian matrix. If we replace $A$ by a normal matrix $A'=A-z I$ with $z\in \CC$ so that $|z|\ll \norm{A}$ and $A'$ is invertible, the conclusion still holds. This is the case in the contour integral formulation of computing matrix functions in \cref{sec:fast_algs_matrix_functions}, and in the computation of Green's functions in \cref{sec:many_body_greens_function}.

\begin{rem}
In the context of solving the $d$-dimensional Poisson equation, the fast inversion method is different from the method in \cite{CaoPapageorgiouPetrasEtAl2013} using the HHL algorithm, which employs the Hamiltonian simulation of the form $e^{-\I At}$. This requires $\Omega(\log(1/\epsilon))$ ancilla qubits to store the eigenvalues, and if the Hamiltonian simulation is implemented directly, the circuit depth is $\Omega(\norm{A}t)$. On the other hand, \cref{eqn:fastinv_diagonalize} is based on the direct access of oracles $O_D$ and $U'_{D}$, and the QFT part is decoupled from the inversion part. Neither the circuit depth nor the number of ancilla qubits needed depends on $\norm{A}$ or the accuracy $\epsilon$.
\end{rem}

{
\begin{rem}[Complexity lower bound with respect to $\xi$]
\label{rem:complexity_lower_bound_xi}
As can be seen in the above discussion, when $A$ is fast-invertible solving $A\ket{x}\propto\ket{b}$ can still have a $1/\xi$ dependence, where $\xi=\|A^{-1}\ket{b}\|$, and we assume for simplicity $\|A^{-1}\|=1$.
This dependence is in fact the best we can get. Consider the following simple example constructed from the unstructured search problem with $n$ bits and $N=2^n$: let $U_w$ be the oracle for  the unstructured search problem marking the target element $w$ through
\[
U_w\ket{s}=
\begin{cases}
\ket{s},&\ s\neq w \\
-\ket{s},&\ s=w
\end{cases}
\]
Now we let
\[
A = \frac{\sqrt{N}-1}{2}U_w + \frac{\sqrt{N}+1}{2}I,
\]
and this is a diagonal matrix whose diagonal entries are efficiently computable, and therefore fast-invertible.
Solving the QLSP $A\ket{x}\propto\ket{u}$, where $\ket{u}$ is the uniform superposition of all $n$-bit strings, results in a solution
\[
\ket{x} = \frac{\sqrt{N}}{\sqrt{2N-1}}\ket{w} +  \frac{1}{\sqrt{2N-1}}\sum_{s\neq w}\ket{s}
\]
with 
\[
\xi = \left\|A^{-1}\ket{u}\right\| = \left\|\frac{1}{\sqrt{N}}\ket{w} + \frac{1}{N}\sum_{s\neq w}\ket{s}\right\| = \Theta(N^{-1/2}). 
\]
If the QLSP with this $A$ can be solved with $o(1/\xi)$ queries to $A$, then it means we can obtain $\ket{x}$ with $o(\sqrt{N})$ queries to $U_w$. Measuring all qubits with the state $\ket{x}$ yields $w$ with probability around $1/2$. Therefore we would be able to solve the unstrctured search problem with query complexity $o(\sqrt{N})$, which is impossible. We therefore think the name ``fast-inversion'' for our method to deal with this kind of QLSP is appropriate because it cannot be asymptotically improved without further assumptions.
\end{rem}
}

\section{Preconditioned QLSP solver}
\label{sec:preconditioned_qlsp}
We consider the following linear system \cref{eqn:AplusB}, with a large condition number $\kappa(A+B)$, where $A$ and $B$ are $n$-qubit  matrices. {We are primarily interested in the following scenario: we assume $A+B$ is rescaled so that $\|(A+B)^{-1}\|=\Theta(1)$ and $\|A\|\gg \|B\|, \|(A+B)^{-1}\|, \|A^{-1}\|$ (if needed we may replace $A$ and $B$ with $A-z I$ and $B+z I$ for some $z\in \CC$). The condition number $\kappa:=\kappa(A+B)=\Theta(\|A\|)$.} 
The linear system is therefore ill-conditioned mainly as a result of the large $\|A\|$. 

We make the following assumptions regarding the query access in this problem. We assume we have $U'_{A}$, an $({\alpha'_A},m'_A,0)$-block-encoding of $A^{-1}$ prepared by the fast inversion procedure, and $U_{B}$, an $(\alpha_B,m_B,0)$-block-encoding of $B$.  For simplicity of presentation we assume these given block-encodings are error-free (see \cref{rem:noerror_blockencoding}). The right-hand side $b$ is accessed through a quantum circuit $U_b$, i.e. $\ket{b}=U_b\ket{0^n}$. 

In the algorithm we need multiple ancilla registers, and will refer to the register in which $\ket{b}$ is prepared and $\ket{x}$ is produced as the main register (also called the system register).

\subsection{Preconditioning the linear system}
 It is possible to reduce the condition number of the linear system by considering the following equivalent formulation
\begin{equation}
(I+A^{-1}B) \ket{x}\propto A^{-1}\ket{b}.
\label{eq:precond_linear_sys}
\end{equation}
\cref{lem:cond_precond_linear_sys} explains why this linear system might have a much smaller condition number than the linear system \eqref{eqn:AplusB}. 

\begin{lem}
\label{lem:cond_precond_linear_sys}
Define $W=I+A^{-1}B$, then the smallest singular value $\sigma_{\min}$ and largest singular value $\sigma_{\max}$ of $W$ satisfy
\[
\begin{aligned}
1/\sigma_{\min} &\leq 1+\|(A+B)^{-1}\|\|B\| =: C_{AB} \\
\sigma_{\max} &\leq 1+\|A^{-1}\|\|B\| =: C'_{AB}
\end{aligned}
\]
Hence the condition number of $W$ can be upper bounded as
\begin{equation}
\kappa(W)\leq \left(1+\|(A+B)^{-1}\|\|B\|\right)\left(1+\|A^{-1}\|\|B\|\right) = C_{AB} C'_{AB}.
\label{eqn:kappaW_upperbound}
\end{equation}
\end{lem}
\begin{proof}
Let $W \ket{x}=\ket{y}$. Then we have
\[
(A+B)\ket{x}=A\ket{y},
\]
therefore
\[
\begin{aligned}
(A+B)(\ket{x}-\ket{y}) &= -B\ket{y} \\
A(\ket{x}-\ket{y}) &= -B \ket{x},
\end{aligned}
\]
and these two equalities lead to
\[
\begin{aligned}
\|\ket{x}\| &\leq \|\ket{y}\|+\|\ket{x}-\ket{y}\| \leq (1+\|(A+B)^{-1}\|\|B\|)\|\ket{y}\|  \\
\|\ket{y}\| &\leq \|\ket{x}\|+\|\ket{x}-\ket{y}\| \leq (1+\|A^{-1}\|\|B\|)\|\ket{x}\|.
\end{aligned}
\]
These two inequalities then give a lower bound for the smallest singular value, and an upper bound for the largest singular value, as stated in the lemma. 
\end{proof}

\begin{figure}
    \centering
    \includegraphics[width=0.5\textwidth]{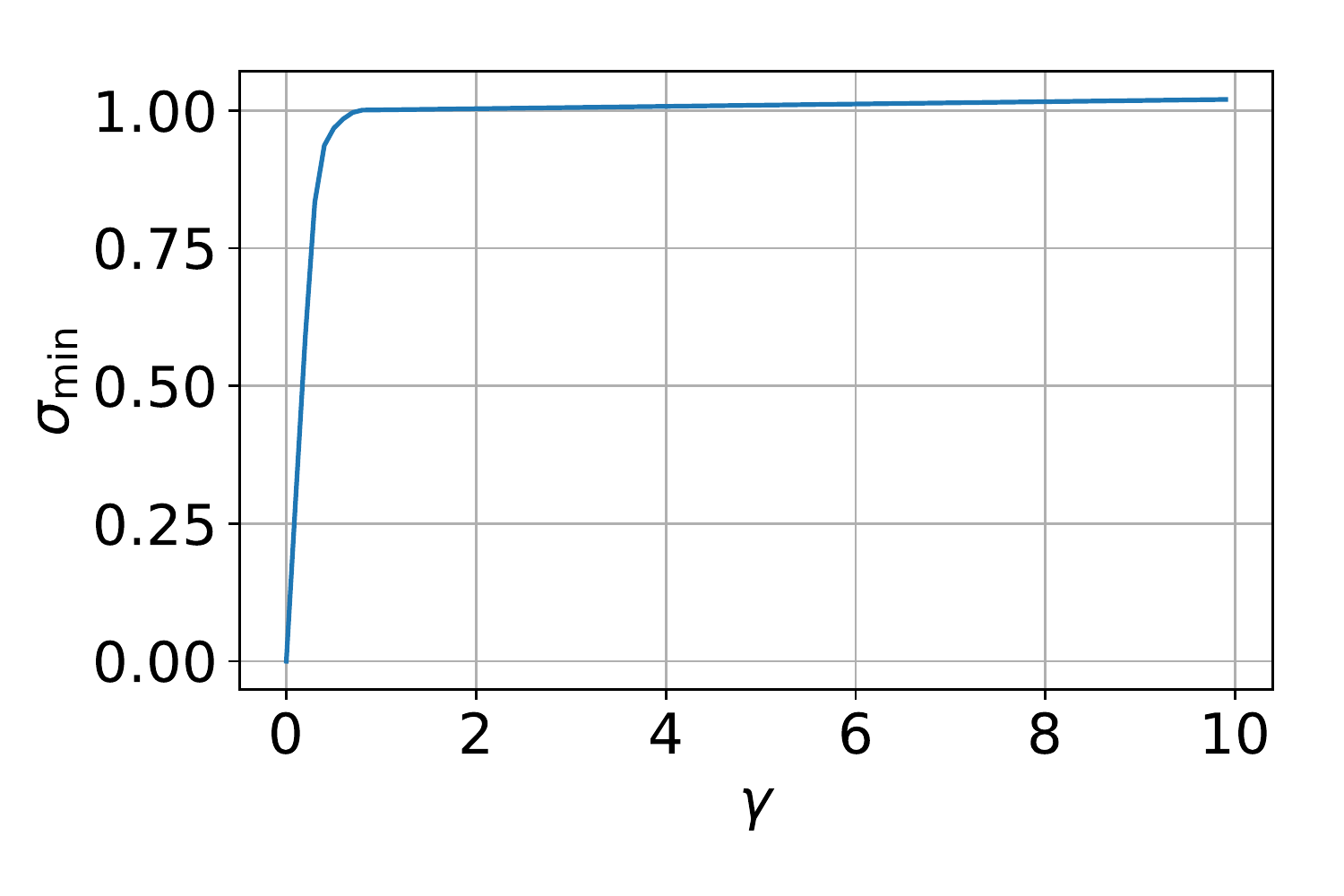}
    \caption{{The smallest singular value of $W=I+A^{-1}B$ where $A$ and $B$ come from discretizing the differential operator $-\Delta + \gamma (3+\cos(5x))$ with periodic boundary condition in 1D. $A=-\Delta_h + I$ where $\Delta_h$ is the discrete Laplacian operator, and $B=\gamma V-I$ where $V$ is the diagonal matrix whose diagonal elements contain the potential.}}
    \label{fig:sigma_min_W}
\end{figure}

\begin{rem}\label{rem:upperbound_kappaW}
The upper bound of $\kappa(W)$ does not depend on $\|A\|$ which we assume to be the main reason why the linear system \eqref{eqn:AplusB} is ill-conditioned. For a given pair of $A$ and $B$ 
we can always rescale $A$ and $B$, and possibly shifting by a multiple of identity, i.e. consider instead $A-\mu I$ and $B+\mu I$, so that the smallest singular values of $A+B$ and $A$ are $\Omega(1)$.  {Eq.~\eqref{eqn:kappaW_upperbound}} then gives us a bound for the condition number of $W$ that is independent of $\|A\|$. 
{When $\|B\|^2\ll A$, we have $\kappa(W)\ll \kappa(A+B)$.}

{Our bound of $1/\sigma_{\min}$ scales linearly with respect to $\|B\|$. So $\kappa(W)$ may scale quadratically with respect to $\|B\|$, leading to an undesirable polynomial dependence on the block-encoding subnormalization factor of $B$ in later applications. However, such estimate can be overly pessimistic in practice. In Figure~\ref{fig:sigma_min_W} we plot the smallest singular value of the matrix $W$ corresponding to discretizing a differential operator $-\Delta + \gamma (3+\cos(5x))$ with periodic boundary condition in 1D. In this example $\|A\|$ dominates because the Laplacian operator is unbounded in $L^2$ norm. We can see that instead of going to $0$ as $\gamma$ increases, $\sigma_{\min}$ actually increases. Therefore $1/\sigma_{\min}$ in this example can be bounded by a constant that does not depend on $\|B\|$ for $\gamma$ bounded away from 0.}
\end{rem}

In Sections~\ref{sec:greens_function} and \ref{sec:preconditioned_greens_function_solver}, we will see this procedure can be applied to many linear systems of practical interest. Next we consider how to construct a quantum circuit to solve the preconditioned linear system (\ref{eq:precond_linear_sys}).

\subsection{Quantum circuit construction}
\label{sec:quantum_circ_qlsp}

We want to construct a quantum circuit to block-encode $W^{-1}$ and
thereby solve the preconditioned linear system. This is done by using
QSVT (see \cref{sec:qsp_qlsp}).  In the following, we first construct a block-encoding for $W$. We then use QSVT to obtain a block-encoding of $W^{-1}$ to solve the linear system.

For the first step, since we can first apply the multiplication of block-encodings
\cite[Lemma 30]{GilyenSuLowEtAl2018} to obtain an $({\alpha'_A}\alpha_B,m'_A+m_B,0)$-block-encoding of $A^{-1}B$, then apply the linear combination of block-encodings \cite[Lemma 29]{GilyenSuLowEtAl2019} to obtain an $({\alpha'_A}\alpha_B+1,m'_A+m_B+1,0)$-block-encoding of $W=I+A^{-1}B$, which we denote by $U_{W}$. 



Following \cite[Corollary 69]{GilyenSuLowEtAl2018}, we may construct an odd polynomial $P(x)$ of degree $\Or(\frac{1}{\delta}\log(\frac{1}{\epsilon'}))$ that satisfies
\[
\left|P(x)-\frac{3\delta}{4x}\right|\leq\epsilon',\quad x\in[-1,-\delta]\cup[\delta,1].
\]
{Suppose the  smallest singular value of ${W}$ is lower bounded by $\wt{\sigma}_{\min}$. By Lemma~\ref{lem:cond_precond_linear_sys} we can choose $\wt{\sigma}_{\min}=1/C_{AB}$.} We will apply the polynomial to implement $(P^{\diamond}(W/({\alpha'_A}\alpha_B+1)))^\dagger$, where $P^\diamond(\cdot)$ denotes the generalized matrix function as defined in \cref{def:gen_matrix_function}. We intend to use this to approximate $(W/({\alpha'_A}\alpha_B+1))^{-1}$. Therefore we need to ensure the singular values of $W/({\alpha'_A}\alpha_B+1)$ lie in the interval $[\delta,1]$.  Recall that a lower bound of the smallest singular value of $W$ is {$\wt{\sigma}_{\min}$}.
For this purpose we choose 
\[
\delta={\wt{\sigma}_{\min}}/({\alpha'_A}\alpha_B+1).
\]

With this odd polynomial $P(x)$, we use QSVT \cite[Corollary 11]{GilyenSuLowEtAl2019} to construct a circuit to block-encode the matrix obtained by applying $P(x)$ to singular values of $W/({\alpha'_A}\alpha_B+1)$, using only one extra ancilla qubit (see \cref{sec:qsp_qlsp} and \cref{fig:qsp_circuit_real} for the circuit). This is a $(1,m'_A+m_B+2,0)$-block-encoding of $P^{\diamond}(W/({\alpha'_A}\alpha_B+1))$. 
We denote the Hermitian conjugate  (see \cref{eqn:ainv_qsvt} in \cref{sec:qsp_qlsp}) of the block-encoding constructed in this way by $U'_W$. $U'_W$ is therefore a $\left(\frac{4}{3{\wt{\sigma}_{\min}}},m'_A+m_B+2,\epsilon''\right)$-block-encoding of ${W}^{-1}$, where 
\[
\epsilon''=\frac{4\epsilon'}{3{\wt{\sigma}_{\min}}}.
\]
In other words
\[
\Big\|\frac{4}{3{\wt{\sigma}_{\min}}}(\bra{0^l}\otimes I)U'_W(\ket{0^l}\ket{y})-W^{-1}\ket{y}\Big\|\leq \epsilon'',
\]
where the first register contains $l=m'_A+m_B+2$ ancilla qubits, and the second register is the main register. 


From the above analysis we have obtained a block-encoding of $W^{-1}$. Note that $(A+B)^{-1}=W^{-1}A^{-1}$ is a product of two block-encoded matrices. Therefore by the multiplication of block-encodings \cite[Lemma 30]{GilyenSuLowEtAl2019} we have a $\left(\frac{4{\alpha'_A}}{3{\wt{\sigma}_{\min}}},2m'_A+m_B+3,{\alpha'_A}\epsilon''\right)$-block-encoding of $(A+B)^{-1}$.  This only uses $U'_A$ one extra time. 
We want the block-encoding error to be ${\alpha'_A}\epsilon''=\delta'$, then we need to choose
\begin{equation}
\label{eq:choose_epsilon'}
    \epsilon'=\frac{3\delta'{\wt{\sigma}_{\min}}}{4\alpha'_A }
\end{equation}
and the polynomial degree can then be expressed with respect to $\delta'$ as 
\[
\Or\left(\frac{1}{\delta}\log\left(\frac{1}{\epsilon'}\right)\right)=\Or\left(\frac{\alpha'_A\alpha_B}{{\wt{\sigma}_{\min}}}\log\left(\frac{\alpha'_A}{\delta'{\wt{\sigma}_{\min}}}\right)\right).
\]
The cost of applying QSVT scales linearly with the polynomial degree.
We can then summarize the result in the following theorem: 
\begin{thm}[Block-encoding of preconditioned matrix inverse]
\label{thm:precond_linear_sys}
Let $U'_A$ be an $({\alpha'_A},m_{A}',0)$-block-encoding of $A^{-1}$ implemented via fast inversion, 
$U_B$ be an $(\alpha_B,m_B,0)$-block-encoding of $B$. 
{Let $\wt{\sigma}_{\min}$ be a lower bound for the smallest singular value of $I+A^{-1}B$, which can be chosen to be $1/(1+\|(A+B)^{-1}\|\|B\|)$ as discussed in Lemma~\ref{lem:cond_precond_linear_sys}.}
Then for any $\delta'>0$ there exists a ${\left(\frac{4\alpha'_A}{3\wt{\sigma}_{\min}},2m'_A+m_B+3,\delta'\right)}$-block-encoding of $(A+B)^{-1}$ using ${\Or\left(\frac{\alpha'_A\alpha_B}{\wt{\sigma}_{\min}}\log\left(\frac{\alpha'_A}{\delta'\wt{\sigma}_{\min}}\right)\right)}$ applications of $U'_A$, $U_B$, their controlled versions, their inverses, and other primitive gates. 
\end{thm}
There are many parameters involved in the above discussion which can be confusing to readers. Here we briefly summarize their relations. The complexity depends directly on two parameters: $\delta$, which is how far the singular values of $W/(\alpha'_A\alpha_B+1)$ are bounded away from $0$, and $\epsilon'$, which is the error of polynomial approximation. In the block-encoding of $W^{-1}$, the polynomial approximation error is amplified into the block-encoding error $\epsilon''$, which is then amplified into the final block-encoding error $\delta'$ through the multiplication of two block-encoded matrix. We assume $\delta'$ is chosen a priori and therefore it requires us to choose $\epsilon'$ according to Eq.~\eqref{eq:choose_epsilon'}.


Next we want to solve the linear system $\ket{x}\propto(A+B)^{-1}\ket{b}$ where $\ket{x}$ is the normalized solution to the linear system.
We denote the block-encoding of $(A+B)^{-1}$ in the above theorem by $\mathcal{U}$. We denote
\[
\ket{w}=\frac{4\alpha'_A}{3{\wt{\sigma}_{\min}}}(\bra{0^{l'}}\otimes I)\mathcal{U}(\ket{0^{l'}}\otimes \ket{b}),
\]
where $l'=2m'_A+m_B+3$. Therefore $\ket{w}/\|\ket{w}\|$ is the output state after applying the block-encoding. Also we denote $\ket{y}=(A+B)^{-1}\ket{b}$, $\zeta=\|\ket{w}\|$, and $\xi=\|\ket{y}\|$.
Thus the normalized solution $\ket{x}=\ket{y}/\|\ket{y}\|$. Then we have
\[
\|\ket{w}-\ket{y}\|\leq\delta',\quad |\zeta-\xi|\leq\delta'.
\]
This leads to 
\[
\left\|\frac{\ket{w}}{\zeta}-\frac{\ket{y}}{\xi}\right\| \leq \frac{\zeta \|\ket{w}-\ket{y}\|+|\zeta-\xi|\|\ket{w}\|}{\zeta\xi}\leq \frac{2\delta'}{\xi}.
\]
Therefore in order to make sure the output normalized quantum state $\ket{w}/\|\ket{w}\|$ is $\epsilon$-close to $\ket{x}$ in terms of 2-norm distance, we need $\delta'=\xi\epsilon/2$. The success probability is
\[
\left\|(\bra{0^{l'}}\otimes I)\mathcal{U}(\ket{0^{l'}}\otimes \ket{b})\right\|^2=\frac{9\|\ket{w}\|^2{\wt{\sigma}_{\min}^2}}{16{\alpha'_A}^2}\geq \frac{9\xi^2(1-\epsilon/2)^2{\wt{\sigma}_{\min}^2}}{16{\alpha'_A}^2}.
\]
We can boost the success probability to be greater than $1/2$ by amplitude amplification, using 
$\Or(\alpha'_A/\xi{\wt{\sigma}_{\min}})$ repetitions.
Therefore we have the following corollary:
\begin{cor}[Preconditioned linear system solver]\label{cor:precond_linear_sys_solver}
Under the same assumptions as Theorem~\ref{thm:precond_linear_sys}, for the QLSP \eqref{eqn:AplusB},
an $\epsilon$-close solution vector can be obtained with $\Or\left(\frac{ {\alpha'_A}^2\alpha_B }{\xi{\wt{\sigma}_{\min}^2}}\log\left(\frac{{\alpha'_A}}{{\wt{\sigma}_{\min}}\xi\epsilon}\right)\right)$ applications of $U'_A$, $U_B$, their controlled versions, their inverses, and other primitive gates, in addition to $\Or\left(\frac{{\alpha'_A}}{{\wt{\sigma}_{\min}}\xi}\right)$ applications of $U_b$ and its inverse 
, where $\xi=\|(A+B)^{-1}\ket{b}\|$. {As in Theorem~\ref{thm:precond_linear_sys}, $\wt{\sigma}_{\min}$ can be chosen to be $1/(1+\|(A+B)^{-1}\|\|B\|)$}.
\end{cor}

Below we compare the dependence on the condition number of our preconditioning method against the dependence of directly using QSVT. 
{Let us consider the scenario we proposed at the beginning of Section~\ref{sec:preconditioned_qlsp}. We assume a rescaling is applied to $A+B$ so that $\|(A+B)^{-1}\|=\Theta(1)$, $\|A\|\to\infty$, and  $\|A^{-1}\|,\|B\|=\Or(1)$. As discussed before if $\|A^{-1}\|$ is large we can always replace it with $A-zI$ for some $z\in\CC$ that is away from the spectrum of $A$. Furthermore, we assume $\alpha'_A$ and $\alpha_B$ are not much larger than $\|A^{-1}\|$ and $\|B\|$, i.e. $\alpha'_A, \alpha_B=\Or(1)$.}
We also assume $\epsilon=\Omega(1)$ so that we do not need to consider the dependence on $\epsilon$. This is because both methods have a logarithmic dependence on $1/\epsilon$ and are therefore similar in this aspect.

{Under these assumptions we have}
\begin{equation}
\label{eq:kappa_A+B}
    \kappa(A+B)=\|A+B\|\|(A+B)^{-1}\|=\Theta(\|A+B\|){=\Theta(\|A\|)}.
\end{equation}
{The smallest singular value of $W=I+A^{-1}B$ is lower bounded by
\[
\wt{\sigma}_{\min} = 1/(1+\|(A+B)^{-1}\|\|B\|) = \Omega(1).
\]}From Corollary~\ref{cor:precond_linear_sys_solver}, we can see that the number of queries to all oracles become $\Or(\frac{1}{\xi}\log(\frac{1}{\xi}))$. This is {no longer directly dependent on $\kappa(A+B)$, though such dependence can exist indirectly through the dependence of $\xi$ on $\kappa(A+B)$.}
We consider the following two cases
\begin{enumerate}
    \item In the worst case, the following inequality becomes an equality: $\xi=\|(A+B)^{-1}\ket{b}\|\geq \|A+B\|^{-1}=\Omega(1/\kappa(A+B))$ by Eq.~\eqref{eq:kappa_A+B}. Therefore we have $\Or(\kappa(A+B)\log(\kappa(A+B)))$ query complexity for $U'_A$, $U_B$, and $\Or(\kappa(A+B))$ query complexity for $U_b$. 
    \item In the best case, $\ket{b}$ has {$\Omega(1)$} overlap  with the left singular vectors of $A+B$ corresponding to small singular values, and therefore $\xi$ can be as large as $\Omega(1)$, giving us a query complexity of $\Or(1)$ for all oracles. For a concrete example of this scenario see  \cref{prop:cost_elliptic}.
\end{enumerate}

In both cases we can compare with direct application of QSVT as discussed in Appendix~\ref{sec:qsp_qlsp}. The worst and best scenarios are discussed for QSVT without preconditioning on Page~\pageref{enum:two_cases} {in Appendix~\ref{sec:qsp_qlsp}}. In the worst case, under the assumption that the block-encoding of $A+B$, {denoted by $U_{A+B}$}, entails a subnormalization factor $\|A+B\|$, direct application of QSVT will need to query $U_{A+B}$ and its inverse $\Or(\kappa(A+B)^2\log(\kappa(A+B)))$ times, and $U_b$ and its inverse $\Or(\kappa(A+B))$ times. The preconditioning method results in a quadratic improvement. In the best case, {i.e. $\ket{b}$ has {$\Omega(1)$} overlap with the left singular vectors of $A+B$ corresponding to small singular values,} under the same assumption, 
{the number of queries to $U_{A+B}$ is $\Or(\kappa(A+B))$ and the number of queries to $U_b$ is $\Or(1)$. This improvement comes from the fact that the large overlap makes the success probability of applying the QSVT bounded away from zero by a constant.}
The improvement of using preconditioning in this case is more significant, as we can dispense with this linear dependence on the condition number altogether.



{
We remark that the speedup in our method does not come from a reduced effective condition number. Consider the following simple example: let $\ket{u_{\min}}$ and $\ket{u_{\max}}$ be the left-singular vectors of $A+B$ associated with the smallest and largest singular values respectively. Then for the QLSP $(A+B)\ket{x}\propto \ket{b}$ where $\ket{b}=\frac{1}{\sqrt{2}}(\ket{u_{\min}}+\ket{u_{\max}})$, we have $\xi=\Theta(1)$, and in the scenario discussed above solving the QLSP only requires $\Or(1)$ queries to all oracles. The effective condition number of this problem is however the same as the condition number of $A+B$, which is $\Theta(\|A\|)$.
}

\begin{rem}\label{rem:scale_alphaB}
The above procedure assumes we know the constants such as $C_{AB}=1+\|(A+B)\|^{-1}\|B\|$ and $\xi$. The algorithm still works if upper bounds to these constants are known.  In \cref{thm:precond_linear_sys} and \cref{cor:precond_linear_sys_solver}, 
a superlinear dependence on the block-encoding factor $\alpha_B$ (or alternatively the matrix norm $\|B\|$) will arise if we use the bound in Lemma~\ref{lem:cond_precond_linear_sys}, letting $\wt{\sigma}_{\min}=1/C_{AB}$.
However, according to the discussion in \cref{rem:upperbound_kappaW}, it is possible that ${\wt{\sigma}_{\min}}$ 
can be independent of $\alpha_B$. In this case, the preconditioned linear system solver can scale linearly with respect to $\alpha_B$.  
\end{rem}

\section{Evaluating Green's functions of quantum many-body systems}
\label{sec:many_body_greens_function}

We first give a short introduction of the representation of fermionic systems.  We consider here a second quantized representation wherein each state is referred to as an orbital.  The Pauli exclusion principle forbids two electrons from being in the same spin state and spatial state simultaneously.  Since spin for an electron can either be up or down, there are four possible occupation states for each orbital.  This means that two qubits are needed to represent a given configuration of an orbital.  For example, the qubit states $\ket{00}, \ket{01}, \ket{10}, \ket{11}$ are taken to represent an orbital containing no electrons, no spin down and one spin up, one spin down and no spin up and one spin up and down electron respectively.

Since an orbital is naturally expressed as a pair of qubits in quantum computing, it is natural to further divide an orbital into two spin-orbitals which correspond to both the quantum states for both the spin and spatial degrees of freedom.  In this notation, we can describe the occupation and dynamics of a set of spin orbitals using creation and annihilation operators such that a spin orbital corresponding to spin state $\sigma$ and spatial orbital $\nu$ is given by the state $\hat a_{\nu,\sigma}^\dagger \ket{0}_{\nu,\sigma}=\ket{1}_{\nu,\sigma}$ where $\hat a^\dagger_{\nu,\sigma}$ is a fermionic (anti-commuting) creation operator acting on the spin orbital.  Similarly $\hat a^\dagger_{\nu,\sigma}\ket{1}_{\nu,\sigma} =0$.  This leads us to the number operator $\hat n_{\nu,\sigma} = \hat a^\dagger_{\nu,\sigma} \hat a_{\nu,\sigma}$, which has the property that $\hat n_{\nu,\sigma} \hat a^\dagger_{\nu,\sigma}\ket{0}_{\nu,\sigma}=\hat a^\dagger_{\nu,\sigma}\ket{0}_{\nu,\sigma}$ and $\hat n_{\nu,\sigma} \ket{0}_{\nu,\sigma}=0$.  Thus this operator is often called a number operator because it counts the number of electrons in a spin orbital. 

We may also use a single index $i$ to represent the multi-index $(\nu,\sigma)$. Then the creation and annihilation operators $\hat{a}^{\dag}_i,\hat{a}_i$ can be expressed using the Pauli operator via e.g. the Jordan-Wigner transform \cite{NegeleOrland1988} as
\begin{equation}
\hat{a}_{i}=Z^{\otimes (i-1)}\otimes\frac12(X+\I Y)\otimes I^{\otimes (N-i)}, \quad \hat{a}^{\dagger}_i=Z^{\otimes (i-1)}\otimes\frac12(X-\I Y)\otimes I^{\otimes (N-i)},
\label{eqn:jordan_wigner}
\end{equation}
Correspondingly the number operator can be represented as
\begin{equation}
    \hat{n}_{i}=\frac12 (I-Z_{i}).
    \label{eq:ndecomp}
\end{equation}
Here $X,Y,Z,I$ are single-qubit Pauli-matrices.  Note that \cref{eqn:jordan_wigner,eq:ndecomp} naturally provide a $(1,1,0)$-block-encodings of $\hat{a}_i, \hat{a}^\dagger_i,\hat n_i$.

As a practical application of the preconditioned linear system solver, we consider the evaluation of the single-particle Green's function in quantum many-body systems. The fermionic Hamiltonian (in the spin-orbital formulation \cite{NegeleOrland1988}) can be naturally separated into the sum of two terms as
\[
\hat{H}=\hat{H}_0+\hat{H}_1.
\]
Here $\hat{H}_0,\hat{H}_1$ are the non-interacting part and interacting part of the Hamiltonian, respectively:
\begin{equation}
\hat{H}_0=\sum_{ij=1}^{N} T_{ij} \hat{a}_{i}^{\dagger} \hat{a}_{j}, \quad \hat{H}_1=\sum_{ijkl=1}^{N} V_{ijkl} \hat{a}_i^{\dagger}\hat{a}_j^{\dagger}\hat{a}_l \hat{a}_k.
\label{eqn:fermion_hamiltonian}
\end{equation}
In this section, $N$ is the number of spin orbitals used to discretize the continuous Hamiltonian, and the dimension of the Hamiltonian matrix $\hat{H}$ is $2^N$.

We denote by $\ket{\Psi_0}$ the ground state of $\hat{H}$ with $N_e$ electrons ($N_e\le 2N$), and $E_0$ is the corresponding ground state energy. We assume that the ground state $\ket{\Psi_0}$ can be prepared by an oracle with error $\varsigma$ and success probability at least $p$, and $E_0$ is known to some error $\varsigma'$. We will provide examples of the realization of $\hat{H}_0$, $\hat{H}_1$ in \cref{sec:example_manybody}. 

\begin{rem}[Complexity of solving the ground state problem]
The above assumptions are, in general, unlikely to be satisfied for generic fermionic systems.  This assumption is of course very strong because if it were true in general then $\QMA \subseteq \BQP$, which is widely believed to be false. Even the problem of deciding whether the ground state energy is above or below a threshold (within a fixed promise gap) is known to be in \QMA-hard \cite{KitaevShenVyalyi2002, KempeKitaevRegev2006, OliveiraTerhal2005, AharonovGottesmanEtAl2009}. Nonetheless, it is reasonable to make these assumptions for systems where an ansatz can be constructed that has polynomial overlap with the ground state, and where the spectral gap of the Hamiltonian can be bounded from below. These are the assumptions made in, e.g. \cite{ge2019faster,LinTong2020groundstate}, and is believed to occur for a wide range of realistic systems in physics and chemistry.
\label{rem:groundstate}
\end{rem}


\subsection{Single-particle Green's function}
\label{sec:greens_function}

The single-particle Green's function is a matrix valued function (formally mapping $\mathbb{C} \mapsto \mathbb{C}^{N\times N}$ matrix) that we denote $G(z)$.  Here the input $z$ can often be interpreted to be an energy shift and $G(z)$ is defined provided $E_0 + z$ is not an eigenvalue of $H$.  Also note that the dimension of the underlying Hilbert space for the problem is $2^{N}$, the matrix is relatively small compared to the dimension of the Hamiltonian.  
We first define the advanced and retarded Green's function (denoted by $G^{(+)}(z)$ and $G^{(-)}(z)$, respectively) as
\begin{equation}
\label{eq:greens_function_def}
\begin{aligned}
G_{i j}^{(+)}(z)&:=\left\langle\Psi_{0}\left|\hat{a}_{i}\left(z-\left[\hat{H}-E_{0}\right]\right)^{-1} \hat{a}_{j}^{\dagger}\right| \Psi_{0}\right\rangle \\
G_{i j}^{(-)}(z)&:=\left\langle\Psi_{0}\left|\hat{a}_{j}^{\dagger}\left(z+\left[\hat{H}-E_{0}\right]\right)^{-1} \hat{a}_{i}\right| \Psi_{0}\right\rangle.
\end{aligned}
\end{equation}
Then the time-ordered single-particle Green's function, or Green's function for short, is the sum of the two components
\begin{displaymath}
G(z)=G^{(+)}(z)+G^{(-)}(z).
\end{displaymath}
We assume $|\Im(z)|\geq\eta$. The value of $\eta$ is often referred to as the broadening parameter, and determines the resolution of Green's functions along the energy spectrum.

Below we demonstrate a procedure that allows us to directly compute $G_{ij}^{(\pm)}(z)$. 
 Now suppose we have an $(\alpha^{(+)},m^{(+)},\epsilon^{(+)})$-block-encoding of the matrix inverse $\left(z-\left[\hat{H}-E_{0}\right]\right)^{-1}$ denoted by $U^{(+)}$, then using the Jordan-Wigner transformation \eqref{eqn:jordan_wigner} and the block-encoding of product of matrices \cite[Lemma 30]{GilyenSuLowEtAl2019}, we can  construct an $(\alpha^{(+)},m^{(+)}+2,\epsilon^{(+)})$ block-encoding of $\hat{a}_i \left(z-\left[\hat{H}-E_{0}\right]\right)^{-1} \hat{a}_j^\dagger$, which we denote by $\wt{U}^{(+)}$.  Then the Hadamard test for non-unitary matrices in Appendix~\ref{sec:non_unitary_hadamard} tells us how to estimate $G^{(+)}(z)$. The same procedure can be applied to obtain $G^{(-)}(z)$.

We remark that if we are only interested in the imaginary part of the Green's function (or more accurately, the anti-Hermitian part of the Green's function defined as $\Gamma^{(\pm)}:=\frac{1}{2\I}(G^{(\pm)}-G^{(\pm)\dagger})$,  then we can directly use amplitude estimation without using the Hadamard test, which saves us one control qubit.  $\Gamma^{(\pm)}$ is related to the spectral functions of the many-body system. The details of computing $\Gamma^{(\pm)}$ is discussed in \cref{sec:imagGreen}.


\subsection{Preconditioned Green's function solver}
\label{sec:preconditioned_greens_function_solver}

{As can be seen from the above discussion, matrix inversion is a crucial part of evaluating Green's function. In this section, we use the preconditioning technique developed in Section~\ref{sec:preconditioned_qlsp} to efficiently perform the matrix inversion. Unlike in the QLSP setting in Corollary~\ref{cor:precond_linear_sys_solver}, the performance of Green's function solvers does not depend on the amplitude $\xi=\|A^{-1}\ket{b}\|$.}

It now only remains to determine the block-encoding of $\left(z-\hat{H}+E_{0}\right)^{-1}$. Since $|\Im z|\ge \eta$, the smallest singular value of $z-\hat{H}+E_0$ is bounded from below by $\eta$, and 
$\norm{\left(z-\hat{H}+E_{0}\right)^{-1}}\le \eta^{-1}$. 
If $\norm{\hat{H}_0}$ and $\norm{\hat{H}_1}$ are comparable, 
it is natural to consider constructing this block-encoding using existing QLSP solvers such as HHL or the ones based on LCU or QSVT.
The complexities of these direct approaches are in Table~\ref{tab:compare_algs_hubbard}, and for completeness the analysis is in Appendix~\ref{sec:query_complexities_HHL_LCU_Greens}.
{Direct block-encoding the matrix inverse $\left(z-\hat{H}+E_{0}\right)^{-1}$ using LCU or QSVT results in a linear dependence on $\|\hat{H}\|$ in the query complexity (here we assume the block-encoding factor of $z-\hat{H}+E_0$ is comparable to $\|\hat{H}\|$), as shown in Table~\ref{tab:compare_algs_hubbard}.}
However, in certain physical settings, we may have $\norm{\hat{H}_0}\gg \norm{\hat{H}_1}$ or $\norm{\hat{H}_1}\gg \norm{\hat{H}_0}$ (see \cref{sec:example_manybody}). Then we may reduce the complexity through preconditioned linear system solvers in Theorem~\ref{thm:precond_linear_sys}. 
{As shown in Table~\ref{tab:compare_algs_hubbard}, our method enables us to replace the dependence on $\|\hat{H}\|$ with a dependence on the smaller one between $\norm{\hat{H}_0}$ and $\norm{\hat{H}_1}$.}

According to \cref{rem:groundstate}, we assume the ground energy is known to a precision $\varsigma'$, and the ground state can be prepared to within trace-distance error $\varsigma$ with probability at least $p$ by some circuit $U_\Psi$. In the analysis below we first ignore the error of the ground energy for simplicity, but add back its contribution at the end.

Without loss of generality, we re-partition the Hamiltonian as $\hat{H}=\hat{A}+\hat{B}$, where $\|\hat{A}\|\gg \|\hat{B}\|$, and $\hat{A}$ can be efficiently unitarily diagonalized as in \cref{prop:fastinverse}. In order to use the preconditioning technique in Section~\ref{sec:preconditioned_qlsp}, we first split $z-\hat{H}+E_0$ into the sum of $z+\I-\hat{A}+E_0$ and $\hat{B}-\I$ for $\Im (z)>0$, or $z-\I-\hat{A}+E_0$ and $\hat{B}+\I$ for $\Im (z)<0$. An extra shift $\pm \I$ is introduced to so that $\|(z\pm \I-\hat{A}+E_0)^{-1}\|\leq 1$, and $z\pm\I-\hat{A}+E_0$ is still a normal matrix and can be fast inverted.

 For simplicity we first assume $\Im(z)>0$. We can construct a $(1,m'_A,0)$-block-encoding of $(z+\I-\hat{A}+E_0)^{-1}$ using the fast inversion of unitarily diagonalizable matrix technique in \cref{prop:fastinverse}, which we denote by $U'_A$, for $|\Im(z)|\geq\eta$. We also construct an $(\alpha_B+1,m_B+1,0)$-block-encoding for $\hat{B}-\I$ (assuming we have $(\alpha_B,m_B,0)$-block-encoding of $\hat{B}$), which can be constructed using \cite[Lemma 29]{GilyenSuLowEtAl2019}, and we denote it by $U_B$.
When $\Im(z)<0$ we only need to flip the sign of the extra shift.

{Let $\wt{\sigma}_{\min}$ be a lower bound of the smallest singular value of $I+(z+\I-\hat{A}+E_0)^{-1}(\hat{B}-\I)$.}
By Theorem~\ref{thm:precond_linear_sys}, {the smallest singular value is lower bounded $1/C_{AB}$, where $C_{AB}=1+\|(A+B)^{-1}\|\|B\|$, and it is easy to check $C_{AB}\leq 1+\frac{\alpha_B+1}{\eta}$. Thus a choice for $\wt{\sigma}_{\min}$ can be
\[
\wt{\sigma}_{\min} = \frac{\eta}{\eta+\alpha_B+1},
\]
which works for all $\hat{A}$ and $\hat{B}$. However larger values for $\wt{\sigma}_{\min}$ might be possible given more information about $\hat{A}$ and $\hat{B}$, as discussed in Remark~\ref{rem:scale_alphaB}.
} 
We can then construct a $(\frac{4}{3{\wt{\sigma}_{\min}}},2m_A'+m_B+3,\epsilon'')$-block-encoding of $(z-\hat{H}+E_0)^{-1}$, which uses $U'_A$, $U_B$, and other primitive gates for a total of $\Or(\frac{\alpha_B}{{\wt{\sigma}_{\min}}}\log(\frac{1}{\epsilon''{\wt{\sigma}_{\min}}}))$ times.

Now we determine the complexity of Green's function evaluation using this preconditioned solver. We apply the Hadamard test for non-unitary matrices described in Appendix~\ref{sec:non_unitary_hadamard}, and specifically Lemma~\ref{lem:non_unitary_hadamard_test}, to the matrix $(z-\hat{H}+E_0)^{-1}$, for which we have constructed a block-encoding using the preconditioning technique. Because amplitude estimation \cite{BrassardHoyerMoscaEtAl2002} is used in Lemma~\ref{lem:non_unitary_hadamard_test}, we have an $1/\epsilon$ dependence on the precision rather than the $1/\epsilon^2$ often seen in Monte-Carlo type methods. We also repeat amplitude estimation multiple times and take the median to exponentially reduce the failure probability of amplitude estimation, which is discussed in more detail in Appendix~\ref{sec:non_unitary_hadamard}.

Compared to Lemma~\ref{lem:non_unitary_hadamard_test}, there is a further source of error due to the inexact ground energy. Suppose instead of the exact ground energy $E_0$ we have an approximate $\wt{E}_0$. Then 
\[
\left\|\left(z-\hat{H}+E_0\right)^{-1}-\left(z-\hat{H}+\wt{E}_0\right)^{-1}\right\|=|\wt{E}_0-E_0|\left\|\left(z-\hat{H}+E_0\right)^{-1}\left(z-\hat{H}+\wt{E}_0\right)^{-1}\right\|.
\]
Since
\[
\left\|\left(z-\hat{H}+E_0\right)^{-1}\right\|\leq \frac{1}{\eta},\quad \left\|\left(z-\hat{H}+\wt{E}_0\right)^{-1}\right\|\leq \frac{1}{\eta},
\]
when $|\wt{E}_0-E_0|\leq \varsigma'$ as assumed at the beginning of this section, the error that comes from the inexact ground energy is upper bounded by $\varsigma'/\eta^2$.

After taking into account both the error due to the block-encoding of $(z-\hat{H}+E_0)^{-1}$ and the ground energy, we set $\epsilon''=\eta\epsilon/2$ in the above analysis, and arrive at the following theorem:

\begin{thm}
\label{thm:greens_function}
Given a unitary circuit $U_\Psi$ to prepare the $N$-particle ground state $\ket{\Psi_0}$ to trace-norm error $\varsigma$ with probability at least $p$, a $(1/\eta,m_A',0)$-block-encoding $U'_A$ of $(z+\I-\hat{A}+E_0)^{-1}$, an $(\alpha_B+1,m_B+1,0)$-block-encoding $U_B$ of $\hat{B}-\I$ for $\Im(z)\geq\eta>0$, {a lower bound $\wt{\sigma}_{\min}$ for the smallest singular value of $I+(z+\I-\hat{A}+E_0)^{-1}(\hat{B}-\I)$,} and an estimate of the ground energy that has an error upper bounded by $\varsigma'$, we can evaluate $G_{ij}(z)=\braket{\Psi_0|\hat{a}_i(z-\hat{H}+E_0)^{-1}\hat{a}_j^\dagger|\Psi_0}$ to precision $\frac{8}{3{\wt{\sigma}_{\min}}}\varsigma+\frac{\varsigma'}{\eta^2}+\epsilon$ with probability $\delta$ using
\begin{enumerate}
    \item $\Or( \frac{\alpha_B}{{\wt{\sigma}_{\min}^2} \epsilon}  \log(\frac{1}{\epsilon{\wt{\sigma}_{\min}}})\log(\frac{1}{\delta}))$ applications of $U'_A$ and $U_B$,
    \item 
    {$\Or( \frac{1}{\wt{\sigma}_{\min}\sqrt{p}\epsilon}\log(\frac{1}{\varsigma})\log(\frac{1}{\delta}) )$}
    applications of $U_\Psi$,
    \item Other primitive gates whose number is proportional to the sum of the two numbers above.
\end{enumerate}
{In the absense of a tighter bound, $\wt{\sigma}_{\min}$ can be chosen to be $\eta/(1+\alpha_B+\eta)$.}
When $\Im(z)\leq-\eta<0$ the complexity is the same and we only need to flip the sign of the shift $\I$ in the block-encodings.
\end{thm}

 
\subsection{Examples}\label{sec:example_manybody}

{Below we discuss the application of our preconditioned Green's function evaluation method to the Hubbard model, the second quantized quantum chemistry Hamiltonian in the plane wave dual basis, and the Schwinger model, and discuss whether there is speedup compared to the direct evaluation of Green's function through QSVT or LCU as discussed in Appendix~\ref{sec:query_complexities_HHL_LCU_Greens}. In all three models the Hamiltonian can be written as the sum of two terms $\hat{H}=\hat{A}+\hat{B}$, with $\|\hat{A}\|\gg\|\hat{B}\|$. The direct method using QSVT or LCU results in a linear dependence on $\|\hat{H}\|$, and therefore the dominating term, as discussed in Appendix~\ref{sec:query_complexities_HHL_LCU_Greens}. Our preconditioning method replaces the dependence on $\|\hat{H}\|$ with the dependence on the much smaller quantity $\|\hat{B}\|$. Such a dependence can however be cubic when no additional information about $\wt{\sigma}_{\min}$ in Theorem~\ref{thm:greens_function} is available.}

\subsubsection{Hubbard Model}\label{sec:hubbard}

The Hubbard model can be viewed as a prototypical system for describing electrons in strongly correlated materials. The Hamiltonian shares similarity with the full Coulomb Hamiltonian in \cref{sec:planewave}. Consider the 2D Hubbard model with $\sqrt{N}\times \sqrt{N}$ grid points that only has on-site interaction between electrons of opposite spins. The Hamiltonian reads
\[
\hat{H}=\sum_{\vx,\vy,\sigma} T(\vx-\vy) \hat{a}_{\vx, \sigma}^{\dagger} \hat{a}_{\vy, \sigma}+U \sum_{\vx} \hat{n}_{\vx,\uparrow}\hat{n}_{\vx,\downarrow}=:\hat{H}_0+\hat{H}_1.
\]
Here $\hat{n}_{\vx,\sigma}=\hat{a}_{\vx, \sigma}^{\dagger}\hat{a}_{\vx, \sigma}$, $\vx=(x_1,x_2)\in \ZZ^2$, and the non-interacting part only involves nearest neighbor interaction as 
\[
T(\vx-\vy)=
\begin{cases}
-t,& d(\vx,\vy)=1,\\
0,& \text{otherwise}.
\end{cases}
\]
For most fermionic problems of interest we have $t,U>0$. The two-dimensional domain is assumed to be periodic.

Using the fermionic fourier transform (FFFT)~\cite{babbush2018low}, one may transform the creation and annihilation operators to the momentum space as
\[
\hat{c}^{\dag}_{\vG,\sigma}=\mathrm{FFFT}^{\dag} \hat{a}^{\dag}_{\vx,\sigma}\mathrm{FFFT}, \quad 
\hat{c}_{\vG,\sigma}=\mathrm{FFFT}^{\dag} \hat{a}_{\vx,\sigma}\mathrm{FFFT},
\]
where $\vG=(G_1,G_2)$ is the index of vectors in the reciprocal space given by
$G_\alpha = 2\pi k_\alpha / \sqrt{N}$, and $k_\alpha \in \{-\sqrt{N}/2+1,\ldots, \sqrt{N}/2\},\alpha=1,2$. For simplicity we assume that $\sqrt{N}$ is an even number. 

Explicit circuits for the FFFT, for the Jordan-Wigner fermionic representation, can be found in~\cite{babbush2018low}. Using FFFT, the translation-invariant kinetic energy operator can be diagonalized in the momentum space as
\begin{equation}
\sum_{\vx, \vy,\sigma} T(\vx-\vy) \hat{a}_{\vx,\sigma}^{\dagger} \hat{a}_{\vy,\sigma}=\mathrm{FFFT} \left(\sum_{\vG,\sigma}\hat{T}(\vG) \hat{c}^{\dag}_{\vG,\sigma}\hat{c}_{\vG,\sigma}\right)\mathrm{FFFT}^\dagger\label{eq:HubbardFFFT}
\end{equation}
Here $\hat{T}(\vG)$ is the Fourier transform of the discretized kinetic operator.

There are a number of ways to express such a diagonal matrix.   In particular, this transformation yields $\hat{c}^\dagger_{\vG,\sigma}\hat{c}_{\vG,\sigma} = (I-Z_{\vG,\sigma})/2$, which can be simply implemented as $-Z_{\vG,\sigma}/2$ by neglecting a dynamically irrelevant shift to the energy in the equality. Applying the (ordinary) two-dimensional Fourier transform on $T(\vx)$ to find $\hat{T}(\vG)$, we find that there exists a unitary decomposition of the kinetic term such that the sum of the absolute value of the coefficients in the unitary decomposition is at most $\alpha_T$ which obeys
\begin{equation}
    \alpha_T\le \frac{N}{2}\max_{\vG}|\hat{T}(\vG)| = \max_{\vG} \left| \frac{1}{2}\sum_{\vx} T(\vx) e^{ \I \vG\cdot \vx} \right| \le N |t|.
    \label{eq:alphaT}
\end{equation}

We further find that (again up to an irrelevant constant shift in the energy)
\[
U\sum_{\vx} \hat{n}_{\vx,\uparrow}\hat{n}_{\vx,\downarrow}= \sum_{\vx} \frac{Z_{\vx,\uparrow} Z_{\vx,\downarrow} - Z_{\vx,\uparrow} - Z_{\vx,\downarrow}}{4}.
\]
Therefore the sum of the coefficients of the unitaries in this unitary decomposition, $\alpha_U$ satisfies
\begin{equation}
    \alpha_U \le \sum_{\vx} \frac{3|U|}{4} = \frac{3N|U|}{4} \le N|U|.
    \label{eq:alphaU}
\end{equation}
It then follows that we can construct an $(\alpha_H, \mathcal{O}(\log(N/\epsilon)),0)$-block encoding of $\hat{H}$ for
\begin{equation}
    \alpha_H \le \alpha_T + \alpha_U \le N(|t| + |U|).
\end{equation}

If $|U|$ is small compared to $|t|$, then the kinetic energy term dominates and so it makes sense to use the preconditioned algorithm to compute the Green's function where $\hat{A}$ is taken to be the kinetic operator and $\hat{B}$ the electron electron interaction.  Since the kinetic term can be diagonalized by the FFFT, we can use the fast inversion result of Proposition~\ref{prop:fastinverse} we can invert the kinetic part of the Hamiltonian using a single query to an oracle that yields the diagonal elements of $\hat{T}(\vG)$ and a single application of FFFT and its inverse. On the other hand, if $|t|$ is small compared to $|U|$ (which corresponds to the strongly correlated regime, and is often the case of interest) then we simply take $\hat{A}$ to be the onsite interaction term.  This is similar in spirit to the hybridization expansion in quantum Monte Carlo calculations \cite{GullMillisLichtensteinEtAl2011}. In either case, the scaling for the normalization constant is $\min(|U|,|t|)$.

We also require for our Green's function approach an oracle that yields a block encoding of $\hat B+\I$, where $\hat B$ is the two-body operator describing electron-electron interaction.  As mentioned previously, a block-encoding of $\hat B$ exists with a block-encoding factor $N \min(|t|,|U|)$. Therefore we can construct a block-encoding for $\hat B+\I$ with a value of $\alpha_B = N\min(|t|,|U|) +1 \in \Or(N\min(|t|,|U|))$.  Such a block-encoding requires $2$ queries to an oracle that computes the energy $U \sum_{\vx} \hat{n}_{\vx \uparrow}\hat{n}_{\vx \downarrow}$ for a particular configuration of electrons in position space.  

The Green's function can therefore be computed  with error at most $\epsilon$ and failure probability $\delta$ using a number of queries to oracles that compute the kinetic and potential energy of a given configuration (in momentum space and position space respectively) that are of
\begin{equation}
    \widetilde{\Or}\left(\frac{N^3(\min(|U|,|t|)^3}{\eta^2 \epsilon} \log\left(\frac{1}{\delta}\right)\right),
\end{equation}
which is independent of the value of $|t|$. {Here we used Theorem~\ref{thm:greens_function} and the worst-case bound $\wt{\sigma}_{\min}=\eta/(1+\eta+\alpha_B)$ as discussed in the theorem.} This can provide an advantage over the non-preconditioned version of the algorithm if $\min(|U|,|t|)$ is much smaller than $\max(|U|,|t|)$.

\begin{rem}
Using arguments given in Appendix~\ref{app:signature}, we can further optimize this scaling to depend on the number of electrons in the initial state.  If we denote the number of electrons to be $N_e$ then the scaling of the number of energy evaluations is improved to
\begin{equation}
    \widetilde{\Or}\left(\frac{N_e^3(\min(|U|,|t|)^3}{\eta^2 \epsilon} \log\left(\frac{1}{\delta}\right)\right),
\end{equation}
\end{rem}

\subsubsection{Plane Wave Dual Basis}\label{sec:planewave}

 For treating electrons in a periodic, infinite lattice, it is appropriate to use a periodized Coulomb operator (also called the Ewald interaction) \cite{FraserFoulkesRajagopalEtAl1996}.  This representation of electronic structure is also significant because it takes a similar form to the Hubbard model and consists of a kinetic and interaction term where the former can be diagonalized using the FFFT and the latter is diagonal in the computational basis. We omit the detailed discussion of the quantum many-body Hamiltonian in planewave dual basis here and refer readers to Ref.~\cite{babbush2018low}.  Unlike the Hubbard model, we find (somewhat surprisingly) that there is no clear advantage for using our preconditioned algorithm for this Hamiltonian in terms of Green's function evaluation, at least according to our worst case analysis.

Here we will take the $U_A$ to be a block encoding of the external potential and two-body interaction operators and $U_B$ to be a block-encoding of the kinetic energy term.
It has been found that the block-encoding factor for the kinetic energy term (denoted by $\alpha_B$) is $\Or\left(\frac{N^{5/3}}{\Omega^{2/3}} \right)$, while the block-encoding factor of the potential term (denoted by $\alpha_A'$) is $\Or\left(N^{7/3}/\Omega^{2/3} \right)$~\cite{babbush2018low}.  Thus Theorem~\ref{thm:greens_function}  tells us that the number of applications of oracles that compute the diagonal matrix elements of the kinetic and potential operators in the plane wave and plane wave-dual basis respectively that are needed to compute the Green's functions within error $\epsilon$ and failure probability $\delta$ for $|{\rm Im}(z)|\ge \eta$ is in 

\begin{equation}
    \wt{\Or}\left(\frac{N^5}{ \Omega^2\eta^2 \epsilon} \log\left(\frac{1}{\delta}\right)\right).
\end{equation}
{Again we used Theorem~\ref{thm:greens_function} and the worst-case bound $\wt{\sigma}_{\min}=\eta/(1+\eta+\alpha_B)$ as discussed in the theorem.}

\begin{rem}
This scaling shows that advantages do not necessarily follow by applying our technique to problems in chemistry.  If we were to compare the results in Table~\ref{tab:compare_algs_hubbard} then we see that the non-preconditioned Green's function computation scales would require a number of queries that is in $\widetilde{\Or}(N^{7/3}/\Omega^{2/3})$ (for constant $\epsilon, \eta, \delta$) (similar to the Hubbard model, we show that a small advantage can be gained by imposing a fermion number constraint but this does not change the conclusion here).  Therefore despite the asymptotic separation between the terms, we are not able to show an advantage of preconditioned linear system solvers in the context of the planewave dual basis set.  Therefore preconditioned linear system solvers with better scaling with respect to $\alpha_B$ would be of interest for future works. It is also possible that the preconditioned solver could be more useful in a different basis sets, such as the molecular orbital basis set. 
\end{rem}

\subsubsection{Schwinger Model}
As a final example of a class of models that our methods can be applied to, we will examine computing Green's functions for the Schwinger model.  The Schwinger model describes quantum electrodynamics in the $1+1$ dimension and is also used as a toy model for quantum chromodynamics.  

The Hilbert space for the Schwinger model is the tensor product of two spaces.  One consists of a tensor product of $N+1$ fermionic spaces and the other consists of a product of $N$ gauge field spaces.  Each gauge field can be thought of as a qubit register that can take $2L+1$ different integer values ranging from $-L$ to $L$.  There are two operators that we need to define that act on the gauge field space.  First we have $\hat{E}_r^2$ which is a diagonal operator and counts the energy stored in the gauge field with index $r\in \{1,\ldots,N\}$.  The second is $\hat{U}_r$, which adds one to value stored in the gauge field register and is analogous to a bosonic creation operator.  The action of these operators is given formally below
\begin{equation}
    \hat E_r^2 = \sum_{\varepsilon=-L}^L \varepsilon^2 \ket{\varepsilon}_r \bra{\varepsilon}_r,\qquad \hat U_r = \sum_{\varepsilon=-L}^L \ket{\varepsilon+1}_r \bra{\varepsilon}_r \qquad \hat U_r^\dagger = \sum_{\varepsilon=-L}^L \ket{\varepsilon-1}_r \bra{\varepsilon}_r.
\end{equation} 
Here we assume for $\hat U_r$ and its adjoint that the gauge field satisfies periodic boundary conditions at the cutoff located at $\varepsilon = \pm L$.

The Hamiltonian can be expressed for the Schwinger model on $N$ sites in one dimension using the operators $\hat{E}_r$ and $\hat{U}_r$ through the work of Kogut and {Susskind}~\cite{KS1975} as

\begin{align}
    H &=  \sum_{r=1}^N \hat E_{r}^2\otimes \hat I^{\otimes (N+1)} + x \sum_{r=1}^N \left[ \hat U_r \hat a_r^\dagger \hat a_{r+1} -  \hat U_r^\dagger  \hat a_r \hat a_{r+1}^\dagger \right] + \mu \sum_{r=1}^N (-1)^r \hat I^{\otimes N}_{L}\hat a^\dagger_r \hat a_r, \label{eq:SchwingerHM}
\end{align} 
Here we use $\hat I_L$ to denote the identity operation on the link variables and $\hat I$ to be the identity operation on the fermionic modes.
where $\hat E_r^2$ gives the energy in the gauge field that links two sites in the one-dimensional lattice and $\hat U_r$ is an operator that raises or lowers excitation number for the gauge field,  $\mu=2m/(ag^2)$ and $x=1/(ag)^2$, with $a$ the lattice spacing, $m$ the fermion mass and $g$ the coupling constant.

Once we have made this identification, we can use the construction in~\cite{shaw2020quantum} to express the Hamiltonian as a sum of unitary operations.  The simplest way to construct this as a sum of unitary matrices is to block encode $\hat{E}_r^2$
by noting
$\hat E_r^2 = \sum_k k^2\ket{k}\!\bra{k}$ is block-encoded by the unitary 
$$\sum_k \ket{k}\!\bra{k} \otimes e^{-\I Y \cos^{-1}(k^2/L^2)}: \ket{k} \ket{0} \mapsto \ket{k} \left(\frac{k^2}{L^2} \ket{0} + \sqrt{1- \frac{k^2}{L^2}} \ket{1}\right)$$ as per the construction that we describe in Appendix~\ref{sec:fastinv_1sparse_general}.  Similarly, $\hat U_r$ can be written as a sum of unitary adder circuits.  Let $U_A'$  block-encode $\sum_r \hat E_r^2 \otimes 1^r$ and let $U_b$ block-encode \[
x \sum_{r=1}^N \left[ \hat U_r \hat a_r^\dagger \hat a_{r+1} -  \hat U_r^\dagger  \hat a_r \hat a_{r+1}^\dagger \right] + \mu \sum_r (-1)^r I^{\otimes  r-1}_{L}\hat a^\dagger_r \hat a_r.  
\]
From the discussion in~\cite[Sections 2.1 and 2.3]{shaw2020quantum} we have $\alpha_{A'} = (N-1)L^2 $ and $\alpha_B =\Or((x+\mu)N)$.  

If we allow the gauge field cutoff to grow unboundedly, then asymptotically $\alpha_{A'}$ will dominate the complexity.  If we follow the reasoning used in the previous sections, then Theorem~\ref{thm:greens_function}, {together with the worst-case bound  $\wt{\sigma}_{\min}=\eta/(1+\eta+\alpha_B)$ as discussed in the theorem}, implies that the number of queries to $U_{A'}$ and $U_B$ are  
\begin{equation}
    \widetilde{\Or}\left(\frac{\alpha_B^3}{\eta^2 \epsilon} \right) = \widetilde{\Or} \left( \frac{(x+\mu)^3 N^3}{ \eta^2\epsilon} \right).
\end{equation}
There as $L$ increases while all other quantities remain fixed, this offers a potentially exponential improvement relative to the non-preconditioned example.  Therefore in the simulations of quantum field theory, preconditioning solvers can significantly reduce the computational complexity with respect to the size of the cutoff.



%
\section{Fast algorithm for evaluating matrix functions}
\label{sec:fast_algs_matrix_functions}

In this section we focus on implementing the matrix function $e^{-\beta H}$ where $H=A+B$, and applying it to a given quantum state $\ket{b}$.
For a positive semi-definite Hamiltonian $H=A+B\in\CC^{N\times N}$ with $\|A\|\gg \|B\|$ and $N=2^n$, we want to apply the imaginary-time evolution $e^{-\beta H}$. Following \cref{sec:fastinv_diagaonlize}, we further assume we have access to the eigendecomposition of $A=VDV^\dagger$, where $V$ can the efficiently implemented on a quantum computer, and $D$ is a diagonal matrix whose entries can be queried by the following oracle: 
\begin{equation}
O_D \ket{k}\ket{0^r} = \ket{k}\ket{D_{kk}}.
\end{equation}
Here we assume the diagonal entries can be represented exactly by $r$ bits. We also assume there is an $(\alpha_B,m_B,0)$-block-encoding of $B$ denoted by $U_B$.

This input model is inspired by the quantum many-body Hamiltonian for which we can diagonalize the non-interacting part efficiently on a classical computer (see \cref{sec:many_body_greens_function} for examples). Other input models exist for different settings. For example, Ref.~\cite{vanApeldoorn2020quantum} assumes access to a block-encoding of the Hamiltonian, and Ref.~\cite{chowdhury2016quantum} assumes the Hamiltonian is given through a linear combination of unitaries or projections, and we have access to what is essentially the square root of the Hamiltonian $\sqrt{H}$. This allows their algorithm to achieve $\Or(\sqrt{\beta})$ dependence. 

Table~\ref{tab:compare_algs_thermal} compares the query complexities of a few different algorithms assuming we are given the block-encoding of the Hamiltonian as an oracle. Note that in Table~\ref{tab:compare_algs_thermal} we did not consider the dependence on $\beta$ (or taking $\beta=1$), but this dependence is included in our analysis in this section. The reason for omitting $\beta$ in the table  is that $\beta$ is tied to the success probability of the procedure and the subnormalization of the output quantum state. When the state is normalized, the subnormalization factor amplifies the error in the process. This fact, combined with the different assumptions made in different methods, for example Ref.~\cite{vanApeldoorn2020quantum} assumes $\beta H\gg I$, complicates the fair comparison of different methods.

The rest of the section is organized as follows. We introduce an algorithm based on the contour integral formulation in \cref{sec:contour_int}, and a different algorithm based on the inverse transform in \cref{sec:inverse_transform}. Both formulations can be used to prepare a purified Gibbs state, which is discussed in \cref{sec:purified_gibbs}.

\subsection{Contour integral formulation}
\label{sec:contour_int}

We use the contour integral representation 
\begin{equation}
\label{eq:contour_int}
e^{-\beta x} = \frac{1}{2\pi \I}\oint_\Gamma \frac{e^{-\beta z}}{z-x}\ud z,
\end{equation}
where $x\geq 0$ and the contour is chosen as 
\begin{equation}
\label{eq:contour}
\Gamma = \{t^2-\zeta+\I t\in\CC:t\in\RR\}.
\end{equation}

\begin{figure}[ht]
    \centering
    \includegraphics[width=0.5\textwidth]{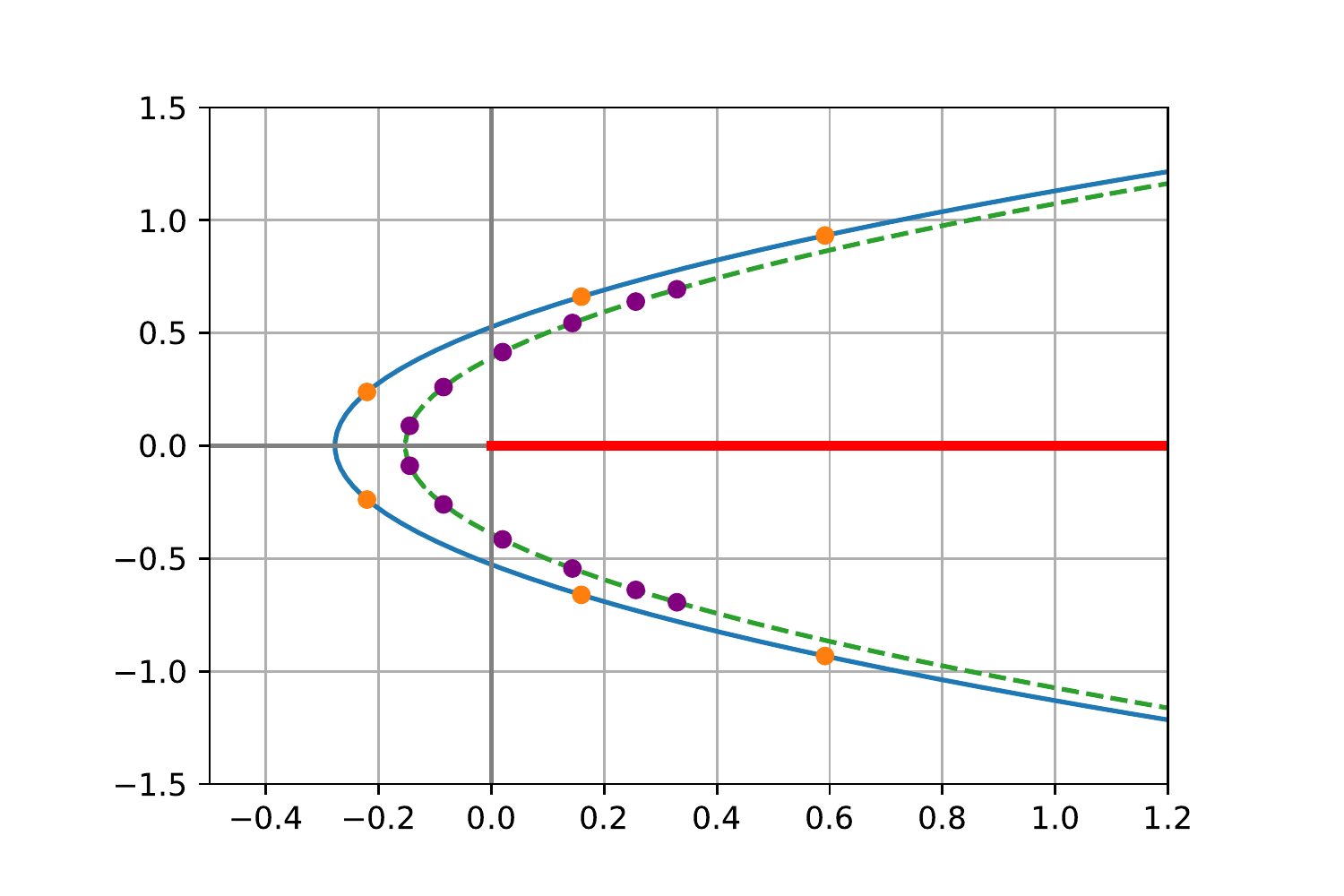}
    \caption{The parabolic contours $\Gamma$ and Gauss-Legendre quadrature nodes $\{z_j\}$. The parabola with solid line is for $\beta=3$ 
    and the one with dashed line is for $\beta=6$. The spectrum of the Hamiltonian $H=A+B$ is on the positive part of the real axis (solid red line). The dots on the parabolas show the quadrature nodes. Unlike the usual contour integral this parabolic contour extends to the infinity. This is however not a problem because the integrand decays very quickly to 0.}
    \label{fig:contour}
\end{figure}

We choose the parameter $\zeta$ in the following way, and will explain the reason in Appendix~\ref{sec:proof_quadrature_err}: 
\begin{equation}\label{eqn:contour_param}
b = \min\left(\frac{1}{2\beta},\frac{1}{6}\right) ,\quad \zeta=2b(1-b).
\end{equation}
In particular, $\beta\zeta\le 1$. Eq.~\eqref{eq:contour_int} then enables us to express the matrix function $e^{-\beta H}$ in terms of a linear combination of matrix inverses:
\begin{equation}
    e^{-\beta H} = \frac{1}{2\pi \I}\oint_\Gamma e^{-\beta z}(z-H)^{-1}\ud z,
\end{equation}

The contour integral formulation has been widely used to compute matrix functions on classical computers \cite{Higham2008}. In Ref.~\cite{trefethen2006talbot}, a number of parabolic contours have been considered, which are optimally tuned so that fast exponential convergence can be reached with respect to the number of discretization points $N$. This type of discretized contour integrals have also been used to obtain rational approximation \cite{cody1969chebyshev}, and  to invert Laplace transform \cite{talbot1979accurate,weideman2007parabolic,piessens1971gaussian}. 
However the parabolas they use for $e^{-\beta x}$ move away from the
origin in the negative direction along the real axis, and from our
analysis in this section we find that this results in an exponentially growing subnormalization factor in the LCU procedure. Therefore we need to design new parameterization of contours in \cref{eqn:contour_param}.

Once the contour is chosen, we may truncate the contour $\Gamma$ on a finite interval $t\in[-T,T]$, and apply the Gauss-Legendre quadrature formula to discretize this truncated contour integral. This leads to 
\begin{equation}
\label{eq:rational_approx}
e^{-\beta x} \approx \sum_{j\in[J]} \frac{\varrho_j}{z_j-x},
\end{equation}
where 
\begin{equation}
\label{eqn:z_rho_factor}
z_j = t_j^2-\zeta +\I t_j,\quad \varrho_j = \frac{T}{2\pi \I} w_j e^{-\beta z_j}(2t_j+\I),\quad t_j=Ts_j,
\end{equation}
and $s_j$, $w_j$ are the nodes and weights of Gauss-Legendre quadrature respectively, and the truncation range $T\geq 1$ is to be chosen. According to Appendix~\ref{sec:proof_quadrature_err}, the truncation error decays very rapidly as $T$ increases and therefore we do not need to choose a large $T$.
We first bound the error of this approximation in the following lemma
\begin{lem}
\label{lem:quadrature_err}
With $z_j$ and $\varrho_j$ defined above, and $\wt{\beta}=\min(\beta,3)$, we have  
\[
\left| e^{-\beta x} - \sum_{j\in[J]} \frac{\varrho_j}{z_j-x} \right| \leq \sqrt{\frac{2}{\beta\pi}}e^{1-\beta T^2}+\frac{64T^2\wt{\beta}e^{3/2}}{1-e^{-1/(8T\wt{\beta})}}e^{-J/(4T\wt{\beta})},
\]
for all $x\geq 0$.
\end{lem}
The proof of this lemma is in \cref{sec:proof_quadrature_err}. As can be seen here as we increase the truncation range $T$ and number of quadrature nodes $J$ the error decays to zero. 

With the approximation in Eq.~\eqref{eq:rational_approx}, we can implement $e^{-\beta H}$ as a linear combination of $(z_j-H)^{-1}$ by 
\begin{equation}
e^{-\beta H} \approx \sum_{j\in[J]} \varrho_j (z_j-H)^{-1}.
\end{equation}
To do this, we first need the block-encoding of the following matrix (called a select oracle in the context of LCU)
\begin{equation}
S=\sum_{j\in[J]} \ket{j}\bra{j}\otimes (z_j-H)^{-1} = \left(\sum_{j\in[J]}\ket{j}\bra{j}\otimes (z_j+\xi_j-A)-\sum_{j\in[J]} \ket{j}\bra{j}\otimes(B+\xi_j)\right)^{-1},
\end{equation}
where we choose
\begin{equation}
\xi_j = \begin{cases}
\I,\ &\mathrm{if}\  \Im z_j>0,\\
-\I,\ &\mathrm{if}\  \Im z_j\leq 0.\\
\end{cases}
\end{equation}
We can see that the operator inside the inverse, which can be seen as a block-diagonal matrix, is the sum of two parts, with the operator norm of the first part being much larger than that of the second part. Therefore we may employ the preconditioned linear system solver.

Since we have access to the eigendecomposition of $A$, we can obtain the various block-encodings needed in Theorem~\ref{thm:precond_linear_sys} easily. We summarize the cost for constructing these block-encodings in the following lemma: 
\begin{lem}
\label{lem:oracles_ab}
(a) a $(1,m_1,0)$-block-encoding of $\left(\sum_{j\in[J]}\ket{j}\bra{j}\otimes (z_j+\xi_j-A)\right)^{-1} $ can be implemented by $\Or(1)$ applications of $V$, $O_D$, and their inverses, with $m_1=\Or(\poly(r+\log(J)))$.

(b) a $(1+\alpha_B,m_2,0)$-block-encoding of $\sum_{j\in[J]} \ket{j}\bra{j}\otimes(B+\xi_j)$ can be implemented by one application of $U_B$, with $m_2=\Or(\polylog(J))$.
\end{lem}
We will provide the construction in \cref{sec:contour_integral_blockencoding}. Here the $\polylog$ factors originates from the classical computation related to $\{z_j,\xi_j\}$, as discussed in \cref{rem:polylog_overhead}.  We further have the following bounds: 
\begin{equation}
\begin{aligned}
&\left\| \left(\sum_{j\in[J]}\ket{j}\bra{j}\otimes (z_j+\xi_j-A)\right)^{-1} \right\| \leq 1, \\
&\left\| (z_j-A-B)^{-1} \right\| \leq \frac{1}{\zeta} \leq 2\max(\beta,3), \\
&\left\| \sum_{j\in[J]} \ket{j}\bra{j}\otimes(B_j+\xi_j) \right\| \leq 1+\alpha_B.
\end{aligned}
\end{equation}

{We introduce the parameter $\wt{\sigma}'_{\min}$ to be a lower bound of the smallest singular values of $I+(z_j+\xi_j-A)^{-1}(B_j+\xi_j)$ for all $j$. In other words,
\begin{equation}
    \label{eq:condition_sigma_prime}
    1/\wt{\sigma}'_{\min} \geq \max_{j\in [J]} \|(I+(z_j+\xi_j-A)^{-1}(B_j+\xi_j))^{-1}\|.
\end{equation}
We want to invert the block-diagonal matrix $\sum_{j\in[J]} \ket{j}\bra{j}\otimes (z_j-A-B)^{-1}$, and therefore $\wt{\sigma}'_{\min}$ plays the same role as $\wt{\sigma}_{\min}$ in Theorem~\ref{thm:precond_linear_sys}. We would prefer $\wt{\sigma}'_{\min}$ to be a tight lower bound in order to save computational cost.
By Lemma~\ref{lem:cond_precond_linear_sys}, when no better bound is available, we can  choose
\[
\wt{\sigma}'_{\min} = 1/\max_{j\in[J]}(1+\|(z_j-A-B)^{-1}\|\|B\|)=\Omega(1/\beta\alpha_B).
\]
}

Using these block-encodings and bounds, by \cref{thm:precond_linear_sys}, we have the following lemma:
\begin{lem}
\label{lem:select_oracle}
{Let $\wt{\sigma}'_{\min}$ satisfy Eq.~\eqref{eq:condition_sigma_prime}.}
An $(\alpha_S,m_S,\epsilon')$-block-encoding of the operator $S$ defined in Eq.~\eqref{lem:select_oracle} can be implemented using {$\Or\left(\frac{\alpha_B}{\wt{\sigma}'_{\min}}\log(\frac{1}{\wt{\sigma}'_{\min}\epsilon'})\right)$} applications of $V$, $O_D$, $U_B$, and their inverses, where {$\alpha_S=\Or(1/\wt{\sigma}'_{\min})$} and $m_S=\poly(r+\log(J))$. {Furthermore it is guaranteed that $\wt{\sigma}'_{\min}=\Omega(1/(\beta\alpha_B))$.}
\end{lem}
This block-encoding of the operator $S$, with some unitaries to prepare a state containing the coefficients 
\begin{equation}
\ket{c}=\frac{\sum_j \sqrt{|\varrho_j|}\ket{j}}{\sum_j |\varrho_j|},\quad \ket{c'}=\frac{\sum_j \sqrt{|\varrho_j|}e^{i\theta_j}\ket{j}}{\sum_j |\varrho_j|},
\end{equation}
where the phase factor $\theta_j$ satisfies $\varrho_j=|\varrho_j|e^{i\theta_j}$. This enables us to perform the LCU procedure for $\sum_j  \varrho_j (z_j-H)^{-1}$ through
\[
(\bra{c}\otimes I_n)\left(\sum_{j\in[J]}\ket{j}\bra{j}\otimes (z_j-H)^{-1}\right)(\ket{c'}\otimes I)=\sum_{j\in [J]}  \varrho_j (z_j-H)^{-1}.
\]
What we get in the end is an $(\alpha_{\mathrm{LCU}},m_{\mathrm{LCU}},\epsilon')$-block-encoding of $\sum_j  \varrho_j (z_j-H)^{-1}$, where
\begin{equation}
\label{eq:block_encode_fact_lcu}
\alpha_{\mathrm{LCU}} = \alpha_S\sum_{j} |\varrho_j|.
\end{equation}
Note that (using \cref{eqn:z_rho_factor}) 
\begin{equation}
\sum_{j} |\varrho_j| = \frac{T}{2\pi}\sum_j w_j \left|e^{-\beta z_j}(2t_j+1)\right| \rightarrow \frac{1}{2\pi}\int_{-\infty}^{\infty} \ud t e^{-\beta (t^2-\zeta)}(2|t|+1) =\Or\left(\frac{1}{\sqrt{\beta}}\right),
\end{equation}
where we have used the relation $\beta\zeta\leq 1$. Therefore substituting this into Eq.~\eqref{eq:block_encode_fact_lcu} we have $\alpha_{\mathrm{LCU}} = \Or(\alpha_B \sqrt{\beta})$.  

We then estimate the error of this block-encoding. So far the error for block-encoding $\sum_j \varrho_j (z_j-H)^{-1}$ has been accounted for, and it is bounded by $\epsilon'$. An additional source of error is due to the difference between $\sum_j \varrho_j (z_j-H)^{-1}$ and $e^{-\beta H}$. This error is bounded by Lemma~\ref{lem:quadrature_err}. The total error $\epsilon$ is the sum of these two errors. Therefore we set $\epsilon'=\epsilon/2$, and choose $J$ and $T$ so that the error bound in Lemma~\ref{lem:quadrature_err} is bounded by $\epsilon/2$. We need to choose $T$ and $J$ to be
\begin{equation}
T=\Or\left(\sqrt{\max\left(1,\frac{1}{\beta}\right)\log\left(\frac{1}{\epsilon}\right)}\right),\quad J=\wt{\Or}\left(\max\left(1,\beta\right)\left(\log\left(\frac{1}{\epsilon}\right)\right)^{3/2}\right).
\end{equation}
The above results can be summarized into the following theorem:
\begin{thm}
\label{thm:matrix_exp_contour}
{Let $\wt{\sigma}'_{\min}$ satisfy Eq.~\eqref{eq:condition_sigma_prime}.}
An $(\alpha_{\mathrm{LCU}},m_{\mathrm{LCU}},\epsilon)$-block-encoding can be constructed for $e^{-\beta H}$, with $H=A+B$ as described at the beginning of this section and the oracles $V$, $O_D$, and $U_B$ described above, where
\[
\alpha_{\mathrm{LCU}} = {\Or(1/ (\sqrt{\beta} \wt{\sigma}'_{\min}))}, \quad m_{\mathrm{LCU}}=\Or(\poly(r+\log(\beta)+\log\log(1/\epsilon))),
\]
using {$\Or\left(\frac{\alpha_B}{\wt{\sigma}'_{\min}}\log(\frac{1}{\wt{\sigma}'_{\min}\epsilon'})\right)$}
applications of $V$, $O_D$, $U_B$, and their inverses.
{Furthermore it is guaranteed that $\wt{\sigma}'_{\min}=\Omega(1/(\beta\alpha_B))$.}
\end{thm}


We remark that the number of qubits needed is $n+\Or(\poly(r+\log(\beta)+\log\log(1/\epsilon)))$ where the $\poly$ comes from the cost of classical Boolean circuits to perform algebraic operations on eigenvalues of $A$ and the conversion to quantum circuit using \cite[Lemma 10.10]{arora2009computational}. A more detailed estimate of the number of qubits needed and the gate complexity will require estimating these costs. 

Given the block-encoding we can apply the matrix function $e^{-\beta H}$ to a quantum state $\ket{b}$ which is prepared by a circuit $U_b$, i.e. $U_b\ket{0^n}=\ket{b}$. We let $\xi=\|e^{-\beta H}\ket{b}\|$, and the goal is to prepare a state $e^{-\beta H}\ket{b}/\xi$. Directly applying the block-encoding has a success probability of $\Omega(\xi^2/\alpha_{\mathrm{LCU}}^2)$. Using amplitude amplification \cite{BrassardHoyerMoscaEtAl2002} 
we can boost the success probability to at least 1/2 with $\Or(\alpha_{\mathrm{LCU}}/\xi)$ applications of the block-encoding circuit and its inverse. 
Because of the subnormalization of the output quantum state by a factor of $\xi$, we need the block-encoding error to be $\Or(\xi\epsilon)$ instead of $\Or(\epsilon)$. Therefore in total we need to query $V$, $O_D$, $U_B$, and their inverses 
{$\Or(\frac{\alpha_B}{\xi\sqrt{\beta}\wt{\sigma}'^2_{\min}}\log(\frac{1}{\xi\epsilon\wt{\sigma}'_{\min}}))$}
times, and $U_b$ and its inverse 
{$\Or(\frac{1}{\xi\sqrt{\beta}\wt{\sigma}_{\min}})$}
times. Therefore we have obtained the complexity for applying matrix exponential using contour integral in Table~\ref{tab:compare_algs_thermal}.

\subsection{Inverse transform formulation}\label{sec:inverse_transform}



The basic idea of using the inverse transform to accelerate the computation of $f(A+B)$ is as follows. We assume that $(A+B)\succ 0$ and that $(A+B)^{-1}$ can be efficiently block-encoded with a block-encoding factor $\alpha_{A+B}'$, as indicated in \cref{thm:precond_linear_sys}. The inverse transform is simply 
\begin{equation}
    f(A+B) = g((A+B)^{-1}),
\end{equation}
where
\begin{equation}
    g(y) = f(y^{-1}).
\end{equation}
Then instead of finding a block-encoding of $f(A+B)$, we can alternatively find a block-encoding for $g((A+B)^{-1})$.  The efficiency of such a transformation relies on the behavior of $g$  within the interval $[-1,1]$.
In particular, the behavior of $g$ near the origin plays a critical role, which reflects the decay of the original function $f$ at infinity.
Our strategy is then to approximate $g(y)$ uniformly on $[-1,1]$ by a Chebyshev
series, and the truncated Chebyshev series can then be implemented via QSVT. Compared to the contour integral formulation, the use of the inverse transform does not require the usage of the LCU technique, and provides a simpler method for preparing the thermal state. 

An example is the imaginary time evolution $e^{-\beta H}$. 
The corresponding function is 
$f(x) =\operatorname{sign}(x)e^{-\zeta x}$,  and we may construct $g(y) = \operatorname{sign}(y)e^{-\zeta |y|^{-1}}$. 
The parameter $\zeta$ will later be specified to be $\zeta = \beta/\alpha_{A+B}'$ 
to block encode $e^{-\beta H} = g((A+B)^{-1}/\alpha_{A+B}')$. 
The reason why we put a sign function in $f(x)$ and $g(y)$ is to ensure 
that $g(y)$ is an odd function on $[-1,1]$, then the corresponding 
truncated Chebyshev series is also odd and can be implemented via 
QSVT discussed in \cref{sec:qsp_qlsp}, in which we only describe how to 
apply QSVT to block-encode odd polynomials for technical simplicity. 
We remark that QSVT can also be used to block-encode general 
polynomials (see \cref{rem:qsvt}), therefore the inverse transform 
approach can be applied to general functions, provided that 
the function $g(y)$ is in the Gevrey class (to be defined later), which 
means that the original function $f(x)$ decays at infinity in some sense.

We first note that the function $\operatorname{sign}(y)e^{-\zeta \abs{y}^{-1}}$ is in $C^{\infty}([-1,1])$ but not real analytic at $y=0$. This complicates the convergence analysis when we approximate this function via polynomials of $y$. Nonetheless, we shall show that the query complexity only scales logarithmically with respect to the error $\epsilon$. To this end, we first define a subset of smooth functions, called the Gevrey class as follows. 
\begin{defn} The Gevrey class of order $\sigma$ on  an interval $\mc{I}$ is defined as
\begin{equation}
    G^{\sigma}(\mc{I}) = \{g(y)\in C^{\infty}(\mc{I}): \exists C>0,R>0, s.t.  |g^{(k)}(x)| \leq CR^k(k!)^{\sigma}, \forall x\in \mc{I}, k \geq 0\}. 
\end{equation}
\end{defn}

Note that $G^1([-1,1])$ represents the class of real analytic functions on $[-1,1]$, 
and $G^1([-1,1]) \subset G^{\sigma}([-1,1]) \subset C^{\infty}([-1,1])$ for all $\sigma > 1$. In order to show that $\operatorname{sign}(y)e^{-\zeta \abs{y}^{-1}}$ belongs to the Gevrey class, we will use the chain rule of high order derivatives (called the Fa\`a di Bruno's formula, see e.g.~\cite[Theorem 1.3.2]{krantz2002primer}:
\begin{lem}[Fa\`a di Bruno's formula]
\label{lem:chain_rule}
    Let $h,g$ be smooth functions, and $f(s) = h(g(s))$, then
    \begin{equation}
       f^{(k)}(s) = \sum \frac{k!}{q_1!(1!)^{q_1}q_2!(2!)^{q_2}\cdots q_k!(k!)^{q_k}}h^{(q_1+q_2+\cdots+q_k)}(g(s))\prod_{j=1}^k\left(g^{(j)}(s)\right)^{q_j}
    \end{equation}
    where the sum is taken over all $k$-tuples of non-negative integers $(q_1,\cdots,q_k)$ satisfying $\sum_{j=1}^k jq_j = k$. 
\end{lem}
\begin{prop}\label{prop:Chebyshev_gevrey_coeff}
    Let $g(y) = \operatorname{sign}(y)e^{-\zeta |y|^{-1}}$. Then for any $\zeta > 0$, $g(y) \in G^3([-1,1])$, with 
    $C = 1, R = 16e\zeta^{-1}$. 
\end{prop}
\begin{proof}
   Without loss of generality we only consider $y > 0$. 
   We view $e^{-\zeta y^{-1}}$ as the composition of $e^{w}$ and 
   $w = -\zeta y^{-1}$. Then by \cref{lem:chain_rule}, 
   \begin{align*}
       g^{(k)}(s) &= \sum_{\sum jq_j=k} \frac{k!}{q_1!(1!)^{q_1}q_2!(2!)^{q_2}\cdots q_k!(k!)^{q_k}}e^{-\zeta y^{-1}}\prod_{j=1}^k\left((-1)^{j+1}\zeta j!y^{-j-1}\right)^{q_j}\\
       &= \sum_{\sum jq_j=k} \frac{k!}{q_1!q_2!\cdots q_k!}e^{-\zeta y^{-1}}(-1)^{k+\sum q_j}\zeta^{\sum q_j} y^{-k-\sum q_j}.
   \end{align*}
   Using the substitution $z = \zeta y^{-1}$ and the fact that the 
   function $e^{-z}z^{m}$ achieves its maximum at $z = m$, 
   we have 
   \begin{align*}
       e^{-\zeta y^{-1}}\zeta^{\sum q_j} y^{-k-\sum q_j} 
       &= \zeta^{-k} e^{-z}z^{k+\sum q_j}  \\
       &\leq \zeta^{-k} e^{-k-\sum q_j}\left(k+\sum q_j\right)^{k+\sum q_j} 
\\
       &\leq \left(\frac{4}{e\zeta}\right)^k (k^k)^2 \\
       & \leq \left(\frac{4e}{\zeta}\right)^k (k!)^2
   \end{align*}
   where the last step is due to the inequality $k^k \leq e^k k!$. 
   Then we have
   \begin{align*}
       |g^{(k)}(s)| 
       &\leq \left(\frac{4e}{\zeta}\right)^k (k!)^2 \sum_{\sum jq_j=k} \frac{k!}{q_1!q_2!\cdots q_k!}.
   \end{align*}
   Finally, by directly loosing $\frac{k!}{q_1!q_2!\cdots q_k!}$ to $k!$ and noticing that the number of tuples $(q_1,\cdots, q_k)$ satisfying $\sum jq_j = k$ is less than $(k+1)(k/2+1)\cdots(k/k+1) = \binom{2k}{k} < 2^{2k}$, 
   we have 
   \begin{align*}
       |g^{(k)}(s)| 
       &\leq \left(\frac{4e}{\zeta}\right)^k (k!)^2 2^{2k}k! = \left(\frac{16e}{\zeta}\right)^k (k!)^3.
   \end{align*}
   
\end{proof}


Now let us consider the convergence analysis of the Chebyshev polynomial
expansion of $g(y)$. We remark that the proof based on the contour integral
formulation is only possible if the function is complex analytic in a neighborhood of $[-1,1]$ in the complex plane, as shown in the proof of \cref{lem:quadrature_err}. This is not satisfied in the case of the inverse transform. Here we present another approach following e.g.~\cite[Section 5.7]{MasonHandscomb2003}. 
The proof of \cref{thm:Chebyshev} is given in \cref{sec:Chebyshev_convergence_proof}.

\begin{thm}\label{thm:Chebyshev}
    Let $g(y)\in C^{r+1}([-1,1])$, and $T_k(\cdot)$ be the $k$-th order 
    Chebyshev polynomial on $[-1,1]$. 
    Let 
    \begin{equation}
        c_k = \frac{2-\delta_{k0}}{\pi}\int_{-1}^1\frac{g(y)T_k(y)}{\sqrt{1-y^2}}\ud y, \quad k\ge 0.
    \end{equation}
    Then for $d \geq 1$ 
    \begin{equation}
        \left\|g(y) - \sum_{k=0}^d c_k T_k(y)\right\|_{\infty}    
        \leq 32\frac{8^{r} (r+1)!\|g^{(r+1)}\|_{\infty}}{d^r}.
    \end{equation}
\end{thm}


If the function $g(y)$ is smooth, we can choose $r$ to be arbitrarily 
large to obtain better convergence with respect to $d$. 
However, if the derivatives of $g$ also grows rapidly when $r$ increases, 
the error will eventually grow up again. 
Therefore we may choose an optimal $r$ to balance the increasing norm of 
higher order derivatives and better convergence order to obtain smallest error.
Under the assumption of Gevrey class, we have the following result. 
\begin{thm}\label{thm:Chebyshev_degree}
    Assume $g(y) \in G^{\sigma}([-1,1])$ for some $\sigma \geq 0$. Then 
    for any $\epsilon>0$, 
    to achieve an $\epsilon$-approximation of $g(y)$, \emph{i.e. } 
    $\|g(y) - \sum_{k=0}^d c_k T_k(y)\|_{\infty} \leq \epsilon $, 
    it suffices to choose 
    \[
    d \geq  8eR\left(\log\left(32CR/\epsilon\right)+2\right)^{\sigma+1}.
    \]
\end{thm}
\begin{proof}
    By Theorem~\ref{thm:Chebyshev} and the definition of Gevrey class, 
    \begin{equation}
        \left\|g(y) - \sum_{k=0}^d c_k T_k(y)\right\|_{\infty}    
        \leq 32C \frac{8^{r}R^{r+1}[(r+1)!]^{\sigma+1}}{d^r}
    \end{equation}
    holds for all $r \geq 1$. 
    To achieve an $\epsilon$-approximation, it then suffices to choose $m$ 
    such that 
    \begin{equation}
        32C\frac{8^{r}R^{r+1}[(r+1)!]^{\sigma+1}}{d^r} \leq \epsilon, 
    \end{equation}
    which is equivalently
    \begin{equation}
        d \geq \frac{8(32C)^{\frac{1}{r}}R^{1+\frac{1}{r}}[(r+1)!]^{\frac{\sigma+1}{r}}}{\epsilon^{\frac{1}{r}}}.
    \end{equation}
    Using the estimate $(r+1)! \leq (r+1)^r$, it 
    suffices to choose 
    \begin{equation}
        d \geq 8R \left(\frac{32CR}{\epsilon}\right)^{\frac{1}{r}} (r+1)^{\sigma+1}.
    \end{equation}
    Now choose $r = \lceil\log(32CR/\epsilon)\rceil$, then the sufficient 
    condition for $d$ becomes
    \begin{equation}
        d \geq 8eR\left(\log\left(\frac{32CR}{\epsilon}\right)+2\right)^{\sigma+1}.
    \end{equation}
    We complete the proof. 
\end{proof}

Finally, to implement the truncated Chebyshev series $\sum_{k=0}^d c_kT_k(A)$ for some Hamiltonian $A$, we may use QSVT  (see \cref{sec:qsp_qlsp}). 
Let us now consider the complexity to construct a block-encoding for
$\frac{1}{2}e^{-\beta(A+B)}$ for $(A+B)\succ 0$.  The reason for adding a sub-normalization constant $1/2$ is to ensure that the corresponding Chebyshev polynomial is bounded 
by $1$, which allows to use QSVT approach for block-encoding. 

\begin{thm}
\label{thm:matrix_exp_inverse_transform}
    {Let $\wt{\sigma}_{\min}$ be a lower bound for the smallest singular value of $I+A^{-1}B$.}
    Let $U'_A$ be an $(\alpha_A',m_A',0)$-block-encoding of $A^{-1}$, 
    $U_B$ be an $(\alpha_B,m_B,0)$-block-encoding of $B$.  Then 
    for any $\beta > 0$ and $0 < \epsilon < 1$, 
    a $(1 ,m_{\mathrm{QSVT}},\epsilon)$-block-encoding can be constructed 
    for $\frac{1}{2}e^{-\beta H}$ with $H = A+B \succ 0$, where 
    $$\quad m_{\mathrm{QSVT}} = 2m_A'+m_B+4,$$
    with the following costs: 
    \begin{enumerate}
    \item 
        {$\wt{\Or}\left(\frac{\alpha_A'^2\alpha_B }{\beta\wt{\sigma}^2_{\min}}\left[\log\left(\frac{1}{\epsilon}\right)\right]^5\right)$}
    applications of $U'_A$, $U_B$, their controlled 
    versions, and their inverses,
    \item $(n+2m_A'+m_B+4)$ qubits, 
    \item 
        {$\Or\left((2m_A'+m_B+4)\frac{\alpha_A'}{\beta\wt{\sigma}_{\min}}\left[\log\left(\frac{1}{\epsilon}\right)\right]^4\right)$}
    additional primitive quantum gates due to the QSVT approach.
    \end{enumerate}
    {Furthermore $\wt{\sigma}_{\min}=\Omega(1/(1+\|(A+B)^{-1}\|\|B\|))$.}
\end{thm}
\begin{proof}
    By \cref{thm:precond_linear_sys}, for any $\epsilon' > 0$ 
    (to be specified later), we can construct an 
    $(\alpha_{A+B}', m_{A+B}', \epsilon')$-block-encoding of $(A+B)^{-1}$, 
    with {$\alpha_{A+B}' = \frac{4{\alpha'_A}}{3\wt{\sigma}_{\min}}, m_{A+B}' = 2m'_A+m_B+3$}, using
    {$\Or(\frac{\alpha'_A\alpha_B}{\wt{\sigma}_{\min}}\log(\frac{{\alpha'_A}}{\epsilon'\wt{\sigma}_{\min}}))$} 
    applications of $U'_A$, $U_B$, their controlled versions, their inverses, and other primitive gates. 
    
    Consider the function $g(y) = \frac{1}{2}\operatorname{sign}(y)e^{-\beta \alpha_{A+B}'^{-1} |y|^{-1}}$. 
    By \cref{thm:Chebyshev_degree} and \cref{prop:Chebyshev_gevrey_coeff} (here $C = 1/2$ and $R = 16e\beta^{-1}\alpha_{A+B}'$),
    there exists a $d$-degree odd polynomial $P(y)$ with 
    $d = \Theta(\alpha_{A+B}'\beta^{-1}[\log(\alpha_{A+B}'\beta^{-1}\epsilon^{-1})]^4)$, such that 
    $$\left\|\frac{1}{2}e^{-\beta(A+B)}-P\left(\frac{(A+B)^{-1}}{\alpha_{A+B}'}\right)\right\| = \left\|g\left(\frac{(A+B)^{-1}}{\alpha_{A+B}'}\right)-P\left(\frac{(A+B)^{-1}}{\alpha_{A+B}'}\right)\right\|\leq \frac{\epsilon}{2}.$$ 
    Note that $\left\|g\left(\frac{(A+B)^{-1}}{\alpha_{A+B}'}\right)\right\| \leq \frac{1}{2}$. 
    Therefore for any $\epsilon < 1$, we always have
    $\left\|P\left(\frac{(A+B)^{-1}}{\alpha_{A+B}'}\right)\right\| < 1$. 
    By QSVT and the block-encoding of $(A+B)^{-1}$, 
    we can construct a $(1,m_{A+B}'+1,2d\epsilon')$-block-encoding of $P\left(\frac{(A+B)^{-1}}{\alpha_{A+B}'}\right)$, using 
    $d$ queries of the block-encoding of $(A+B)^{-1}$ and its inverse, 
    and $\Or((m_{A+B}'+1)d)$ additional primitive gates. 
    To control the total error by $\epsilon$ from above, 
    it suffices to choose $\epsilon' = \epsilon/(4d)$. 
    Plug this back into the complexity analysis and 
    multiply them in all the steps, we obtain the estimates 
    of the costs as stated in the theorem.
\end{proof}

Finally, let us consider applying the matrix function $e^{-\beta H}$ to a quantum state $\ket{b}$, \emph{i.e. }preparing a state $e^{-\beta H}\ket{b}/\xi$ with $\xi=\|e^{-\beta H}\ket{b}\|$. Similar to the discussion in \cref{sec:contour_int}, the number of queries to $V$, $O_D$, $U_B$, and their inverses scales
{$\wt{\Or}\left(\frac{\alpha_A'^2\alpha_B}{\beta \xi \wt{\sigma}^2_{\min}}\left[\log\left(\frac{1}{\epsilon}\right)\right]^5\right)$,}
and the number of queries to $U_b$ and its inverse becomes  $\Or\left(\frac{1}{\xi}\log\left(\frac{1}{\epsilon}\right)\right)$.

\subsection{Purified Gibbs state}\label{sec:purified_gibbs}

The Gibbs state is a mixed state whose density matrix is
\[
\frac{1}{Z_\beta}e^{-\beta H},
\]
where $Z_\beta=\Tr(e^{-\beta H})$ is the partition function. In this section we assume $H=A+B$ with the $A$ and $B$ accessed through the oracles outlined at the beginning of Section~\ref{sec:fast_algs_matrix_functions}.

The Gibbs state can be constructed as a partial trace of a pure state, called the purified Gibbs state: 
\begin{equation}
    \label{eq:purified_gibbs}
    \ket{\Psi}=\frac{1}{\sqrt{Z_\beta}}\sum_{x\in[N]}\ket{x}(e^{-\beta H/2}\ket{x}).
\end{equation}
Here $x$ enumerates the elements of the $n$-qubit computational basis.
When we trace out the first register in the density matrix $\ket{\Psi}\bra{\Psi}$, we will obtain the Gibbs state.  These states are important for any algorithm that seeks to use thermal states and gain advantages from amplitude amplification, such as recent results on semi-definite programming as well as results on training Boltzmann machines~\cite{vanApeldoorn2020quantum, wiebe2016quantum,kieferova2017tomography}. We may rewrite $\ket{\Psi}$ as
\[
\ket{\Psi} = \sqrt{\frac{N}{Z_\beta}}(I\otimes e^{-\beta H/2})\left(\frac{1}{\sqrt{N}}\sum_{x\in[N]} \ket{x}\ket{x}\right).
\]
Therefore we can prepare the purified Gibbs state by applying the matrix function $I\otimes e^{-\beta H/2}$ to the maximally entangled state $(1/\sqrt{N})\sum_x\ket{x}\ket{x}$ \cite{poulin2009sampling}, which in turn can be prepared by applying a Hadamard transform and a series of CNOT gates similar to how one prepares the EPR pair \cite{NielsenChuang2000}. 




From the discussion at the end of \cref{sec:contour_int} and \cref{sec:inverse_transform}, since we have $\xi=\sqrt{Z_\beta/N}$, the query complexity of preparing the purified Gibbs state is as follows: 
\begin{cor}
In order to prepare the purified Gibbs state, the total number of queries to $U'_A$, $O_D$,  $U_B$ and their inverses is 
{$\wt{\Or}((\alpha_B/ \sqrt{\beta}\wt{\sigma}'^2_{\min})\log(1/\epsilon)\sqrt{N/Z_\beta})$ }
in the contour integral approach, and \\
{$\wt{\Or}\left((\alpha_A'^2\alpha_B/\beta\wt{\sigma}^2_{\min})\left[\log\left(1/\epsilon\right)\right]^5\sqrt{N/Z_{\beta}}\right)$ }
in the inverse transform approach, {where $\wt{\sigma}'_{\min}$ satisfies Eq.~\eqref{eq:condition_sigma_prime} and $\wt{\sigma}_{\min}$ is a lower bound for the smallest singular value of $I+A^{-1}B$.}
\end{cor}

Note that if we only care about the case when $\beta = 1$ and $A+B \succ 0$, 
we can shift $A$ such that the assumptions 
$A \succeq I$ and $A+B \succeq I$ hold without 
introducing much overhead in the subnormalization constant, 
then {$\alpha_A' = \Or(1)$} 
and the query complexity becomes 
{$\wt{\Or}\left((\alpha_B/\wt{\sigma}^2_{\min})\left[\log\left(1/\epsilon\right)\right]^5\sqrt{N/Z_{\beta}}\right)$}, as shown in \cref{tab:compare_algs_thermal}. 


\section{Discussion}\label{sec:discuss}

We have a quantum primitive called fast inversion to solve a class of quantum linear systems problem (QLSP) $\ket{x}\propto A^{-1}\ket{b}$. If $A$ is a normal matrix and diagonalized as $A=VDV^{\dag}$, then fast inversion is applicable if 1) there is an efficient quantum circuit to implement the unitary matrix $V$ (for instance, when $V$ can be implemented via the quantum Fourier transform), and 2) the inverse of the diagonal matrix $D^{-1}$ can be efficiently implemented via a classical circuit. Here compared to the standard approach of first finding a block-encoding of $A$ and then invert $A$, the condition 2) is the key reason for fast inversion to achieve reduced circuit depth and query complexities. 

Using fast inversion, we developed a preconditioned linear system solver for solving a class of QLSP of the form $\ket{x}\propto (A+B)^{-1}\ket{b}$. Here we assume the matrix $A$ can be fast inverted, but $\norm{A}\gg\norm{B},\norm{A^{-1}},\norm{(A+B)^{-1}}$. Our main result is that the query complexity of the preconditioned linear system solver can be independent of $\norm{A}$, or the condition number $\kappa(A+B)$, and can therefore significantly reduce the computational cost.

We demonstrate an application of the preconditioned quantum linear system solver for computing the single-particle Green's function in quantum many-body physics. Although quantum linear system solver is often considered to be used as a subroutine of a larger quantum algorithm, the computation of Green's functions in quantum many-body physics is entirely a linear system problem in high dimensions. Hence we need to take into account all components of the algorithm, in particular the readout errors due to the Monte Carlo sampling. Using fast inversion, we again demonstrate that the query complexity can be independent of the norm of the total Hamiltonian.

Observe that $(A+B)^{-1}\ket{b}$ is only one example of a more general class of problems of computing $f(A+B)\ket{b}$, where $f(A+B)$ is a matrix function. To be specific, we consider the problem of thermal state preparation, where $f(H)=e^{-\beta H}$ and $H=A+B$ is a Hermitian matrix. We again assume that the main difficulty comes from that $\norm{H}\sim\norm{A}$ is very large.
We developed two methods to approximately compute the matrix function via a series of linear systems. The first method is based on Cauchy's contour integral formula, and the second is based on a simple inverse transform. Using fast inversion, both methods can be used to fast prepare a thermal state, and the cost is independent of $\norm{H}$. We remark that both the contour integral formula and the inverse transform are quite general, and can be used to accelerate the computation of other matrix functions as well.

We would like to remark that fast inversion is intimately connected to the fast forwarding of Hamiltonian simulation $e^{\I A \tau}$, where $A$ is Hermitian and $\tau\in\RR$ is some arbitrary time. For example, based on fast forwarding, the interaction picture Hamiltonian simulation algorithm \cite{LowWiebe2019} allows fast evaluation of $e^{\I (A+B) t}\ket{b}$. The assumption and the main result in \cite{LowWiebe2019} are both similar to this paper, i.e. when $\norm{A}\gg \norm{B}$ and $e^{\I A \tau}$ can be fast forwarded, the query complexity of the Hamiltonian simulation $e^{\I (A+B) t}\ket{b}$ can be independent of $\norm{A}$. It is interesting to compare fast inversion and fast forwarding from a numerical analysis perspective, i.e. whether one can use fast inversion to perform fast forwarding, and vice versa. More generally, when considering the fast evaluation of certain matrix functions $f(A+B)$, whether fast inversion or fast forwarding can lead to a more efficient algorithm. The answer will likely depend on the details of the function $f$ of interest, as well as the preconstants of different methods. This will be our future work.

\vspace{1em}
\textbf{Acknowledgment}
This work was partially supported by the Department of Energy under Grant No.
DE-AC02-05CH11231 (L.L. and Y.T.), No. DE-SC0017867 (D.A. and L.L.), the Google Quantum Research Award, Quantum Algorithms, Software, and Architectures (QUASAR) Agile Investment at Pacific Northwest National Laboratory, the Pacific Northwest Laboratory LDRD program (N.W.), and the National Science Foundation under the QLCI program through grant number OMA-2016245 (L.L. and N.W.).   

\bibliographystyle{abbrvnat}
\bibliography{ref}

\clearpage
\appendix

\section{{Quantum linear system solvers}}
\label{sec:QLS_solvers}

We first briefly review the literature for solving QLSPs of the form 
$
A\ket{x} \propto \ket{b}
$
without preconditioning. 
Here we consider a block-encoding model \cite{GilyenSuLowEtAl2019}, i.e. there exists a unitary matrix $U_{A}$ that encodes the matrix $A$ (see \cref{sec:notations} for the details of block-encoding), while the right hand side vector $\ket{b}$ is prepared by a unitary $U_b$ as $\ket{b}=U_b\ket{0^n}$. For simplicity we only compare the query complexity to the block-encoding $U_{A}$. Throughout the paper the notation $\wt{\Or}(f)$ means $\Or(f\polylog(f))$.

The query complexity of the original quantum linear systems algorithm by Harrow-Hassidim-Lloyd (HHL) algorithm~\cite{HarrowHassidimLloyd2009}\footnote{The original HHL algorithm was not formulated in the language of block-encoding. However, this does not affect the query complexity. We refer readers to \cref{sec:query_complexities_HHL_LCU_Greens} for more details.} scales as $\widetilde{\Or}(\kappa^2/\epsilon)$, where $\kappa=\kappa(A)$ and $\epsilon$ is the target accuracy. The HHL algorithm is based on the phase estimation method. In the past few years, there have been significant progresses towards reducing the query complexity for quantum linear solvers. In particular, the linear combination of unitaries (LCU) \cite{BerryChildsKothari2015,ChildsKothariSomma2017} and quantum singular value transformation (QSVT) \cite{GilyenSuLowEtAl2018,GilyenSuLowEtAl2019} (closely related to quantum signal processing (QSP) \cite{LowYoderChuang2016,LowChuang2017}) techniques can reduce the query complexity to $\Or(\kappa^2 \poly\log(\kappa/\epsilon))$. It is worth commenting that with respect to the condition number, the worst case complexity $\wt{\Or}(\kappa^2)$ is very much inherent to all aforementioned algorithms (HHL, LCU, QSVT). This is because the cost of each algorithm is $\wt{\Or}(\kappa)$ per run, and the worst case success probability of each run is $\wt{\Or}(\kappa^{-2})$. Hence in order to boost to $\Omega(1)$ success probability, the cost of the naive application of each algorithm is $\wt{\Or}(\kappa^3)$. When the standard amplitude amplification \cite{BrassardHoyerMoscaEtAl2002} is used, one can reduce the number of repetitions from $\wt{\Or}(\kappa^2)$ to $\wt{\Or}(\kappa)$, and hence the complexity is reduced to $\wt{\Or}(\kappa^2)$. Furthermore, while the circuit depth of all methods above is independent of the right hand side $\ket{b}$, the number of repetitions can depend on $\ket{b}$. Hence the total number of queries can scale better than $\wt{\Or}(\kappa^2)$ (see \cref{sec:qsp_qlsp} for detailed discussion in the context of QSVT).

In order to further reduce the worst case complexity with respect to $\kappa$, it is possible to use a technique called variable-time amplitude amplification (VTAA) \cite{Ambainis2012}, which is a generalization of the standard amplitude amplification algorithm and allows to amplify the success probability of quantum algorithms by stopping different branches stop at different times. In \cite{Ambainis2012},  VTAA was first used to successfully improve the dependence of the HHL algorithm on $\kappa$ to be almost linear, and the query complexity is $\wt{\Or}(\kappa/\epsilon^3)$. In~\cite{ChildsKothariSomma2017}, VTAA was combined with a low-precision phase estimate to improve the complexity of LCU to $\wt{\Or}(\kappa\polylog(1/\epsilon)))$, which is near-optimal with respect to both $\kappa$ and $\epsilon$. A similar strategy may be applied to accelerate QSVT. It is worth noting that the VTAA algorithm is a complicated procedure and can be difficult to implement. 

Inspired by adiabatic quantum computation (AQC) \cite{FarhiGoldstoneGutmannEtAl2000,AlbashLidar2018,JansenRuskaiSeiler2007}, the recently developed randomization method (RM) \cite{SubasiSommaOrsucci2019} can solve QLSP with a runtime complexity $\Or(\kappa \log(\kappa)/\epsilon)$, which is the first algorithm to yield $\wt{\Or}(\kappa)$ complexity without using VTAA. One can use an optimal Hamiltonian simulation method (e.g. \cite{LowChuang2017}) to yield a gate-based implementation of RM, of which the query complexity becomes $\wt{\Or}(\kappa/\epsilon)$. This is significantly better than the bounds proven using the vanilla AQC method, of which the complexity is $\wt{\Or}(\kappa^2/\epsilon)$ \cite{JansenRuskaiSeiler2007,SubasiSommaOrsucci2019}. Using a fast eigenpath traversal method \cite{BoixoKnillSomma2009}, which relies on repeated usage of phase estimation, the cost of RM can be further reduced to be near-optimal as $\Or(\kappa \polylog(\kappa/\epsilon))$. 

Since repeated usage of phase estimation or VTAA are both difficult to implement and can require a significant number of ancilla qubits, it remains of great interest to obtain alternative algorithms to solve QLSP with near-optimal complexity scaling without resorting to such procedures. The first algorithm along this line is achieved by an optimally scheduled AQC algorithm \cite{AnLin2019}, called AQC(exp), and the runtime complexity is $\Or(\kappa \polylog(\kappa/\epsilon))$. There is also a slightly more versatile class of methods called AQC(p), of which the runtime complexity is simply $\Or(\kappa / \epsilon)$
\cite{AnLin2019}. Both AQC(exp) and AQC(p) require the implementation of a time-dependent simulation procedure. Using the recently developed near-optimal method for time-dependent Hamiltonian simulation \cite{LowWiebe2019}, the query complexity of AQC(exp) is also $\Or(\kappa \polylog(\kappa/\epsilon))$. This immediately implies that the optimal runtime complexity of the quantum approximate optimization algorithm (QAOA) \cite{FarhiGoldstoneGutmann2014} for solving QLSP is also $\Or(\kappa \polylog(\kappa/\epsilon))$, as is verified by numerical experiments \cite{AnLin2019}. Using a different strategy called quantum eigenstate filtering \cite{LinTong2020}, one can boost any algorithm that prepares a solution to $\ket{x}$ with $\Or(1)$ accuracy $\Or(\epsilon)$ with $\Or(\kappa \polylog(\kappa/\epsilon))$
complexity. This is a very simple procedure and has a fully gate-based implementation via QSVT. In particular, using a method inspired by the quantum Zeno effect, \cite{LinTong2020} obtains a simple algorithm to solve QLSP with $\Or(\kappa \polylog(\kappa/\epsilon))$ complexity, without using any amplitude amplification, phase estimation, or Hamiltonian simulation (time-independent or time-dependent). We remark that for both AQC(exp) and quantum eigenstate filtering, the result is a pure state, and the success probability of a single run is already $\Omega(1)$. This avoids the problem of repeated measurements inherent to HHL, LCU and QSVT.

\section{Quantum singular value transformation, and its application to quantum linear system problem}\label{sec:qsp_qlsp}


For a square matrix $A \in \mathbb{C}^{N \times N}$, where for simplicity we assume $N=2^n$ for some positive integer $n$, the singular value decomposition (SVD) of the normalized matrix $A$ can be written as
\begin{equation}
A=W\Sigma V^{\dag},
\label{eqn:A_SVD}
\end{equation}
or equivalently
\begin{equation}
A\left|v_{k}\right\rangle= \sigma_{k}\left|w_{k}\right\rangle, \quad A^{\dagger}\left|w_{k}\right\rangle= \sigma_{k}\left|v_{k}\right\rangle, \quad k\in [N].
\end{equation}
We may apply a function $f(\cdot)$ on its singular values and define the generalized matrix function \cite{HawkinsBen-Israel1973,ArrigoBenziFenu2016} as below.

\begin{defn}[Generalized matrix function {\cite[Definition 4]{ArrigoBenziFenu2016}}]
\label{def:gen_matrix_function}
 Let $f: \mathbb{R} \rightarrow \mathbb{R}$ be a scalar function such that $f\left(\sigma_{i}\right)$ is defined for all $i=1,2, \ldots, N$. The generalized matrix function is defined as
\begin{equation}
f^{\diamond}(A):=Wf(\Sigma)V^{\dag},
\label{eqn:gen_matrix_function}
\end{equation}
where 
$$
f\left(\Sigma\right)=\operatorname{diag}\left(f\left(\sigma_{1}\right), f\left(\sigma_{2}\right), \ldots, f\left(\sigma_{N}\right)\right).
$$
\end{defn}
This implies that the singular values satisfy
\[
0\le \sigma_k \le 1, \quad k\in[N].
\]
In the discussion below, we assume $f$ is a \textit{real, odd polynomial of finite degree}, and satisfies
\begin{equation}
|f(x)|\leq 1, \quad \forall x\in[-1,1].
\label{eqn:f_condition}
\end{equation}
We assume there is an $(\alpha,m,0)$-block-encoding of $A$ denoted by $U_A$, so that the singular values of $A/\alpha$ are in $[0,1]$. The quantum singular value transformation (QSVT) \cite[Theorem 2]{GilyenSuLowEtAl2019} provides an efficient way to implement $f^{\diamond}(A/\alpha)$ on a quantum computer using a very simple circuit, shown in \cref{fig:qsp_circuit_real}. It only uses one extra qubit. 
Here H, Z are the Hadamard and Pauli-Z gates, respectively. The polynomial degree is assumed to be $2d+1$ for $d\in \NN$. The real numbers $\{\varphi_{i}\}_{i=0}^{2d+1}$ are called the phase factors of the QSVT circuit, and the block-encoding $U_A$ and its adjoint $U_A^{\dag}$ appear alternatively. For a given polynomial $f$, the phase factors is an effective way to encode the polynomial in $\mathrm{SU}(2)$ \cite[Corollary 11]{GilyenSuLowEtAl2019}. Although the computation of phase factors can be entirely carried out on classical computers, it is a non-trivial task to compute these phase factors. Significant progress has been achieved recently, enabling robust computation of phase factors for polynomials of degrees ranging from thousands to tens of thousands \cite{Haah2019,chao2020finding,DongMengWhaleyLin2020}.

\begin{figure}[htb]
\begin{center}
\includegraphics[width=1.0\textwidth]{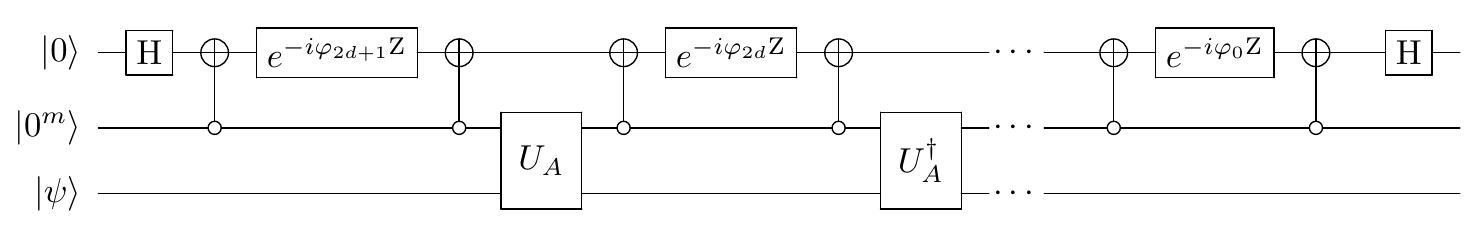}
\end{center}
\caption{Quantum circuit for implementing the quantum singular value transformation $f^{\diamond}(A/\alpha)$, where $f$ is a real, odd polynomial of degree $2d+1$ satisfying \cref{eqn:f_condition}. }
\label{fig:qsp_circuit_real}
\end{figure}

\begin{rem}\label{rem:qsvt}The condition on $f$ in QSVT appears to be much stronger than that in \cref{def:gen_matrix_function}. Indeed, when $f$ is a real, even polynomial, the counterpart of \cref{fig:qsp_circuit_real} implements a type of generalized matrix functions taking a different form from \cref{eqn:gen_matrix_function}. Since the solution of QLSP only uses the formulation with odd polynomials, we refer interested readers to \cite{GilyenSuLowEtAl2019} for the construction of QSVT for real, even polynomials, as well as more general complex valued polynomials. 
\end{rem}

If $f$ is a general odd function, we may first approximate $f(x)$ on the interval $[-1,1]$ using a degree-$d$ odd polynomial (we may apply a scaling factor if needed to satisfy \cref{eqn:f_condition}), and apply QSVT to the resulting polynomial.


The matrix inversion can be implemented as a quantum singular value transformation in the following way: when $A$ is invertible $A^{-1}=V \Sigma^{-1} W^\dagger$. Therefore 
\begin{equation}
(A/\alpha)^{-1}=(f^{\diamond}(A/\alpha))^\dagger
\label{eqn:ainv_qsvt}
\end{equation}
where $f(x)=x^{-1}$. However $f(\cdot)$ here is not bounded by 1 and is in fact singular at $x=0$. Therefore instead of approximating $f(x)=x^{-1}$ on the whole interval $[-1,1]$ we consider an odd polynomial $p(x)$ such that
\[
\left|p(x)-\frac{3\delta}{4x}\right|\leq \epsilon',\quad \forall x\in [-1,-\delta]\cup[\delta,1].
\]
and $|p(x)|\leq 1$ for all $x\in[-1,1]$. The existence of such an odd polynomial of degree $\Or(\frac{1}{\delta}\log(\frac{1}{\epsilon'}))$ is guaranteed by \cite[Corollary 69]{GilyenSuLowEtAl2018}.

Then \cite[Theorem 2]{GilyenSuLowEtAl2019} enables us to implement $(p^{\diamond}(A/\alpha))^\dagger=V p(\Sigma/\alpha)W^\dagger$. We have
\begin{equation}
\label{eq:block_encoding_err_A_inv}
    \|(p^{\diamond}(A/\alpha))^\dagger-(3\delta/4)(A/\alpha)^{-1}\|=\|p(\Sigma/\alpha)-(3\delta/4)(\Sigma/\alpha)^{-1}\|\leq \epsilon',
\end{equation}
if all diagonal elements of $\Sigma/\alpha$, i.e. the singular values of $A/\alpha$, are in the interval $[\delta,1]$. Therefore we want all singular values of $A/\alpha$ to be at least $\delta$ distance away from the origin. We then use QSVT to block-encode $(p^{\diamond}(A/\alpha))^\dagger$ given a block-encoding of $A$.

We assume $A$ can be accessed by its $(\alpha,m,0)$-block-encoding $U_A$. Let $\kappa$ be the condition number of $A$, then the singular values of $A/\alpha$ are contained in $[\norm{A}/(\alpha\kappa),\norm{A}/\alpha]$. Therefore we choose $\delta = \norm{A}/(\alpha\kappa)$. Using QSVT, a $(1,m+1,0)$-block-encoding of $p^{\diamond}(A/\alpha)$ can be implemented \cite[Theorem 2]{GilyenSuLowEtAl2019}. We denote this block-encoding by $\mathcal{U}$. Then by Eq.~\eqref{eq:block_encoding_err_A_inv}
\[
\left\|\frac{4}{3\delta\alpha}(\bra{0^{m+1}\otimes I})\mathcal{U}^\dagger(\ket{0^{m+1}\otimes I})-A^{-1}\right\|
=\left\|\frac{4}{3\delta\alpha}(p^{\diamond}(A/\alpha)^\dagger)-A^{-1}\right\|
\leq \frac{4\epsilon'}{3\delta\alpha}. 
\]
Therefore $\mathcal{U}^\dagger$ is a  $(4/(3\delta\alpha),m+1,4\epsilon'/(3\delta\alpha))$-block-encoding of $A^{-1}$. Because the cost of QSVT scales linearly with respect to the degree of the polynomial $p(x)$, the total number of queries to to $U_A$ and its inverse is
\[
\Or\left(\frac{1}{\delta}\log\left(\frac{1}{\epsilon'}\right)\right)
=\Or\left(\frac{\alpha\kappa}{\|A\|}\log\left(\frac{1}{\epsilon'}\right)\right).
\]

With the block-encoding $\mathcal{U}^\dagger$ of matrix inversion we are then able to solve the linear system $A\ket{x}\propto\ket{b}$ when we are given the quantum state $\ket{b}$  representing the right-hand side of the QLSP. We assume $\ket{b}$ can be accessed through the oracle $U_b$ such that
\[
U_b\ket{0^n}=\ket{b}.
\]
We introduce the parameter 
\[
\xi=\norm{A^{-1}\ket{b}},
\]
which plays an important part in the success probability of the procedure. This parameter also appear in Sections~\ref{sec:fastinv_1sparse}, \ref{sec:fastinv_diagaonlize}, and Corollary~\ref{cor:precond_linear_sys_solver}.
Let $\ket{\wt{x}}=(3\delta/4)(A/\alpha)^{-1}\ket{b}$, then the normalized solution state is $\ket{x}=\ket{\wt{x}}/\|\ket{\wt{x}}\|$ satisfying $A\ket{x} \propto \ket{b}$. Now denote $\ket{\wt{y}}=(p^{\diamond}(A/\alpha))^{\dagger}\ket{b}$, and $\ket{y}=\ket{\wt{y}}/\|\ket{\wt{y}}\|$. If the polynomial approximation has error $\epsilon'$, then we have $\|\ket{\wt{x}}-\ket{\wt{y}}\|\leq\epsilon'$. However for the normalized state $\ket{y}$, this error is scaled accordingly. When $\epsilon'\ll \|\ket{\wt{x}}\|$, we have
\begin{equation}
\|\ket{x}-\ket{y}\|\approx \frac{\|\ket{\wt{x}}-\ket{\wt{y}}\|}{\|\ket{\wt{x}}\|}\leq \frac{\epsilon'}{\|\ket{\wt{x}}\|}.
\end{equation}
Also we have
\[
\|\ket{\wt{x}}\| = \left\|\frac{3\delta}{4}\left(\frac{A}{\alpha}\right)^{-1}\ket{b}\right\| = \frac{3\alpha\delta\xi}{4} \geq \frac{3\norm{A}\xi}{4\kappa}.
\]
Therefore in order for the normalized output quantum state to be $\epsilon$-close to the normalized solution state $\ket{x}$, we need to set $\epsilon'=\Or(\epsilon\norm{A} \xi /\kappa)$.

The success probability of the above procedure is $\Omega(\|\ket{\wt{x}}\|^2)=\Omega(\norm{A}^2\xi^2/\kappa^2)$. With amplitude amplification we can boost the success probability to be greater than $1/2$ with one qubit serving as a witness, i.e. if measuring this qubit we get an outcome 0 it means the procedure has succeeded, and if 1 it  means the procedure has failed. 
It takes $\Or(\kappa/\norm{A}\xi)$ rounds of amplitude amplification, i.e. using $\mathcal{U}^\dagger$, $\mathcal{U}$, $U_b$, and $U_b^\dagger$ for $\Or(\kappa/\norm{A}\xi)$ times.
A single $\mathcal{U}$ or its inverse uses $U_A$ and its inverse
\[
\Or\left(\frac{1}{\delta}\log\left(\frac{1}{\epsilon'}\right)\right)
=\Or\left(\frac{\alpha\kappa}{\|A\|}\log\left(\frac{\kappa}{\epsilon\norm{A}\xi}\right)\right)
\]
times.
Therefore the total number of queries to $U_A$ and its inverse is 
\[
\Or\left(\frac{\kappa}{\|A\|\xi}\times \frac{\alpha\kappa}{\|A\|}\log\left(\frac{\kappa}{\epsilon\norm{A}\xi}\right)\right)=\Or\left(\frac{\alpha\kappa^2}{\|A\|^2\xi} \log\left(\frac{\kappa}{\|A\|\xi\epsilon}\right)\right).
\]
The number of queries to the $U_b$ and its inverse is $\Or(\kappa/\|A\|\xi)$. 
To summarize, we have the refined version of using QSVT to solve QLSP:

\begin{thm}[Standard QSVT linear system solver]
\label{thm:qsvt_qlsp}
Let $U_A$ be an $({\alpha},m,0)$-block-encoding of $A$ with condition number $\kappa$, $U_b$ be the oracle preparing the right hand side vector $\ket{b}$, and $\xi=\norm{A^{-1}\ket{b}}\in[1/\|A\|,\|A^{-1}\|]$.
Then the solution $\ket{x}\propto A^{-1}\ket{b}$ can be obtained with precision $\epsilon$ and with a success probability at least $1/2$, 
using $\Or\left(\alpha\kappa^2/(\|A\|^2\xi) \log\left(\kappa/(\|A\|\xi\epsilon)\right)\right)$ queries to $U_A$ and $U_A^{\dag}$, and $\Or(\kappa/(\|A\|\xi))$ queries to $U_b$. 
\end{thm}


{
We consider the following two cases for the magnitude of $\xi$. For simplicity we assume $\alpha=\Theta(\norm{A})$. 
\begin{enumerate}\label{enum:two_cases}
    \item  In general if no further promise is given, then $\xi\geq 1/\|A\|$. The total query complexity of $U_A$ is therefore $\Or(\kappa^2 \log(\kappa/\epsilon))$. This is the typical complexity referred to in the literature.
    \item If $\ket{b}$ has a $\Omega(1)$ overlap with the left-singular vector of $A$ with the smallest singular value, then $\xi=\Omega(\norm{A^{-1}})=\Omega(\kappa/\|A\|)$. This is the best case scenario, and the total query complexity of $U_A$ is  $\Or(\kappa \log(1/\epsilon))$, and the number of queries to the right hand side vector $\ket{b}$ is $\Or(1)$, which is independent of the condition number. 
\end{enumerate} 
}

\section{Fast inversion of 1-sparse matrices}\label{sec:fastinv_1sparse_general}


In this section we discuss the fast inversion of 1-sparse matrices, i.e.  there is at most one non-zero matrix element in every row and column of the matrix.  First, let us assume that there exists an efficient function $f$ such that $f(x)$ yields the column number for the non-zero matrix element of $A$ in row $x$.  We then define the oracle $O_{f}: \ket{x} \ket{c} \mapsto \ket{x} \ket{c\oplus f(x)}$  as well as the oracle $O_A: \ket{x} \ket{y} \ket{c} \mapsto \ket{x} \ket{y} \ket{c\oplus A_{xy}}$. Without loss of generality we assume $A$ is a Hermitian matrix (otherwise we can dilate it into a Hermitian matrix as in \cref{eqn:dilation_trick}. In this case we also need access to the oracle $O_{g}: \ket{x} \ket{c} \mapsto \ket{x} \ket{c\oplus g(x)}$, where $g(x)$ yields the row  number for the non-zero matrix element of $A$ in column $x$).  For convenience, we will assume that the matrix elements of $A$ are encoded in polar form.

\begin{lem}
Let $A$ be a $1$-sparse Hermitian, invertible and row-computable matrix in $\mathbb{C}^{2^n \times 2^n}$ where each $A_{xy}$ can be exactly represented using $n'$ bits of precision then for any $\delta>0$, an $\left(\left(\min_{x} |A_{xf(x)}|\right)^{-1}, 1, \delta\right)$-block encoding of $A^{-1}$ can be implemented using at most $4$ oracle queries and an $O({\rm poly}(n'\log(1/\delta)))$ auxiliary operations taken from a gate set consisting of $\mathrm{H,T,CNOT}$. \end{lem}
\begin{proof}
The proof follows standard reasoning provided in the quantum simulation literature \cite{childs2003exponential,ahokas2004improved}.  The first observation that is important in the reasoning is that all Hermitian $1$-sparse matrices can be written as a direct sum of irreducible one- and two- dimensional matrices.  This direct sum structure can be exploited to allow the matrix inverse to be explicitly computed on each of these subspaces of constant dimension and thereby the desired transformation can be performed to invert the matrix.

Consider a fixed basis vector as the input denoted by $\ket{x}$.  We need to consider two cases.  First, let $\ket{x}$ be part of an irreducible two-dimensional block of $A$.  This means that $f(x) \ne x$.  It then follows that \begin{equation}
    A(a\ket{x} + b\ket{f(x)})  = aA_{f(x)x}\ket{f(x)} + bA_{xf(x)}\ket{x}= |A_{xf(x)}|\left(\frac{aA^*_{xf(x)}}{|A_{xf(x)}|}\ket{f(x)} + \frac{bA_{xf(x)}}{|A_{xf(x)}|}\ket{x}\right).
\end{equation}
This justifies the claim that $\ket{x}$ is in an irreducible two-dimensional space.  Additionally, we can see by applying $A$ again to this result and using that $A$ is Hermitian, the eigenvalues of $A$ within the two-dimensional sub-space are $\pm |A_{xf(x)}|$.  Similarly, the eigenvectors can be taken to be

\begin{equation}
    \ket{\psi^{\pm}_x}:=\frac{1}{\sqrt{2}} \left(\ket{\min(x,f(x))} \pm \frac{A_{xf(x)}^*}{A_{xf(x)}} \ket{\max(x,f(x))} \right),\label{eq:1sparseeigen}
    \end{equation}
where the eigenvectors corresponding to the positive and negative eigenvalues are $\ket{\psi_{x}^+}$ and $\ket{\psi_x^-}$ respectively.

In order to perform this eigen-decomposition of the input vectors, we will map this two-dimensional space to a single qubit space and then diagonalize the operator once it has been reduced to this case.  There are several operations that we need to define for this diagonalization process.  First, let us define ${\rm CMP}$ to be a comparator circuit.  This self-inverse unitary transformation has the action (for any $b\in \mathbb{Z}_2)$
\begin{equation}
    {\rm CMP} : \ket{x} \ket{y} \ket{b} \mapsto \begin{cases}
    \ket{x} \ket{y} \ket{b} &, \text{ if $x \le y$,}\\
    \ket{x} \ket{y} \ket{b\oplus 1} &, \text{ if $x > y$.}
    \end{cases}
\end{equation}
The ${\rm CMP}$ circuit can be implemented using a two's complement adder~\cite{sanders2020compilation} using a polynomial number of gates in the bits of precision of the inputs $x$ and $y$.
The next operation we need is the controlled swap gate ${\rm CSWAP}$ which we define such that for any two states $\ket{\psi}$ and $\ket{\phi}$ of the same dimension with ${\rm SWAP} \ket{\psi} \ket{\phi} = \ket{\phi}\ket{\psi}$
\begin{equation}
{\rm CSWAP} = \ket{1}\!\bra{1}\otimes {\rm SWAP} + \ket{0}\!\bra{0}
\otimes I.
\end{equation}
The ${\rm CSWAP}$ operation can be implemented using a number of CNOT and Toffoli gates that is linear in the bits of $\ket{\psi}$ and $\ket{\phi}$.
We denote this transformation $T$ and implement it through the following steps under the assumption (without loss of generality) that $x\le f(x)$, 
\begin{align}
    (a \ket{x} + b\ket{f(x)})\ket{0^{n+1}} &\mapsto_{ O_f}\qquad (a \ket{x,f(x)} + b\ket{f(x),x})\ket{0^{n'+2}}\nonumber\\
    &\mapsto_{\mathrm{CMP}}\qquad (a \ket{x,f(x)}\ket{0} + b\ket{f(x),x} \ket{1})\ket{0^{n'+1}}\nonumber\\
    &\mapsto_{\mathrm{CSWAP}}\qquad \ket{x,f(x)}(a\ket{0} + b \ket{1})\ket{0^{n'+1}}
    \nonumber\\
    &\mapsto_{O_A}\qquad \ket{x,f(x)}(a\ket{0} + b \ket{1})\ket{0}\ket{A_{xf(x)}}
    \nonumber\\
    &\mapsto_{\mathrm{CNOT}}\qquad \ket{x,f(x)\oplus x}(a\ket{0} + b \ket{1})\ket{0}\ket{A_{xf(x)}}\nonumber\\
    &\mapsto_{\Lambda_{n}(\mathrm{NOT})}\qquad \ket{x,f(x)\oplus x}(a\ket{0} + b \ket{1})\ket{\delta_{x,f(x)}}\ket{A_{xf(x)}}.
    \label{eq:perm}
\end{align}
Here $\Lambda_n({\rm NOT})$ is an $n$-controlled not gate.  
This sequence of operations mirrors standard approaches for simulating one-sparse Hamiltonian dynamics as given in~\cite{aharonov2003adiabatic,berry2007efficient}.

All the information needed to implement the inverse of $A$ on each one- or two-dimensional irreducible subspaces of the one-sparse matrix is computed using a single invocation of $T$.  From that we invert the matrix using two controlled inversion circuits, $R_1$ and $R_2$ which correspond to the one- and two-dimensional inversion circuit respectively.  This strategy is shown schematically in Figure~\ref{fig:construct}.

\begin{figure}
\centering
\includegraphics[width=0.45\textwidth]{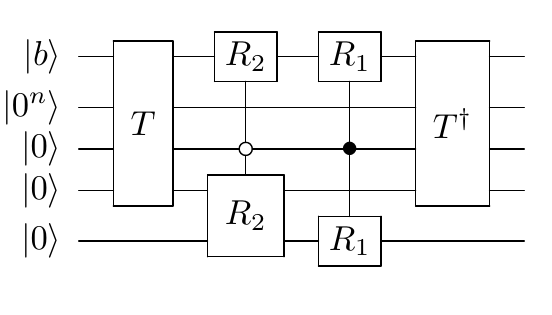}
\caption{Schematic quantum circuit for implementing the one- and two-dimensional transformations. $R_1$ and $R_2$ refer to the rotations implemented on the one- and two-dimensional subspaces respectively. \label{fig:construct}}
\end{figure}

The transformation $R_2$ in Figure~\ref{fig:construct} can be implemented as follows.  If $\ket{x}$ and $\ket{f(x)}$ form an irreducible two-dimensional subspace then we can identify $R_2$ through the following reasoning.  First, note that
\begin{align}
&{\frac{A}{|A_{xf(x)}|}\otimes I}{}(a\ket{x} + b\ket{f(x)})\ket{0^{n+n'+2}}  \nonumber\\
&\qquad=  
T^\dagger \left(I \otimes R_z\left( {\rm Arg}(A_{xf(x)})\right)XR_z\left( -{\rm Arg}(A_{xf(x)})\right)\otimes I\right) T    (a\ket{x} + b\ket{f(x)}) \ket{0^{n'+n+2}},\label{eq:fullstate}
\end{align}
here $X$ is the Pauli-$X$ gate.   Next, we see from the fact that if $C$ is a $2\times 2$ invertible Hermitian matrix then
\begin{equation}
    C = \begin{bmatrix} 0 & c \\ c^* &0 \end{bmatrix}\qquad \Rightarrow \qquad C^{-1}=\begin{bmatrix} 0 & (c^{-1})^* \\ c^{-1} &0 \end{bmatrix}.
\end{equation}
that
\begin{align}
&\frac{A^{-1}\otimes I}{|A_{xf(x)}|^{-1}}(a\ket{x} + b\ket{f(x)})\ket{0^{n+n'+2}}  \nonumber\\
&\qquad=  
P_2^\dagger \left(I \otimes R_z\left({\rm Arg}(A_{xf(x)}^*)\right)XR_z\left(-{\rm Arg}(A_{xf(x)}^*)\right)\otimes I\right) P_2    (a\ket{x} + b\ket{f(x)}) \ket{0^{n'+n+2}},\label{eq:fullstateinv}
\end{align}

Equation~\eqref{eq:fullstateinv} shows how we can block-encode the normalized inverse; however, since the value of $|A_{xf(x)}|$ need not be constant over $x$, we will need to show how to construct a block encoding that will allow the normalization factor to be varied across each block.  The central observation is that
\begin{align}
    &{A^{-1}\otimes I}(a\ket{x} + b\ket{f(x)})\ket{0^{n+n'+2}} \nonumber\\
    &= \left(I\otimes I \otimes \bra{0} \right)\Biggr(\frac{A^{-1}\otimes I}{|A_{xf(x)}|^{-1}}(a\ket{x} + b\ket{f(x)})\ket{0^{n+n'+2}}\nonumber\\
    &\qquad\qquad\qquad\qquad\otimes\left(\frac{\min_{x} |A_{xf(x)}|}{|A_{xf(x)}|}\ket{0} + \sqrt{ 1 -\left(\frac{\min_{x} |A_{xf(x)}|}{|A_{xf(x)}|}\right)^2 } \ket{1}\right)\Biggr)\label{eq:invblock}
\end{align}
Then by applying the procedure described in~\eqref{eq:fullstateinv} and~\eqref{eq:invblock} we can therefore block-encode $A^{-1}$.  We call this operation $R_2$ in Figure~\ref{fig:construct}.  Since $R_2$ is unitary it is also linear and therefore will apply the inverse on every irreducible $2$-dimensional subspace simultaneously.

This procedure requires a circuit of polynomial size in the bits of precision, $n'$, to compute the inverse trigonometric function and reciprocal needed to perform this rotation.  However, this requires no queries and is efficient given $n'$ and therefore the precise details of the arithmetic circuits used are irrelevant for our lemma.  The one aspect that is relevant is that even if the inputs of $A_{xy}$ require finite precision, the computation of the arccosine will often require infinite precision to precisely represent.  This cost is $\Or({\rm poly}(n'\log(1/\delta)))$ quantum operations.  Thus we can perform a $(|A_{xy}|^{-1},1,\delta)-$block encoding of the inverse assuming each irreducible subspace is $2$-dimensional.

The construction for $R_1$ in Figure~\ref{fig:construct} is even simpler than that for $R_2$.
In this case $\ket{x}$ is a member of an irreducible one-dimensional subspace.  This occurs when $f(x)=x$ or, in other words, when the non-zero matrix element in row $x$ of $A$ occurs on the diagonal.  In this case, 
\begin{equation}
    A^{-1}\ket{x} = \frac{1}{A_{xx}} \ket{x} = \frac{{\rm sign}(A_{xx})}{|A_{xx}|} \ket{x}.
\end{equation}
Note that because $A$ is by assumption Hermitian, $A_{xx} \in \mathbb{R}$.  This further means that the diagonalization steps that were needed for the two-dimensional case are unnecessary here.  Thus in this case we can express the inverse as
\begin{equation}
    \ket{x}\ket{0}
    \mapsto {\rm sign}(A_{xx})\ket{x}\left(\frac{{\min_{x} |A_{xf(x)}|}}{|A_{xx}|}\ket{0} + \sqrt{ 1 -\left(\frac{{\min_{x} |A_{xf(x)}|}}{|A_{xx}|}\right)^2 } \ket{1}\right)
\end{equation}
Thus the inversion process in this case looks very similar to the two-dimensional case after applying the transformation $T$.  Then by replacing the Pauli-X operation in~\eqref{eq:fullstateinv} with a controlled-Z operation with the output of this function as control, we can selectively diagonalize the block.  Similarly, if $x=f(x)$ we need to flip the sign this can be achieved by comparing $A_{xx}$ to zero applying $Z$ to the resulting qubit that stores whether $A_{xx} < 0$ with the bit that stores $\ket{(f(x) \ne x)}$.  This can also be achieved using an $\Or(n)$-size circuit consisting of NOT, CNOT and Toffoli.  The part of the circuit that computes $\theta_{xy}$ is common to both cases and thus does not need to be changed.  Thus, if we augment the circuit by making these changes, we can implement a $((\min_{x,y}|A_{xy}|)^{-1},1,\delta)$-block encoding of the matrix inverse regardless of whether $\ket{x}$ is in an irreducible one- or two-dimensional subspace using $4$ oracle queries and $\Or({\rm poly}((n+n') \log(1/\delta)))$ auxiliary gate operations from the $R_z$, H , CNOT, Toffoli gate library, which can be compiled down at polynomial cost to gates taken from H,T,CNOT. This completes the proof.
\end{proof}

\section{Hadamard test for non-unitary matrices}
\label{sec:non_unitary_hadamard}

The well-known Hadamard test is frequently used to obtain the expectation value of a unitary operator. Suppose we want to obtain $\braket{\phi|U|\phi}$ for some unitary $U$ and $\ket{\phi}$, then for the real part we need the following circuit
\[
\Qcircuit @C=1.0em @R=0.9em {
\lstick{\ket{0}} & \gate{\mathrm{H}}      & \ctrl{1} & \gate{\mathrm{H}} & \meter  \\
\lstick{\ket{\phi}} & \qw & \gate{U} & \qw       & \qw  \\
} 
\]
and for the imaginary part we need the following circuit that is slightly different
\[
\Qcircuit @C=1.0em @R=0.9em {
\lstick{\ket{0}} & \gate{\mathrm{H}}      & \ctrl{1} & \gate{\mathrm{S}^{-1}} & \gate{\mathrm{H}} & \meter   \\
\lstick{\ket{\phi}} & \qw & \gate{U} & \qw    & \qw           & \qw  \\
}
\]
Here $\mathrm{H}$ is the Hadamard gate and $\mathrm{S}$ is the phase gate. We find that the probabilities of obtaining 0 when measuring the first qubit are $\frac{1}{2}(1+\Re\braket{\phi|U|\phi})$ and $\frac{1}{2}(1+\Im\braket{\phi|U|\phi})$ respectively for the two circuits.

A small modification gives us a way of computing expectation value of non-unitary matrices when given the block-encodings. If we have the $(\alpha,m,0)$-block-encoding of $A$ which we denote by $U_A$, then for the real part of $\braket{\phi|A|\phi}$ we consider the following circuit
\begin{equation}
\label{eq:hadamard_non_unitary_real}
\Qcircuit @C=1.0em @R=0.9em {
\lstick{\ket{0}} & \gate{\mathrm{H}}      & \ctrl{1} & \gate{\mathrm{H}} & \meter  \\
\lstick{\ket{0^m}} & \qw         & \multigate{1}{U_A}      & \qw  & \qw     \\
\lstick{\ket{\phi}} & \qw & \ghost{U_A} & \qw       & \qw  \\
} 
\end{equation}
and for the imaginary part we consider
\begin{equation}
\label{eq:hadamard_non_unitary_imag}
\Qcircuit @C=1.0em @R=0.9em {
\lstick{\ket{0}} & \gate{\mathrm{H}}      & \ctrl{1} & \gate{\mathrm{S}^{-1}} & \gate{\mathrm{H}} & \meter  \\
\lstick{\ket{0^m}} & \qw         & \multigate{1}{U_A}  & \qw     & \qw     & \qw \\
\lstick{\ket{\phi}} & \qw & \ghost{U_A} & \qw   & \qw     & \qw  \\
} 
\end{equation}
The probabilities of obtaining 0 when measuring the first qubit are  $\frac{1}{2}(1+\frac{1}{\alpha}\Re\braket{\phi|A|\phi})$ and $\frac{1}{2}(1+\frac{1}{\alpha}\Im\braket{\phi|A|\phi})$ respectively. One can derive these two probabilities easily by noting the fact that $(\bra{0^m}\bra{\phi}) U_A (\ket{0^m}\ket{\phi})=\frac{1}{\alpha}\braket{\phi|A|\phi}$.

It will be straightforward to estimate the probability of obtaining all $0$'s in measurement through Monte Carlo sampling, and thereby estimate the quantity of interest $\braket{\phi|A|\phi}$. The efficiency of the Monte Carlo sampling can be generally accelerated by the amplitude estimation procedure \cite[Theorem 12]{BrassardHoyerMoscaEtAl2002}.  However there are two issues we need to take into account, both arising from the preparation of $\ket{\phi}$: 
\begin{enumerate}
   \item The algorithm for preparing $\ket{\phi}$ may only be able to prepare a state close to it, which we denote as $\ket{\wt{\phi}}$ with trace distance $\sqrt{1-|\braket{\wt{\phi}|\phi}|^2}\leq\varsigma$.
    \item The algorithm for preparing $\ket{\phi}$ may involve measurement and have success probability lower bounded by $p<1$. 
\end{enumerate}
The first issue compels us to take the error of the state preparation into account. The second issue requires some further explanation. If we wanted to estimate $\braket{\phi|A|\phi}$ using Monte Carlo sampling, then we could simply measure the ancilla qubit for each successful preparation of the state $\ket{\phi}$, and do nothing when the state preparation is unsuccessful. However, when we want to use amplitude estimation, it is no longer possible to directly select only the successful instances in this way.  In order to construct the reflection operator needed in the amplitude amplification, we will need to prepare the state $\ket{\phi}$, or a state close to it, using a unitary circuit with success probability 1.

First we assume $\ket{\phi}$ can be prepared perfectly with probability 1 using a unitary circuit $U_\phi$. When estimating the real part using amplitude estimation, we will need to run $U_\phi$, the circuit in \eqref{eq:hadamard_non_unitary_real}, and their inverses $\Or(\alpha/\epsilon)$ times to estimate the real part to precision $\epsilon$. The same is true for the imaginary part.
 
Next we consider when the preparation involves some error. Due to the trace distance bound of $\ket{\wt{\phi}}$, we have \[
\left| \braket{\phi|A|\phi}-\braket{\wt{\phi}|A|\wt{\phi}} \right| \leq 2\|A\| \varsigma.
\]
Finally we consider the case when $\ket{\wt{\phi}}$ is produced only with probability at least $p$. Formally, we assume $(\bra{0^r}\otimes I_n) U_\phi \ket{0^{r}}\ket{0^{n}}=a\ket{\wt{\phi}}$ where $a\geq\sqrt{p}$. In this case we first apply the fixed-point amplitude amplification \cite[Theorem 27]{GilyenSuLowEtAl2018} to $U_\phi$. Compared to the standard amplitude amplification \cite{BrassardHoyerMoscaEtAl2002}, the fixed-point amplitude amplification has the advantage of using only a unitary circuit, and can boost the success probability 1. By \cite[Theorem 27]{GilyenSuLowEtAl2018}, there exists a unitary circuit $\wt{U}_\phi$ such that $\|\ket{0^r}\ket{\wt{\phi}}-\wt{U}_\phi \ket{0^{r}}\ket{0^{n}}\|\leq \epsilon'$, and this circuit uses $U_\phi$ and its inverse $\Or(\frac{1}{\sqrt{p}}\log(\frac{1}{\epsilon'}))$ times. We may set we $\epsilon'=\Or(\varsigma)$, so that the trace distance of the output away from $\ket{0^r}\ket{\wt{\phi}}$ is $\Or(\varsigma)$. 


A problem with amplitude estimation is that there is a failure probability, i.e.~the estimated amplitude differs from the true amplitude by more than the allowed error $\epsilon$, and this failure probability as mentioned in \cite[Theorem 12]{BrassardHoyerMoscaEtAl2002} is at most $1-8/\pi^2<1/2$. If we want the failure probability to be 
smaller than $\delta$, then we can run the amplitude estimation multiple times and take the median. Using the Chernoff--Hoeffding theorem we can estimate that we only need to run $\Or(\log(1/\delta))$ times to ensure the failure probability, i.e.~the probability that the median differs from the true amplitude by more than $\epsilon$, is at most $\delta$.
Therefore we have the following lemma:



\begin{lem}
\label{lem:non_unitary_hadamard_test}
Suppose a state $\ket{\phi}$ can be prepared with trace-distance error $\varsigma$ by a unitary circuit $U_\phi$ with probability at least $p$, and $A$ is given through its $(\alpha,m,0)$-block-encoding $U_A$, then $\braket{\phi|A|\phi}$ can be estimated using amplitude estimation to precision $2\alpha\varsigma+\epsilon$ with probability at least $1-\delta$, using $\Or((\alpha/\epsilon)\log(1/\delta))$ applications of $U_A$ and its inverse, $\Or((\alpha/\sqrt{p}\epsilon)\log(1/\varsigma)\log(1/\delta))$ applications of $U_\phi$ and its inverse, and $\Or((\alpha/\sqrt{p}\epsilon)\log(1/\varsigma)\log(1/\delta))$ other one- and two-qubit gates.
\end{lem}

\section{Computing imaginary parts of the Green's function without using the Hadamard test}
\label{sec:imagGreen}
We remark that if we are interested in computing the imaginary part of the Green's function (more accurately, the anti-Hermitian part of the Green's function), the calculation can be simplified as follows. Let $z=E-\I \eta$, and define
\[
\Gamma^{(\pm)}(z)=\frac{1}{2\I }\left(G^{(\pm)}(z)-(G^{(\pm)}(z))^{\dag}\right),
\]
which are real symmetric matrices.
Consider $\Gamma^{(+)}$ for simplicity, then
\[
\Gamma_{i j}^{(+)}(z)=\left\langle\Psi_{0}^{N}\left|\hat{a}_{i}\Im\left(z-\left[\hat{H}-E_{0}\right]\right)^{-1} \hat{a}_{j}^{\dagger}\right| \Psi_{0}^{N}\right\rangle =\eta \left\langle\Psi_{0}^{N}\left|\hat{a}_{i}\left((E\\ +E_0-\hat{H})^2+\eta^2\right)^{-1}\hat{a}_{j}^{\dagger}\right| \Psi_{0}^{N}\right\rangle.
\]
Note that the diagonal entries are 
\[
\Gamma^{(+)}_{ii}(z)=\eta\left\langle\Psi_{0}^{N}\left|\hat{a}_{i}\left((E+E_0-\hat{H})^2+\eta^2\right)^{-1}\hat{a}_{i}^{\dagger}\right| \Psi_{0}^{N}\right\rangle=\eta\norm{\left(z+E_0-\hat{H}\right)^{-1}\hat{a}_{i}^{\dagger}\ket{\Psi_{0}^{N}}}_2^2.
\]
If we solve the QLSP
\[
(z+E_0-\hat{H})\ket{y_i}=\hat{a}_{i}^{\dagger}\ket{\Psi_{0}^{N}},
\]
the success probability in measuring the ancilla qubits and obtain all 0's is
\[
\left(\alpha'_{z+E_0-\hat{H}}\right)^{-2}\norm{\left(z+E_0-\hat{H}\right)^{-1}\hat{a}_{i}^{\dagger}\ket{\Psi_{0}^{N}}}_2^2.
\]
Here $\alpha'_{z+E_0-\hat{H}}$ is the subnormalization factor for $\left(z+E_0-\hat{H}\right)^{-1}$. Hence $\Gamma_{i i}^{(+)}(z)$ can be directly computed from the success probability without using the Hadamard test.

In order to obtain the off-diagonal entries, we simply use the identity
\[
\begin{split}
\Gamma_{i j}^{(+)}(z)=&\frac\eta2\Bigg( \left\langle\Psi_{0}^{N}\left|(\hat{a}_{i}+\hat{a}_j)\left((E+E_0-\hat{H})^2+\eta^2\right)^{-1} (\hat{a}_{i}^{\dagger}+\hat{a}_j^{\dagger})\right| \Psi_{0}^{N}\right\rangle\\
-&\left\langle\Psi_{0}^{N}\left|\hat{a}_{i}\left((E+E_0-\hat{H})^2+\eta^2\right)^{-1}\hat{a}_{i}^{\dagger}\right| \Psi_{0}^{N}\right\rangle\\
-&\left\langle\Psi_{0}^{N}\left|\hat{a}_{j}\left((E+E_0-\hat{H})^2+\eta^2\right)^{-1}\hat{a}_{j}^{\dagger}\right| \Psi_{0}^{N}\right\rangle
\Bigg),
\end{split}
\]
which can again be evaluated as success probabilities for separately solving three linear systems. The treatment of $\Gamma^{(-)}$ follows the same procedure.

\section{Query complexities of estimating Green's function using HHL and LCU/QSVT}
\label{sec:query_complexities_HHL_LCU_Greens}

We assume that we are given an exact block-encoding of $\hat{H}$ denoted by $U_H$ with subnormalization factor $\alpha_H$. Then supposing we know the ground energy $E_0$, we can construct a block-encoding of $z-\hat{H}+E_0$ with subnormalization factor $\Or(|z|+\alpha_H)$ using \cite[Lemma 29]{GilyenSuLowEtAl2019}.

To our best knowledge, the HHL algorithm has not been formulated in terms of block-encoding in the literature, so for completeness for provide such a formulation below. The HHL algorithm relies on time-evolution  to solve the linear system. We assume the time-evolution, with a Hamiltonian $H$ to be specified, has a cost $\Or(\|H\|t)$, which is typical in many Hamiltonian simulation methods \cite{BerryChildsCleveETC2015,BerryChildsKothari2015,berry2007efficient,LowChuang2019,LowChuang2017,LowWiebe2019}, and omit the dependence on the desired precision, because the dependence is only polylogarithmic and therefore does not play an important role in the complexity. The circuit construction of HHL  effectively gives a block-encoding of the matrix inverse. We explain it in detail below.

Assume we have a matrix $M$, given in its block-encoding $U_M$ with subnormalization factor $\alpha_M$. If $M$ is not Hermitian, we may considered the dilated matrix denoted by $\wt{M}$ as in \cref{eqn:dilation_trick}.  If we are able to block-encode $\wt{M}^{-1}$, then extracting the lower-left block of this matrix through single-qubit operations will yield us $M^{-1}$. We then discuss how to get $\wt{M}^{-1}$ using HHL.

The HHL algorithm consists of the following steps: We start with two ancilla registers and a main register, containing $r$ and one qubit respectively, each ancilla qubit in the $\ket{0}$ state, and the main register starting with some state $\ket{b}$. First we perform Hadamard transform on the first ancilla register, which is usually called the clock register. Then we apply $\sum_{\tau=0}^{T-1}\ket{\tau}\bra{\tau}\otimes e^{\I\wt{M}\tau t_0/T}$, the controlled time evolution of $\wt{M}$, on the main register, controlled by the first ancilla register. Next we apply QFT on the first ancilla register, so that this register stores the approximate eigenvalues of $\wt{M}$ in superposition. Then we apply rotation on the second ancilla register, which contains only one qubit, controlled by the first register. Finally we uncompute the first ancilla register, and measure the second ancilla register. Upon obtaining outcome $1$ we have successfully prepared $\wt{M}^{-1}\ket{b}/\|\wt{M}^{-1}\ket{b}\|$. 

The main source of error in HHL is the phase estimation step. To control the phase estimation error to be within $\delta$ we need to let $t_0=\Or(1/\delta)$. However when the eigenvalue $\lambda$ is off by $\delta$ its inverse $1/\lambda$ withh be off by approximately $2\delta/\lambda^2$ for $\delta\ll \lambda$. We denote the smallest singular value of $M$ as $\sigma_{\min}$. Therefore in order to have a block-encoding error of $\wt{M}$ to be smaller than $\epsilon'$, we need to let $t_0=\Or(1/\epsilon'\sigma_{\min}^2)$. In the original HHL algorithm this dependence is mitigated because there is a subnormalization of the output quantum state involved. If we want $M^{-1}$ to have a block-encoding error upper bounded by $\epsilon$, we only need $\wt{M}^{-1}$ to have a block-encoding error upper bounded by $\alpha_M\epsilon$. Therefore we can set $\epsilon'=\alpha_M \epsilon$.


In the case of $M=z-\hat{H}+E_0$, we have $\alpha_M=\Or(|z|+\alpha_H)$, the smallest singular value of $M$ is lower bounded by $\eta$, and therefore $\sigma_{\min}=\Omega(\eta/\alpha_M)$. Therefore we have $t_0=\Or(\alpha_M/\eta^2 \epsilon)$. The dominant cost of constructing block-encoding of $M^{-1}$ is running Hamiltonian simulation of $\wt{M}$ for time $t_0$. Therefore a single block-encoding of $M^{-1}$ takes $\Or(\alpha_M/\eta^2 \epsilon)$ applications of $U_M$ and its inverse. The block-encoding subnormalization factor is upper bounded by the block-encoding subnormalization of $\wt{M}^{-1}$, which is at most $1/\sigma_{\min}$, divided by $\alpha_M$. It is therefore $\Or(1/\sigma_{\min}\alpha_M)=\Or(1/\eta)$.   When we take the estimates for the number of queries to $U_M$ and the subnormalization factor into Lemma~\ref{lem:non_unitary_hadamard_test}, we arrive at the query complexity estimates in the first row of Table~\ref{tab:compare_algs_hubbard}.

The complexities of LCU and QSVT are both easier to estimate than HHL because of the polylogarithmic dependence on the desired block-encoding error, as both methods use results from approximation theory to make this possible. In both cases we first construct the block-encoding of $\wt{M}^{-1}$, whose singular values are in $[1/\sigma_{\min},1]$, where $\sigma_{\min}$ is the same as defined before. To construct a block-encoding of $\wt{M}^{-1}$ with error upper bounded by $\epsilon'$, $\wt{\Or}((1/\sigma_{\min})\poly\log(1/\epsilon'))$ queries to $U_{\wt{M}}$ are required.  When we regard the resulting circuit as a block-encoding of $M^{-1}$ we need to set $\epsilon'=\alpha_M \epsilon$ as before. Therefore in the application to Green's function a single block-encoding queries $U_M$ and its inverse  $\wt{\Or}((\alpha_M/\eta)\poly\log(1/\epsilon))$ times. The resulting block-encoding subnormalization factor of $M^{-1}$ scales the same as in HHL, which in the application to Green's function is upper bounded by $\Or(1/\eta)$. These estimates lead to the complexities in the second row of Table~\ref{tab:compare_algs_hubbard}.

\section{Proof of \cref{lem:quadrature_err}: discretizing the contour integral}\label{sec:proof_quadrature_err}

In this section we prove Lemma~\ref{lem:quadrature_err}.
The goal is to discretize the contour integral
\begin{equation}
I = \frac{1}{2\pi \I}\oint_\Gamma \frac{e^{-\beta z}}{x-z}\ud z = \frac{-1}{2\pi \I}\int_{-\infty}^\infty \frac{e^{-\beta (t^2-\zeta+\I t)}(2t+\I)}{x-t^2+\zeta-\I t}\ud t,
\end{equation}
for $x>0$ where $\Gamma$ is defined in Eq.~\eqref{eq:contour}. We first want to approximate the integral on the real axis by an integral on a finite interval. We define
\begin{equation}
I_T = \frac{-1}{2\pi \I}\int_{-T}^T \frac{e^{-\beta (t^2-\zeta+\I t)}(2t+\I)}{x-t^2+\zeta-\I t}\ud t,
\end{equation}
and 
\begin{equation}
g(t,x) = \frac{e^{-\beta (t^2-\zeta+\I t)}(2t+\I)}{x-t^2+\zeta-\I t}.
\end{equation}
Then we have, for $|t|\geq 1/2$, 
$|g(t,x)|\leq \sqrt{8}e^{-\beta(t^2-\zeta)}$. This implies
\[
|I-I_T|\leq \frac{\sqrt{2}e^{\beta\zeta}}{\pi}\int_{|t|>T}e^{-\beta t^2}\ud t=\sqrt{\frac{2}{\beta\pi}}e^{\beta\zeta}\textrm{erfc}(\sqrt{\beta}T),
\]
when $T\geq 1/2$. Due to the inequality \cite{chiani2003new}
\[
\mathrm{erfc}(x)=\frac{2}{\sqrt{\pi}}\int_{x}^{\infty}e^{-t^2}\ud t\leq e^{-x^2},
\]
we can bound the error due to the finite range truncation as
\begin{equation}
\label{eq:bound_T}
|I-I_T|\leq \sqrt{\frac{2}{\beta\pi}} e^{\beta \zeta}e^{-\beta T^2}.
\end{equation}

It now remains to find a quadrature for $I_T$. To do this we use the Gauss-Legendre quadrature. We write out the Chebyshev expansion of $g(t,x)$ as
\begin{equation}
g(t,x) = \sum_{n=0}^\infty \hat{g}_n(x)T_n(t/T),
\end{equation}
where $T_n(t)$ is the $n$-th Chebyshev polynomial of the first kind. We define 
\begin{equation}
I_{GL} = T\sum_{j\in[J]}  w_j g(s_j T,x),
\end{equation}
where $w_j$ and $s_j$ are selected according to the standard Gauss-Legendre quadrature formula on $[-1,1]$. Because this quadrature formula is exact for polynomials of degree up to $2J-1$, and that
\[
\sum_{j\in[J]}  w_j |T_n(s_j)| \leq 2,
\]
we have
\begin{equation}
\label{eq:bound_GL_error}
|I_T-I_{GL}|\leq 2T\sum_{n\geq 2J} |\hat{g}_n(x)|.
\end{equation}
Therefore we only need to bound the coefficients $|\hat{g}_n(x)|$.

Let $h(\theta,x)=g(T\cos(\theta),x)$, and write out its Fourier expansion
\[
h(\theta,x) = \sum_{n=-\infty}^{\infty}\hat{h}_n(x)e^{in\theta}
\]
then \[
\hat{g}_n(x) = \hat{h}_n(x)+\hat{h}_{-n}(x),\ n> 0,\quad \hat{g}_0(x) = \hat{h}_0(x).
\]
Note that if we introduce a new variable $z=e^{\I \theta}$, then
\[
\wt{h}(z,x)=h(\theta,x)=\sum_{n=-\infty}^{\infty}\hat{h}_n(x) z^n,
\]
which takes the form of a Laurent series. Therefore the  coefficients of the Chebyshev expansion $\{\hat{g}_n(x)\}_{n=0}^{\infty}$ are directly related to the coefficients of the Laurent expansion  
$\{\hat{h}_n(x)\}_{n=-\infty}^{\infty}$.

The function $g(t,x)$ is analytic in $t$ in the domain $\{t\in\CC:t^2-\zeta+\I t-x\neq 0\}$. We need to estimate how far from the real axis $h(\theta,x)$ can be extended as an analytic function of $\theta$ in order to estimate the convergence rate of the Fourier series 
(see e.g. \cite[Chapter 8]{Trefethen2019Approximation}). 

\begin{figure}
    \centering
    \includegraphics[width=0.5\textwidth]{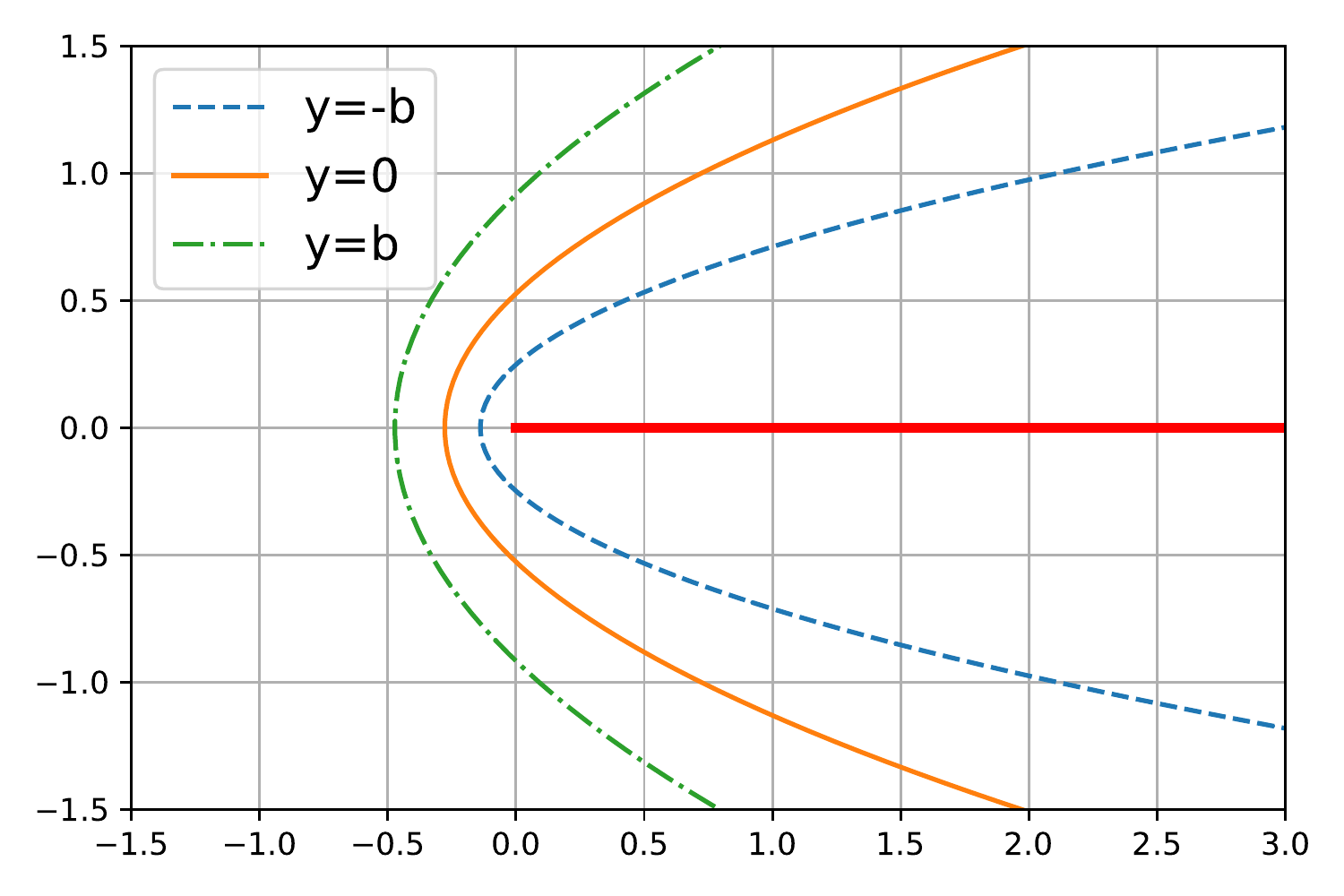}
    \caption{The parabola $\{t^2+\I t-\zeta:\Im t=y\}$ for $y=-b,0,b$. The spectrum of $H=A+B$ is on the positive part of the real axis (solid red line).}
    \label{fig:contour_widen}
\end{figure}

We want to lower bound $|t^2+\I t-x-\zeta|$ when $|\Im t|$ is bounded by some parameter $b$ to be chosen. We write $t=w+\I y$, and require $|y|\leq b$, then choose 
\begin{equation}
\label{eq:choose_zeta}
\zeta=2b(1-b).
\end{equation}
$\zeta$ is chosen in this way so that the distance between the positive part of the real axis and the set $\{t^2+\I t-\zeta:|\Im t|\leq b\}$ can be bounded from below, which we will consider next. For each fixed $\Im t=y$, the image of $\{t^2+\I t-\zeta:\Im t=y\}$ is a parabola in the complex plane parameterized by $w$. The parabola moves to the left and widens when $y$ increases from $-b$ to $b$ (see Figure~\ref{fig:contour_widen} for an illustration). Therefore we only need to consider when $y=-b$. In this case
\[
|t^2+\I t-x-\zeta|^2 = (w^2-b(1-b)-x)^2+w^2(1-2b)^2.
\]
We choose
\begin{equation}
\label{eq:choose_b}
b = \min \left(\frac{1}{2\beta},\frac{1}{6}\right) = \frac{1}{2\max(\beta,3)}.
\end{equation}
In particular, $0\le b\le 1/6$ ensures that $(1-2b)^2\geq 2b(1-b)$. 

We consider two cases. First, when $w^2\geq b(1-b)$, we have
\[
(w^2-b(1-b)-x)^2+w^2(1-2b)^2 \geq w^2(1-2b)^2\geq b(1-b)(1-2b)^2\geq 2b^2(1-b)^2.
\]
Second, when $w^2< b(1-b)$, because $x\geq 0$, we have
\[
(w^2-b(1-b)-x)^2+w^2(1-2b)^2 = w^4 + ((1-2b)^2-2b(1-b))w^2 + b^2(1-b)^2 \geq  b^2(1-b)^2.
\]
Combining these two cases we have
\begin{equation}
\label{eq:integral_bound_denom}
|t^2+\I t-x-\zeta| \geq b(1-b) = \frac{\zeta}{2}.
\end{equation}
for all $|\Im t|\leq b$ and $x \geq 0$. 

Next we determine how far $h(\theta,x)$ can be analytically extended. By the relation $t=T\cos(\theta)$, we have
\[
\Im t = -T\sin(\Re\theta)\sinh(\Im \theta).
\]
Therefore
\[
|\Im t|\leq T\sinh(|\Im \theta|)\leq 2T|\Im \theta|
\]
when $|\Im \theta|\leq 1$. Note that so far we have only used $b\leq 1/6$. In Eq.\eqref{eq:choose_b} we also have $b\leq 1/(2\beta)$, which ensures that $|e^{-\beta z}|$ for all $z$ in a vicinity of the parabola can be bounded from above by a constant. This is very important when we prove the exponential decay of the Fourier coefficients below.

Therefore for $|\Im t|\leq b$ we only need to require $|\Im \theta|\leq b/2T$ when $T\geq 1/2$. Thus, $h(\theta,x)$ can be analytically extended in $\theta$ to the strip $\{\theta\in\CC:|\Im\theta|\leq b/2T\}$. Also $h(\theta,x)$ is periodic in the real direction with a period $2\pi$. 
Now consider $\wt{h}(z,x)=h(\theta,x)$ with  $z=e^{\I\theta}$, then $\wt{h}(\cdot,x)$ is analytic in the annulus $\{z\in\CC: e^{-b/2T}<|z|<e^{b/2T}\}$. The coefficients of the Laurent series can be evaluated using the residue theorem as
\[
\hat{h}_n(x)=\frac{1}{2\pi \I}\oint_{\Gamma(a)} z^{-(n+1)} \wt{h}(z,x)\ud z,
\]
where the contour is $\Gamma(a):=\set{z:\abs{z}=e^{a}}$ with $-b/2T<a<b/2T$. By changing the variable back from $z$ to $\theta$, we arrive at the following bounds for the Fourier coefficients
\begin{equation}
\label{eq:decay_fourier_coefs}
\begin{aligned}
|\hat{h}_n(x)|&\leq\frac{1}{2\pi e^{a|n|}}\int_{0}^{2\pi}|h(\theta-ia,x)|d\theta, \quad n\geq 0, \\
|\hat{h}_n(x)|&\leq\frac{1}{2\pi e^{a|n|}}\int_{0}^{2\pi}|h(\theta+\I a,x)|d\theta, \quad n< 0,
\end{aligned}
\end{equation}
for $0<a<b/2T$.

We then bound $h(\theta+\I a,x)$ for $|a|<b/2T$. Let $t=w+\I y$, then as analyzed above $|a|<b/2T$ guarantees $|y|\leq b$. Hence $\Re(t^2+\I t)$ is minimized when $w=0$ and $y=b$. In this case $\Re(t^2+\I t)=-b(1+b)$. Thus we have
\begin{equation}
\label{eq:integrand_bound_1}
\Re(t^2-\zeta+\I t) \geq -b(1+b)-\zeta = b(1-3b)\geq \frac{1}{4\beta},
\end{equation}
using the fact $b\leq 1/6$ and $b\leq 1/(2\beta)$ in Eq.~\eqref{eq:choose_b}.
This enables us to bound the exponential in $g(t,x)$. 
We combine the bound for the exponential with Eq.~\eqref{eq:integral_bound_denom} to ensure
\[
|h(\theta+\I a,x)|=|g(t,x)|=\left|\frac{e^{-\beta (t^2-\zeta+\I t)}(2t+\I )}{x-t^2+\zeta-\I t}\right|\leq \frac{8Te^{1/4}}{\zeta}.
\]
We have chosen $b$ and $\zeta$ in Eq.~\eqref{eq:choose_b} and \eqref{eq:choose_zeta} respectively. Using these two equations and the fact that $b\leq 1/6<1/2$ we have $1/\zeta\leq 1/b=2\max(\beta,3)$. Therefore
\[
|h(\theta+\I a,x)| \leq 16Te^{1/4}\max(\beta,3).
\]
Taking it into Eq.~\eqref{eq:decay_fourier_coefs} we have
\[
\abs{\hat{h}_n(x)}\leq \frac{16T\max\left(\beta,3\right)e^{1/4}}{e^{a|n|}}
\]
for $0<a<b/2T=1/(4T\max(\beta,3))$ and $n\in\mathbb{Z}$. Using this to get $\hat{g}_n(x)$, setting $a=1/(8T\max(\beta,3))$, and taking into Eq.~\eqref{eq:bound_GL_error}, we have 
\begin{equation}
\label{eq:bound_GL_error_2}
\begin{aligned}
|I_T-I _{GL}| \leq \sum_{n\geq 2J} \frac{64 T^2\wt{\beta} e^{1/4}}{e^{an}} 
=\frac{64T^2\wt{\beta}e^{1/4}}{1-e^{-1/(8T\wt{\beta})}}e^{-J/(4T\wt{\beta})},
\end{aligned}
\end{equation}
where $\wt{\beta}=\max(\beta,3)$.
Combining this inequality with \eqref{eq:bound_T}, in which we use $\beta\zeta\leq 1$, we have
\begin{equation}
\label{eq:total_error_bound}
\begin{aligned}
|I-I_{GL}| &\leq |I-I_T|+|I_T-I_{GL}| \\
&\leq \sqrt{\frac{2}{\beta\pi}} e^{1-\beta T^2}+
\frac{64T^2\wt{\beta}e^{1/4}}{1-e^{-1/(8T\wt{\beta})}}e^{-J/(4T\wt{\beta})}.
\end{aligned}
\end{equation}

\section{Efficient block-encodings for the contour integral approach}\label{sec:contour_integral_blockencoding}

In this section we construct the block-encodings in Lemma~\ref{lem:oracles_ab}.
The first thing we need to implement is the block-diagonal matrix 
\begin{equation}
\label{eq:select_A}
\begin{aligned}
&\sum_{j\in[J]}\ket{j}\bra{j}\otimes (z_j+\xi_j-A)^{-1}  \\
&=\begin{pmatrix}
(z_1+\xi_1-A)^{-1} &  & & \\
 & (z_2+\xi_2-A)^{-1} & & \\
 & & \ddots &                \\
 & & & (z_J+\xi_J-A)^{-1}
\end{pmatrix}
\end{aligned}
\end{equation}
This block-diagonal matrix can be implemented by the following circuit
\begin{equation}
\label{eq:A_inv}
\Qcircuit @C=1.0em @R=0.9em {
\lstick{\ket{0}} & \qw      & \qw & \sgate{\mathrm{INV}_A}{1} & \qw & \qw   \\
\lstick{\ket{j}} & \qw         & \qw      & \sgate{\mathrm{INV}_A}{2} & \qw  & \qw     \\
\lstick{\ket{\phi}} & \gate{V^\dagger} & \multigate{1}{O_D} & \qw  & \gate{V}     & \qw  \\
\lstick{\ket{0^r}} & \qw      & \ghost{O_D} & \gate{\mathrm{INV}_A} & \qw & \qw 
} 
\end{equation}
where $\mathrm{INV}_A$ satisfies
\[
\mathrm{INV}_A\ket{0}\ket{j}\ket{\lambda} = \left(\frac{1}{z_j+\xi_j-\lambda}\ket{0}+\sqrt{1-\left|\frac{1}{z_j+\xi_j-\lambda}\right|^2}\ket{1}\right)\ket{j}\ket{\lambda}
\]
 This is a $(1,\Or(r+\log(J)),0)$-block-encoding of the block-diagonal matrix in Eq.~\eqref{eq:select_A}. The main purpose of the above design is to use $V$ and $O_D$ only $\Or(1)$ times instead of $\Or(J)$ times.

Similar to the above construction, it is possible to construct a $(1+\alpha_B,\Or(1),0)$-block-encoding of the following block-diagonal matrix
\begin{equation}
\label{eq:select_B}
\begin{aligned}
\sum_{j\in[J]} \ket{j}\bra{j}\otimes (B+\xi_j) 
=
\begin{pmatrix}
B+\xi_1 &  & & \\
 & B+\xi_2 & & \\
 & & \ddots &       \\
 & & & B+\xi_J
\end{pmatrix},
\end{aligned}
\end{equation}
using controlled-$U_B$ only once and a logical circuit to determine the value of $\xi_j$ for each input $j$.

\section{Proof of \cref{thm:Chebyshev}: convergence of the truncated Chebyshev series}\label{sec:Chebyshev_convergence_proof}

In this section we prove \cref{thm:Chebyshev}. 
For completeness we repeat the essential steps of the proof in~\cite[Section 5.7]{MasonHandscomb2003} with a more explicit constant dependence, which 
is the key to obtain exponential convergence of functions in the Gevrey class. 

    Using Peano's theorem~\cite[Lemma 5.15]{MasonHandscomb2003}, 
    \begin{equation}\label{eqn:matrix_function_peano}
        \sum_{k=0}^d c_k T_k(y) - g(y) = \int_{-1}^1 g^{(r+1)}(t)K_d(y,t)\ud t.
    \end{equation}
    Here
    \begin{equation}
        K_d(y,t) = \frac{1}{r!}\left(\sum_{j=0}^d c_{jr}T_j(y) - (y-t)^r_+\right),
    \end{equation}
    where
    \[
    (y-t)_{+}^{r}:=
    \begin{cases}
(y-t)^{r}, & y \geq t \\
0, & y<t
\end{cases},
    \]
    and
    \[
        c_{jr} = \frac{2-\delta_{j0}}{\pi}\int_{t}^1\frac{(y-t)^rT_j(y)}{\sqrt{1-y^2}} \ud y.     \]
    Note that $K_d$ does not depend on the function $g(y)$. 

    Now let us bound the coefficient $c_{jr}$
    for all $j \geq 1$. 
    Using the substitution $y = \cos\theta$ and $t = \cos\phi$, 
    \begin{equation}
        c_{jr} = \frac{2}{\pi}\int_0^{\phi}(\cos(\theta)-\cos(\phi))^r\cos(j\theta)\ud\theta.
    \end{equation}
    Define $h(s) = (s-\cos(\phi))^r$ and $f(\theta) = h(\cos(\theta))$. 
    Notice that the $(l+1)$-th order anti-derivative of $\cos(j\theta)$ is 
    $\frac{(-1)^{l+1}}{j^{l+1}}\cos^{(l+1)}(j\theta)$,
    using integration by parts, we obtain 
    \begin{equation}
        \frac{\pi}{2} c_{jr} = -\left[\sum_{l=0}^{r} f^{(l)}(\theta)\frac{1}{j^{l+1}}\cos^{(l+1)}(j\theta) \right]_{0}^{\phi} 
        + \frac{1}{j^{r+1}}\int_{0}^{\phi} f^{(r+1)}(\theta)\cos^{(r+1)}(j\theta)\ud\theta. 
    \end{equation}
    By \cref{lem:chain_rule}, 
    \begin{align}
     f^{(l)}(\theta) &= \sum_{\sum_{p=1}^l pq_p = l} \frac{l!}{q_1!(1!)^{q_1}q_2!(2!)^{q_2}\cdots q_l!(l!)^{q_l}}h^{(q_1+q_2+\cdots+q_l)}(\cos(\theta))\prod_{p=1}^l\left(\cos^{(p)}(\theta)\right)^{q_p} \\
     &= \sum_{\sum_{p=1}^l pq_p = l} \frac{l!}{q_1!(1!)^{q_1}q_2!(2!)^{q_2}\cdots q_l!(l!)^{q_l}}\frac{r!}{(r-\sum_p q_p)!}(\cos(\theta)-\cos(\phi))^{r-\sum_p q_p}\prod_{p=1}^l\left(\cos^{(p)}(\theta)\right)^{q_p}.
    \end{align}  
    Notice that for all $l \leq r - 1$, $\sum_p q_p \leq \sum_{p} p q_p = l < r$, we have $f^{(l)}(\phi) = 0$ for all $l \leq r-1$. 
    Then
    \begin{align*}
        \frac{\pi}{2} c_{jr} &= -f^{(r)}(\phi)\frac{1}{j^{r+1}}\cos^{(r+1)}(j\phi) + \sum_{l=0}^{r} f^{(l)}(0)\frac{1}{j^{l+1}}\cos^{(l+1)}(0) \\
        &\quad\quad+ \frac{1}{j^{r+1}}\int_{0}^{\phi} f^{(r+1)}(\theta)\cos^{(r+1)}(j\theta)\ud\theta. 
    \end{align*}
    Furthermore, when $l$ is odd, from the equation $\sum_{j=1}^l jq_j = l$, 
    for any tuple $(q_1,\cdots,q_l)$, 
    there must exist an odd number $p_0$, such that $q_{p_0} \neq 0$. 
    Therefore $(\cos^{(p_0)}(0))^{q_{p_0}} = 0$, and thus 
    $f^{(l)}(0) = 0$. 
    When $l$ is even, $\cos^{(l+1)}(0) = 0$. 
    Therefore we have, for all $l \leq r$, $f^{(l)}(0)\cos^{(l+1)}(0) = 0$, and
    \begin{equation}
        \frac{\pi}{2} c_{jr} = - f^{(r)}(\phi)\frac{1}{j^{r+1}}\cos^{(r+1)}(j\phi) 
        + \frac{1}{j^{r+1}}\int_{0}^{\phi} f^{(r+1)}(\theta)\cos^{(r+1)}(j\theta)\ud\theta. 
    \end{equation}
    Finally, using some very rough estimates that $|\cos^{(p)}(\theta)|\leq 1$ and dropping all the denominators in $f^{(l)}(\theta)$ 
    (which can be surely improved, but here for technical simplicity we keep these rough estimates), and 
    that the number of $l$-tuples is less than 
    $(l+1)(l/2+1)(l/3+1)\cdots (l/l+1) = \binom{2l}{l} < 2^{2l}$, we have
    \begin{equation}
        \|f^{(l)}\|_{\infty} \leq 2^{2l}l!r!2^r = 2^{2l+r}l!r!,
    \end{equation}
    and we obtain
    \begin{align*}
        |c_{jr}| &\leq \frac{2}{\pi} \frac{2^{3r}(r!)^2}{j^{r+1}} 
        + \frac{2}{\pi}\frac{\pi2^{3r+2}(r+1)!r!}{j^{r+1}} \\
        &\leq 2^{3r+4} \frac{(r+1)!r!}{j^{r+1}}. 
    \end{align*}

    It follows that 
    \begin{align*}
        |K_d(y,t)| & \leq \frac{1}{r!} \left|\sum_{j=d+1}^{\infty}c_{jr}T_j(y)\right| \\
        & \leq \frac{1}{r!}\sum_{j=d+1}^{\infty}\left|c_{jr}\right|  \\
        & \leq \frac{2^{3r+4}(r+1)!r!}{r!}\sum_{j=d+1}^{\infty} \frac{1}{j^{r+1}} \\
        & < {2^{3r+4}(r+1)!} \int_d^{\infty}\frac{1}{x^{r+1}}\ud x\\
        & \leq 16\frac{8^r (r+1)!}{d^{r}}.
    \end{align*}
    Combining this estimate with \cref{eqn:matrix_function_peano}, 
    we complete the proof.

\section{Green's function computation for fixed number of electrons}\label{app:signature}
Here we  review how the scaling of the complexity for evaluating Green's functions for certain quantum many-body Hamiltonians can be improved by utilizing the electron number constraint as in \cref{sec:example_manybody}.  Our aim in this section is to discuss how the block-encoding construction can be modified in these cases to accommodate this.  

First, let us assume that we have a Hamiltonian that is the sum of two orthogonally diagonalizable Hamiltonians. This means that if we express the Hamiltonian in its eigenbasis then there  exist a unitary matrices $\hat{U}$  such that $H= H_0+A:=\hat{U} D \hat{U}^\dagger + A$ such that $A$  and $D$ are diagonal matrices and $\hat{U}$ transforms from the computational basis to the eigenbasis of $H_0$.  Specifically, if we let $\ket{\psi_k}$ be an eigenstate of $H_0$ then  
\begin{equation}
    \hat{U} \ket{\psi_k} = \ket{k},
\end{equation}
where $\ket{k}$ is the $k^{\rm th}$ computational basis vector.

If we let $P_{N_e}$ be a projector onto a constrained manifold, in our case the  manifold of states with a fixed electron number ${N_e}$, then we can define the constrained Hamiltonian $H'$ via
\begin{equation}
    H' = P_{N_e}(\hat{U} {D} \hat{U}^\dagger + A)P_{N_e}
\end{equation}
Further, let us also assume that $P_{N_e}$ commutes not only with $H$ but also $A$.  If this is true then $[H', P_{N_e}]=[\hat{U} D \hat{U}^\dagger,P_{N_e}]=0$.  Therefore $\hat{U} D \hat{U}^\dagger P_{N_e} = P_{N_e} \hat{U} D \hat{U}^\dagger$ and in turn
\begin{equation}
    H'=(\hat{U} (\hat{U}^\dagger P_{N_e} \hat{U}) D \hat{U}^\dagger) + AP_{N_e},
\end{equation}
has the same action within the subspace conditioned on the value of ${N_e}$.  Specifically, for any $+1$ eigenstate of $P_{N_e}$ denoted by $\ket{\psi}$, we have  
\begin{equation}
    H\ket{\psi} = H P_{N_e} \ket{\psi} = H'  \ket{\psi}.
\end{equation}  
This validates the claim that such modifications do not affect the action of the Hamiltonian $H'$ on the fixed-particle manifold.

Furthermore, let $H''= H'+ C(1-P_{N_e})$ for Hermitian matrix $C$.  We then have that for any $+1$ eigenstate of $P_{N_e}$ $\ket{\psi}$
\begin{equation}
    H''\ket{\psi} = H'' P_{N_e} \ket{\psi} = H' \ket{\psi}.
\end{equation}
Thus we can perturb $H$ any way we see fit so long as the perturbation has no impact on the particle number manifold in question.  We will use this fact to simplify our block encoding construction.

Since $H_0=UDU^\dagger$, we simply need to construct a unitary for block-encoding $D$ in order to e convert this classical circuit into a quantum circuit $O_D$ satisfying
\[
O_D\ket{k}\ket{c}=\ket{k}\ket{c\oplus D_{kk}}.
\]

For each eigenstate $\ket{\psi_k}$ of $H_0$  indexed by $k$, we can compute the particle number efficiently through a classical circuit, which we also convert to a quantum circuit $O_{\mathrm{num}}$, that satisfies
\[
O_{\mathrm{num}}\ket{k}\ket{c}=\ket{k}\ket{c\oplus N_e(k)},
\]
where $\hat{N}$ is the number operator and
\[
N_e(k) = \braket{\psi_k|\hat{N}|\psi_k}.
\] 

\begin{figure}
\centering
\includegraphics[width=0.95\textwidth]{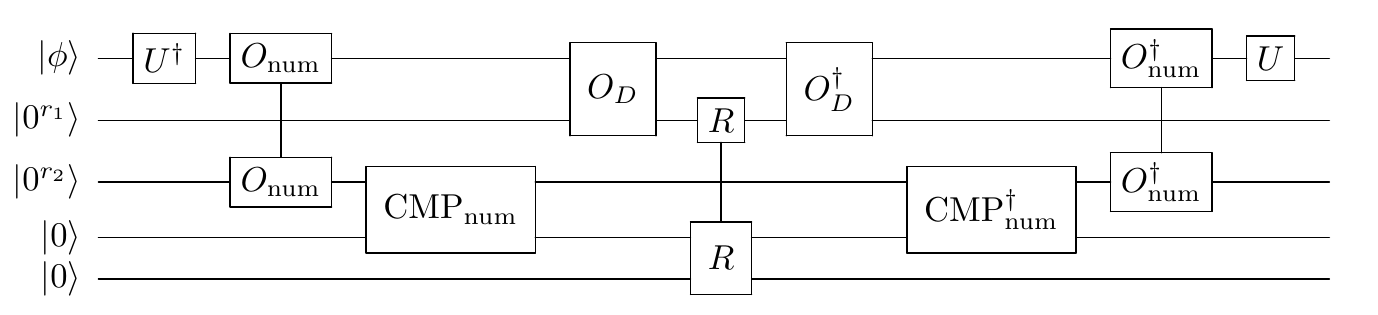}
\caption{Circuit for performing a block encoding projector on a fixed particle number manifold of states.\label{fig:circProj}}
\end{figure}
The components work as follows:
\[
\mathrm{\mathrm{CMP_{num}}}\ket{N_e(k)}\ket{c} = 
\begin{cases}
\ket{N_e(k)}\ket{c\oplus 1},\  &\mathrm{if}\  N_e(k)=N_e, \\
\ket{N_e(k)}\ket{c},\  &\mathrm{if}\  N_e(k)\neq N_e. \\
\end{cases}
\]
and
\[
R\ket{D_{kk}}\ket{c}\ket{0}=
\begin{cases}
\ket{D_{kk}}\ket{c}\left(\frac{D_{kk}}{\alpha(N_e)}\ket{0}+\sqrt{1-\left(\frac{D_{kk}}{\alpha(N_e)}\right)^2}\ket{1}\right),\  &\mathrm{if}\  c=0, \\
\ket{D_{kk}}\ket{c}\ket{1},\  &\mathrm{if}\  c=1. \\
\end{cases}
\]
We only need to make sure $\alpha(N_e)\geq D_{kk}$ for all $k$ such that $N_e(k)=N_e$. Therefore in the case of the kinetic operator that appears in the Hubbard model, the this gives a block-encoding of $H_0 P_{N_e}$ with subnormalization factor $\alpha(N_e)$.

Next we will assume that $[\hat{U},P_{N_e}]=0$.  While this assumption is not strictly needed to simplify the Hamiltonian to accommodate the particle number constraint, the algorithmic design will be easier.  As an example, consider the fermionic Fourier transform
\begin{align}
    \left[{\rm FFFT}, \sum_j n_j\right] &= {\rm FFFT} \sum_j n_j -  \sum_j n_j{\rm FFFT} = {\rm FFFT} \left(\sum_j n_j -  \sum_j c_j^{\dagger} c_j\right).\nonumber\\
    & = \sum_k {\rm FFFT} \left(\sum_{k} k P_k - \sum_k k P_k \right)=0.
\end{align}
Therefore we have that $[{\rm FFFT}, P_{N_e}] =0$ and in turn this holds for the Hubbard model, plane-wave dual simulations and the Schwinger model.  In all these cases it follows that
\begin{equation}
    H'=\hat{U} (P_{N_e}  D) \hat{U}^\dagger + AP_{N_e}.
\end{equation}
Thus in these cases we can zero out any eigenvalues of the orthogonally diagonalizable operators that are outside of the constraint specified by $P_{N_e}$.

For the case of Green's function calculation for the Hubbard model, $\hat{U}={\rm FFFT}$ and we can implement a block-encoding for $H_0P_{N_e}$ on the fixed particle number block  using the construction in Figure~\ref{fig:circProj}.   In particular, from~\eqref{eq:HubbardFFFT} it is clear that if we project the state onto the manifold consisting of ${N_e}$ electrons we obtain using the fact that ${\rm FFFT}$ commutes with the projector onto the manifold with ${N_e}$ electrons (denoted $P_{N_e}$) we observe that
\begin{equation}
\sum_{\vx, \vy,\sigma} T(\vx-\vy) \hat{a}_{\vx \sigma}^{\dagger} \hat{a}_{\vy \sigma}P_{N_e}=\mathrm{FFFT} \sum_{\vG,\sigma}\hat{T}(\vG) \hat{c}^{\dag}_{\vG,\sigma}\hat{c}_{\vG,\sigma} P_{N_e}\mathrm{FFFT}^\dagger
\end{equation}
We therefore have that the normalization factor for the block-encoding of the kinetic operator constrained to this subspace obeys
\begin{equation}
    \alpha_T\le \frac{{N_e}}{2}\max_{\vG}|\hat{T}(\vG)| = \max_{k_x,k_y} \left| \frac{{N_e}}{2N}\sum_{x,y} T(x,y) e^{2\pi i (k_x x +k_y y)/\sqrt{N}} \right| \le {N_e} |t|,
\end{equation}
The argument for $\alpha_U$ is exactly the same except we do not need to worry about the diagonalizing ${\rm FFFT}$ operation.  This means that we can directly apply the reasoning in~\eqref{eq:alphaU} to find
\begin{equation}
    \alpha_U \le N_e|U|. 
\end{equation}
Thus by choosing our preconditioner appropriately, we can achieve $\alpha_B \in \Or(N_{e}\min(|t|,|U|))$.  The reasoning for computing Green's functions for the Coulomb interaction in the planewave dual basis and the Schwinger model follows identically.

As a final note, this block encoding technique is not just specific to Green's function evaluation.  Hamiltonian simulation and other related tasks can also be improved in the continuum limit by using this strategy.
\end{document}